\documentclass[final,5p,times,twocolumn]{elsarticle}

\usepackage[T1]{fontenc}
\usepackage{lmodern}
\usepackage{microtype}
\usepackage{amsmath,amssymb}
\usepackage{amsfonts}
\usepackage{bm}
\usepackage[colorlinks=true, allcolors=blue]{hyperref}
\usepackage{geometry}
\usepackage{amsthm}
\usepackage{chngcntr}
\usepackage{slashed}
\usepackage{dsfont}
\usepackage{algorithm}
\usepackage{algorithmic}
\usepackage{physics}

\newcommand{\1}{\mathds{1}}
\newtheorem{theorem}{Theorem}[section]
\newtheorem{corollary}{Corollary}[section]
\newtheorem{remark}{Remark}[section]
\newtheorem{lemma}{Lemma}[section]

\counterwithin{equation}{section}

\newcommand{\pd}{\partial}
\newcommand{\diag}{\operatorname{diag}}
\newcommand{\Ra}{\Rightarrow}
\newcommand{\8}{\infty}

\newcommand{\I}{\mathrm{i}}

\newcommand{\VEV}[1]{\left\langle #1 \right\rangle}
\newcommand{\ff}[1]{\frac{\delta}{\delta #1}}
\newcommand{\fbyf}[2]{\frac{\delta #1}{\delta #2}}

\newcommand{\R}{\mathbb{R}}
\newcommand{\C}{\mathbb{C}}

\newcommand{\bz}{\bar{z}}
\newcommand{\bpsi}{\bar{\psi}}
\newcommand{\Mathematica}{\textit{Mathematica\textsuperscript{\resizebox{!}{0.8ex}{\textregistered}}}}

\newcommand{\cL}{\mathcal{L}}
\newcommand{\D}{\mathcal{D}}

\def\pbyp#1#2{\frac{\partial#1}{\partial#2}}

\let\oldexp\exp
\renewcommand{\exp}[1]{\oldexp\left(#1\right)}
\def\INT#1#2{\int_{#1}^{#2}}
\def\norm#1{\lVert #1\rVert}

\allowdisplaybreaks[1]

\newenvironment{smalleq}{%
 \begingroup
 \small
 \addtolength{\jot}{.2em}%
 \setlength{\arraycolsep}{1.2pt}%
}{%
 \endgroup
}
\def\EQ#1{\begin{smalleq}\begin{align}#1\end{align}\end{smalleq}}
\def\br{\nonumber\\}

\def\lrb#1{\left(#1\right)}

\def\INT#1#2{\int_{#1}^{#2}}
\def\oh{\frac{1}{2}}
\def\Xint#1{\mathchoice
   {\XXint\displaystyle\textstyle{#1}}%
   {\XXint\textstyle\scriptstyle{#1}}%
   {\XXint\scriptstyle\scriptscriptstyle{#1}}%
   {\XXint\scriptscriptstyle\scriptscriptstyle{#1}}%
   \!\int}
\def\XXint#1#2#3{{\setbox0=\hbox{$#1{#2#3}{\int}$}
     \vcenter{\hbox{$#2#3$}}\kern-.5\wd0}}

\def\dashint{\Xint-}

\def\const{\textit{ const }}

\begin{document}

\begin{frontmatter}

\title{Geometric QCD II: \\
The Confining Twistor String and Meson Spectrum}

\author{Alexander Migdal}
\address{Institute for Advanced Study, Princeton, NJ 08540, USA}
\ead{amigdal@ias.edu}

\begin{abstract}
We present a local, asymptotically free solution of the planar MM loop equations in the continuum limit with full Lorentz invariance. The solution is constructed by quantizing internal Majorana fermions (referred to here as ``elves'') on a rigid Hodge-dual minimal surface. These worldsheet degrees of freedom provide the mechanism to satisfy the unintegrated vector loop equations, with the Pauli principle enforcing planar factorization. In the local limit, the theory reduces to a confining analytic twistor-string representation. By analyzing the monodromy structure of the complexified effective action, we show that the discrete mass spectrum is organized by topological data associated with twistor singularities. The simplest sector with one branch point yields parametric Regge trajectories expressed by a universal formula with trigonometric functions. These trajectories are approximately linear for light mesons over a broad range and are in agreement with experimental data for \textbf{36 light, strange, charmed, and bottom mesons}. The asymptotic behavior of the trajectory $J = \alpha(M^2)$ and its subleading corrections arise from the interaction between the Liouville term and the twistor monodromy, without introducing additional assumptions about string excitations. In our solution, the QCD mass spectrum follows from a generalized eigenvalue problem in complexified phase space, reducing the problem to \emph{classical geometry}. Within this framework, the large-$N_c$ Master Field is realized as a trajectory in twistor space.
\end{abstract}

\end{frontmatter}
\section{Introduction}

The identification of a string description of large-$N_c$ QCD has been a long-standing problem since the formulation of the MM loop equations\footnote{The \textbf{Makeenko--Migdal} equations, referred to hereafter as the MM equations, exhibit factorization properties characteristic of a string description. The foundational derivation is due jointly to the author and Yuri Makeenko.}. A continuum realization remains nontrivial. A \textbf{technical challenge} is the ultraviolet singular behavior of the coordinate-space Wilson loop functional $W[C]$. Even after renormalization, $W[C]$ exhibits cusp and perimeter divergences, and the loop equation contains a \textbf{contact term} $\delta^{(4)}(x-y)$ at self-intersections. In the coordinate representation, the loop equation is local, while perturbative solutions are nonlocal and develop logarithmic singularities at cusps and intersections. 

In this work, we adopt a different approach. Rather than attempting to renormalize the coordinate-space Wilson loop as an observable, we formulate the dynamics in \textbf{momentum loop space}. The momentum loop functional $W[P]$, introduced in \cite{SecQuanM95}, is the functional Fourier transform of $W[C]$. In this representation, the planar loop equation becomes an algebraic--differential functional equation without coordinate-space delta functions. The geometric singularities of $W[C]$ are integrated out, and the continuum equations admit a local formulation appropriate for quark-loop amplitudes. 

The aim of this second paper in the \textit{Geometric QCD} series is to provide the dynamical construction complementing the geometric framework developed in Part I \cite{migdal2025geometric}. In that work, the Hodge-dual additive minimal surface $S_\chi[C]$ was shown to generate a confining factor $\exp{-\kappa S[C]}$ as a zero mode of the loop operator. Here, we introduce worldsheet degrees of freedom on this rigid surface and construct a functional $W[C]$ that satisfies the planar loop equations. Our framework combines four elements:

\begin{enumerate}
\item \textbf{Fermi string (``Elfin theory'').}  
Internal Majorana degrees of freedom on the worldsheet implement planar topology. The Pauli principle cancels nonplanar intersection configurations, leaving planar contributions. 
\item \textbf{Rigid Hodge-dual minimal surface.}  
The fermions propagate on a surface determined by the boundary loop. There is no integration over worldsheet metrics, and the geometry is fixed by the loop data. 
\item \textbf{Momentum loop dynamics.}  
In momentum loop space, the loop equation takes an algebraic--differential form. A representation in terms of the Cuntz algebra exists \cite{SecQuanM95}, but a one-dimensional Master Field construction exhibits a rank deficiency at order $\mathcal{W}^{(8)}$ in the Voiculescu expansion, indicating the need for additional degrees of freedom. 
\item \textbf{Twistor parametrization.}  
The reduced loop-space measure is parametrized by boundary twistor variables. This yields a twistor parametrization in which a boundary sigma model is coupled to a Liouville field determined by analytic continuation from the boundary.
 
\end{enumerate}

\begin{remark}{\textbf{Equation solving vs.\ model building.}}
The starting point of this work is the planar MM loop equations, treated as defining equations of large-$N_c$ Yang--Mills theory. The goal is to identify variables and a functional measure in which these equations admit a finite continuum realization. The ingredients of the construction are constrained by this requirement. The Hodge-dual minimal surface appears because its area satisfies the loop equation, while the Majorana determinant satisfies planar factorization, reproducing planar graphs of QCD. 
\end{remark}

\begin{remark}{\textbf{Role of vector loop equations and the extrinsic surface geometry.}}
A central technical point underlying this construction is that the full, unintegrated \textit{vector} form of the MM loop equation contains essential information that is not captured by scalar or integrated formulations. In particular, when the momentum loop functional $W[P]$ is expanded in a one-dimensional Taylor--Magnus series, the resulting algebraic system becomes overconstrained at order $\mathcal{W}^{(8)}$. This is not a limitation of the expansion technique but a genuine inconsistency: the number of independent tensor constraints exceeds the number of available scalar degrees of freedom. 

This obstruction indicates that a one-dimensional loop functional does not possess sufficient internal structure to represent a solution of the vector loop equations. The missing tensorial degrees of freedom are associated with the extrinsic geometry of the embedding surface. In the present work, these are provided by the boundary twistor variables and the fermionic determinant on the Hodge-dual minimal surface. From this perspective, the construction
 provides the geometric structure required for a consistent local solution of the planar loop equations in the continuum limit. 
\end{remark}

\section{Zero modes of Loop equation}
\label{sec:zeromodes}
The theory of Hodge-dual minimal surfaces and its application to the MM loop equation is summarized here. For detailed discussion and proofs of the relevant theorems, see \cite{migdal2025geometric}; further details on the loop calculus without divergences are provided in \cite{migdal2025SQYMflow,ReviewPaperAM}.

\subsection{Self-dual area derivative and loop equation}
As suggested in \cite{migdal2025geometric}, we examine the self-dual minimal surface, which we will later treat as a background geometry for the QCD string with attached Fermi degrees of freedom. The area derivative for a functional $W[C]$ is defined as the discontinuity of the second variation of the area:
\EQ{
\label{areaderDef}
&\fbyf{W}{\sigma_{\mu\nu}(\theta)} = \frac{\delta^2W}{\delta \dot{C}_\mu(\theta-0)\delta \dot{C}_\nu(\theta+0)}- \mu \leftrightarrow \nu
}
This implements the geometric definition of \cite{MMEq79, Mig83} and is formally equivalent to the functional derivative implementation in \cite{PolyakovRychkov}. A significant feature of this definition is the absence of divergences and its parametric invariance. In previous implementations, the zero modes were often complicated by kinematic divergences.

The loop equation for the Yang--Mills gradient flow reduces to the loop space diffusion equation \cite{migdal2025SQYMflow}, admitting the solution $W = \exp{-\kappa S[C]}$ provided $S[C]$ is a zero mode:
\EQ{
&\mathcal L_\nu(S) = 0;\\
& \mathcal L_\nu = \left(\ff{\dot{C}_\mu(\theta+0)}-\ff{\dot{C}_\mu(\theta-0)}\right)\ff{\sigma_{\mu\nu}(\theta)}
\label{loopEq}
}
There exists a Hodge-dual solution analogous to the Yang--Mills multi-instanton. The loop space diffusion equation \eqref{loopEq} is satisfied for any functional $S[C]$ with a self-dual (SD) or anti-dual (ASD) area derivative:
\EQ{
&\fbyf{S_\pm}{\sigma_{\mu\nu}(\theta)} = \pm\ast\fbyf{S_\pm}{\sigma_{\mu\nu}(\theta)};\\
&e_{\alpha\mu\nu\lambda}\left(\ff{\dot{C}_\alpha(\theta+0)}-\ff{\dot{C}_\alpha(\theta-0)}\right)\fbyf{S_\pm}{\sigma_{\mu\nu}(\theta)} \equiv 0;
\label{Bianchi}
}
The Bianchi identity \eqref{Bianchi} holds as a Jacobi identity for a triple commutator \cite{migdal2025SQYMflow}. The exponential of the Hodge-dual minimal area serves as a factor in a general solution to the loop equation. This zero mode relies on the Leibniz property of the loop operator:
\smalleq{
    \begin{align}
       & \mathcal L_\nu f(\Phi[C]) = f'(\Phi[C]) \mathcal L_\nu \Phi[C];\\
       & \mathcal L_\nu (A[C] B[C]) = \mathcal L_\nu( A[C]) B[C] +  A[C] \mathcal L_\nu (B[C])
    \end{align}
}
This property leads to:
\smalleq{
    \begin{align}
        &\mathcal L_\nu \lrb{\exp{- \kappa S[C]} W_0[C]}=\br
        &   -\kappa \mathcal L_\nu(S[C])\exp{- \kappa S[C]} W_0[C]\br
        & + \exp{- \kappa S[C]} \mathcal L_\nu(W_0[C]) = \exp{- \kappa S[C]} \mathcal L_\nu(W_0[C])
    \end{align}
}
We identify a self-dual minimal surface that is additive at self-intersections. This additivity ensures the confining factor $\exp{-\kappa S[C]}$ is compatible with the chain of loop equations.

\subsection{Compatibility with the MM loop equations}
The chain of loop equations for Wilson loops at finite $N_c$ \cite{MM1981NPB, Mig83} utilizes the same loop operator \eqref{loopEq} as the Yang--Mills gradient flow. We consider the leading order of the $1/N_c$ expansion using the loop calculus developed in \cite{migdal2025SQYMflow}. The MM loop equation is structured as follows:

\begin{figure}[htbp]
    \centering
    \includegraphics[width=0.9\linewidth]{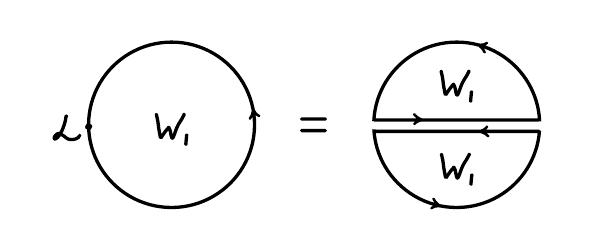}
    \caption{The MM equation, with the loop diffusion operator on the left side and the integral term on the right side. The double line represents the delta function $\delta^4(x-y)$ at self-intersections.}
    \label{fig:MMEquation}
\end{figure}

\EQ{
    & W[C] = 1/N_c \VEV{\tr \hat P \exp{\oint _C d x_\mu  A_\mu(x)}} ;\\
    \label{MMeq}
    &\mathcal L_\nu(W[C]) = \lambda \int_C d y_\nu\delta^4(x-y) W[C_{xy}]W[C_{yx}]
}
where $\lambda = N_c g_0^2$. This equation is understood in the sense of distributions. The confining factor $\exp{- \kappa S[C]}$ factors out by additivity and cancels on both sides of the equation independently of the splitting points $x,y \in C$.
\section{Hodge-dual minimal surface}
\label{sec:HDsurface}
We now define the functional $S[C]$ that satisfies the requirements of both duality and additivity. We follow the definitions and derivations from \cite{migdal2025geometric}, adapting them for the present purpose of a base for the fermions living on that minimal surface.

\subsection{The external and internal geometry}
We define the surface coordinates $X^i_\mu(\xi)$ (for $\xi= (\xi_1, \xi_2),i=1,2,3$. The field $X^i_\mu(\xi)$ maps the disk $\Sigma$ to $\R^3\otimes \R^4$.

\begin{figure}[htbp]
    \centering
    \includegraphics[width=1.\linewidth]{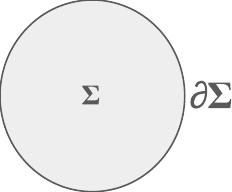}
    \caption{The unit disk $\Sigma$ is mapped to $\R^3\otimes \R^4$ and then projected to $\C^4$ by a holomorphic map, minimizing the area functional.}
    \label{fig:QCDSigma}
\end{figure}

Our surface area is defined as 
\begin{smalleq}     \begin{align}
\label{areaDef}
    &S_\chi[C] = \min_{X}\int_{\Sigma} d^2 \xi\sqrt{\Sigma_{\mu\nu}^2/2};
\end{align} \end{smalleq}
 The area element is defined as
\begin{smalleq}     \begin{align}
\label{sigmaDef}
    \Sigma_{\mu\nu} = \epsilon_{a b}\partial_a X^i_\mu\partial_b X^i_\nu;
\end{align} \end{smalleq}
 The Hodge chirality $\chi = \pm 1$ enters through the boundary conditions
 \begin{smalleq}
    \begin{align}
    \label{BCX}
        &X^i_\mu(\partial \Sigma) = \eta^{\chi,i}_{\mu\nu} C_\nu ;
\end{align}
\end{smalleq}
Here $\eta^{\chi,i}_{\mu\nu}$ are 't Hooft's matrices corresponding to Hodge duality $\chi = \pm 1$
\begin{smalleq}
    \begin{align}
      \eta^{\chi,i}_{\mu\nu} =  \delta_{i\mu} \delta_{\nu 4} - \delta_{i\nu}\delta_{\mu 4}  + \chi e_{i\mu\nu 4} ;
\end{align}
\end{smalleq}
The coordinate functions $X_\mu^a(\xi)$ satisfy the general Euler-Lagrange equations for the area density $\mathcal{L} = \sqrt{\Sigma^2/2}$:
\begin{smalleq}
    \begin{align}
   \left[ \text{E-L Equations} \right]_\mu^a =  \epsilon_{lm} \partial_l\left(  \frac{\Sigma_{\mu\nu} \partial_m X^a_\nu}{\sqrt{\Sigma^2}}\right)=0.
\end{align}
\end{smalleq}
Assuming the bulk variation vanishes by the Euler-Lagrange equations, we use Stokes theorem and get the boundary term
\begin{smalleq}
    \begin{align}
    \label{deltaSBC}
        \delta S = 2 \int d \theta   \frac{\delta X^i_\mu\Sigma_{\mu\nu}\pd_{\theta}X^{i}_\nu}{\sqrt{\Sigma^2}}
    \end{align}
\end{smalleq}
Substituting the boundary conditions \eqref{BCX} for $\delta X, \pd_{\theta}X$ and summing over $a = 1,2,3$ we find
\EQ{
    & \delta S = -2 \int d \theta  \frac{\delta C_\mu \Sigma_{\mu\nu}\pd_{\theta}C_\nu}{\sqrt{\Sigma^2}};\\
    \label{arderHodge}
    & \fbyf{S}{\sigma_{\mu\nu}} = - 2 T_{\mu\nu};\\
    & T_{\mu\nu} = \frac{\Sigma_{\mu\nu}}{\sqrt{\Sigma^2}}
}
Note the change of sign, compared with the ordinary area derivative of the minimal surface in four Euclidean dimensions. Regardless of this sign change, the Hodge duality of the area element $\Sigma$ with these boundary conditions for a generic field $X^a_\mu(z,\bar z)$ leads to the same Hodge duality of the area derivative. This self-duality of the area derivative makes our Hodge-dual minimal area a zero mode of the loop diffusion operator, leading to the solution of the MM loop equations. The area functional is invariant under the reparametrizations of the boundary loop $\delta_{\text{param}} S = 0$ by the skew symmetry of $\Sigma_{\mu\nu}$.

\subsection{The Holomorphic Map}
We resolve the duality constraint using the following holomorphic Ansatz:
\begin{smalleq}
    \begin{align}
    \label{HDsurface}
        &X^i_\mu = \eta^{\chi,i}_{\mu\nu} Y_\nu\\
        & Y_\mu = f_\mu(z) + \bar{f}_\mu(\bar{z});\\
        & z = \xi_1 + \I \xi_2;
\end{align}
\end{smalleq}
Geometrically, this is a projection from $\R^3\otimes \R^4$ to $\C^4$. The boundary conditions for the Hodge-dual surface become a conventional boundary condition for the Dirichlet problem
\begin{smalleq}     \begin{align}
\label{BCf}
    2\Re f_\mu(e^{i \theta}) = C_\mu(\theta)
 \end{align} \end{smalleq}
Our surface area \eqref{areaDef} reduces to the induced metric in this space
\begin{smalleq}  
\begin{align}
\label{SigmaChi}
        &\Sigma_{\mu\nu}= 2  (F_{\mu\nu} + \chi \ast F_{\mu\nu});\\
        & F_{\mu\nu} = \I(f'_\mu \bar{f}'_\nu - \bar{f}'_\mu f'_\nu);\\
    &\sqrt{\Sigma_{\mu\nu}^2/2}  = 2\sqrt{2} \sqrt{-\det g};\\
    & g_{ab}= \pd_{z_a} Y_\mu \pd_{z_b}  Y_\mu; \quad z_a,z_b = (z, \bar z)
\end{align}
\end{smalleq}
This area element is Hodge-dual $\ast \Sigma_{\mu\nu} = \chi\Sigma_{\mu\nu}$ for arbitrary functions $f_\mu(z)$. This duality condition provides the self-duality of the area derivative of the minimal area.

\subsection{The Virasoro Constraint and Uniformization}
We further impose a Virasoro constraint on the holomorphic maps derivatives
\begin{smalleq}     \begin{align}
\label{Virasoro}
    f'_\mu(z)^2 =0
 \end{align} \end{smalleq}
The imposition of the null constraint $(f')^2 = 0$ is a  gauge fixing procedure rooted in the geometric symmetries of the problem. This allows us to work in isothermal (conformal) coordinates, where the induced metric is diagonal: $g_{zz} = (f'_\mu)^2 = 0$. The existence of such a coordinate system is guaranteed by the Uniformization Theorem. In this gauge, the non-linear geometric area density simplifies to the quadratic Lagrangian:
\begin{smalleq}
    \begin{align}
    \label{Dirichlet}
       & S_\chi[C] = 2 \sqrt{2} \int_D |f'(z)|^2 d^2 z ;\\
       \label{areader}
       & \fbyf{S_\chi[C]}{\sigma_{\mu\nu}} = -2 \frac{\Sigma_{\mu\nu}}{\sqrt{\Sigma_{\mu\nu}^2}}
    \end{align}
\end{smalleq}
The solution of the linear boundary problem \eqref{BCf} for a holomorphic vector function is given by the Hilbert transform:
\begin{smalleq}
    \begin{align}
       & f_\mu = \oh(1 + \I\mathcal H) C_\mu;\\
       & H[u](\theta)\;=\;\frac{1}{2\pi}\,\mathrm{P.V.}\!\int_{0}^{2\pi} d\theta'\;
\cot\!\Big(\frac{\theta-\theta'}{2}\Big)\,u(\theta')  \label{HilbertTrans}\\
       & \int_D |f'(z)|^2 d^2 z \;=\;\int_{0}^{2\pi} d\theta\; C'_\mu(\theta)\,H[C_\mu](\theta);
\end{align}
\end{smalleq}
In terms of the Taylor coefficients:
\begin{smalleq}
    \begin{align}
    \label{HilbertFourier}
    & f_\mu(z) = \sum_{n>0} z^n C_{\mu,n};\\
     & C_\mu(\theta) = \sum_{n=-\infty}^\infty e^{\I n \theta} C_{\mu,n}
    \end{align}
\end{smalleq}
The internal geometry of the Hodge-dual surface is identical for the Hodge-dual surface and the Goldschmidt solution for the Plateau problem in 4D. This property is central for establishing the MM equation for the fermionic determinant.

The four-dimensional solution for the minimal surface uses twistors (see \ref{GenTheory}). This is the solution that we use later to build the Regge trajectories.

\subsection{The Physical Vacuum, Master field and Multi-Instanton Resummation}

Before formulating the dynamics of the Fermi string on the minimal surface, we clarify the origin of the rigid geometric background. The MM loop equations serve as quantum equations of motion for Planar QCD. As with any equations of motion, they admit a family of solutions distinguished by the choice of vacuum. The freedom to multiply the perturbative loop functional by a zero mode reflects this vacuum ambiguity.

In the traditional gauge-field picture, the nonperturbative QCD vacuum is a complex mixed state, often modeled as a distribution in multi-instanton moduli space. Because any individual multi-instanton configuration breaks translation invariance, restoring Poincaré symmetry requires integration over the full moduli space.

In the loop-space formulation, this ensemble is represented geometrically. Just as instantons are characterized by the self-duality of the gauge field strength, the corresponding zero mode is characterized by the self-duality of the loop-space area derivative. The instanton moduli are mapped to the moduli of the twistors parametrizing the four-dimensional minimal surface.

The physical meaning of this geometric representation is analogous to the Master Field proposed by Witten \cite{Witten:1980ez}, which describes the large-$N_c$ limit in terms of a single classical configuration in an infinite-dimensional space. In the present formulation, the Hodge-dual minimal surface plays this role. It provides a classical, translation-invariant object encoding the nonperturbative vacuum, without requiring an explicit construction in terms of gauge-field configurations. In this sense, the minimal surface represents the loop-space realization of the Master Field, with its twistor moduli corresponding to the degrees of freedom associated with the multi-instanton ensemble.

Moreover, the whole surface $X^i_\mu(\xi) \in \mathbb{R}^3 \otimes \mathbb{R}^4$ has the same algebraic structure as a gauge field in an $SU(2)$ subalgebra. The observables in our theory are the loop functionals and their area derivatives, related to the left and right Gauss maps of the surface. The twistor construction of the minimal surface is thus analogous to a multi-instanton configuration.

Rotations in the internal $\mathbb{R}^3$ act as $SO(3)$ transformations on $X^i_\mu$, corresponding to $SU(2)$ gauge transformations, and physical observables depend only on $SO(3)$-invariant combinations.

The minimal surface described by these twistors is non-fluctuating, reflecting the classical nature of the Master Field in the planar limit.

Thus, the complex mixed state of multi-instantons in coordinate space is mapped, in loop space, to a single translation-invariant geometric object: a Hodge-dual minimal surface. In this representation, the detailed distribution of topological charge and the explicit multi-instanton moduli are not required. The loop equations can be solved directly, without specifying the underlying measure in coordinate space.

The resulting geometric object defines a translation-invariant nonperturbative vacuum for the worldsheet fermions. The degrees of freedom associated with instanton moduli are encoded in the moduli of the minimal surface. In four dimensions, these moduli admit a natural description in terms of holomorphic twistor data.

\section{The Elfin theory on a flat surface}
\label{Elfintheory}

The elfin theory is defined here independently of the gauge theory; the equivalence is established in subsequent sections.

Let us consider the flat surface bounded by some contour $C$. 
Later, we generalize this to the Hodge-dual minimal surface. All the functionals on the surface will depend on the conformal factor in the mass term
\begin{smalleq}\begin{align} \label{2.5}
e = \pd_+ X^i_\mu \pd_- X^i_\mu.
\end{align}\end{smalleq}
On the minimal surface, we have a conformal map \eqref{HDsurface}, reducing the conformal metric to $e =2 \sqrt{2} |f'_\mu|^2$. This reduction is not needed for the general theory of fermions on the minimal surface, but it will be used in the forthcoming sections when we establish a relation between the fermionic determinants and the MM loop equation.

\begin{remark}[Holes, poles, and global signs]
Multiply connected worldsheets (holes/handles) may require multivalued primitives $f_\mu(z)=\int^z f'_\mu(\zeta)\,d\zeta$ (logarithms/periods), even when the tangents $f'_\mu$ are single-valued meromorphic.
In the applications to the QCD spectrum, we allow such meromorphic $f'_\mu$ with poles in the unit disk; the induced conformal factor
$e(z,\bar z)\propto |f'_\mu(z)|^2$ is single-valued but may diverge at these poles.

This does not invalidate the analysis of the Elfin determinant in the flat gauge: closed fermion paths are weighted by
$\exp{-m\int_\Gamma \sqrt{e}\,|dz|}=\exp{-m\,\ell_\Gamma}$, hence any path approaching a pole (where $\sqrt{e}\sim |z-w|^{-1}$ for a simple pole of $f'$) has $\ell_\Gamma\to\infty$ and is exponentially suppressed.
Since $\ell_\Gamma$ is a local additive functional, the length weights of the locally paired configurations (touching vs.\ intersecting) coincide, while their mod-$2$ intersection signs are opposite, yielding the same cancelation mechanism.

Finally, while handles do not affect the local turning-angle contribution to the loop sign, they can affect global signs of the fermion propagator/determinant (spin-structure/holonomy) and complicate the topological analysis of the corresponding minimal-surface embeddings.
\end{remark}

Let us now return to the flat surface and introduce the Elfin field as a bispinor
\begin{smalleq}\begin{align} \label{2.10}
\psi = \psi^\alpha_\lambda(\xi), \quad \alpha = 1,2, \quad \lambda = \pm 1.
\end{align}\end{smalleq}

There are various equivalent forms of the action. The simplest form reads
\begin{smalleq}\begin{align} \label{2.11}
A_0 = \int \dd^2\xi (\bpsi_\lambda  \sigma^k \pd_k \psi_\lambda +  \bpsi_\lambda \psi_\lambda m \sqrt{e}),
\end{align}\end{smalleq}
where $\sigma_1, \sigma_2$ are the Pauli matrices. We shall also use the complex components
\begin{smalleq}
    \begin{align}
        \sigma_\pm &= \oh(\sigma_1 \pm \I \sigma_2);\\
        \pd_\pm &= \pd_1 \pm \I \pd_2;\\
        \label{2.15}
\psi^\pm_\lambda &= \frac{1}{2} (1 \pm \sigma_3)\psi_\lambda . \\
\label{2.16}
\bpsi^\pm_\lambda &= \psi^\dagger_\lambda \frac{1}{2} (1 \pm \sigma_3). 
    \end{align}
\end{smalleq}This action is covariant with respect to conformal transformations 
\begin{smalleq}\begin{align} 
\label{2.12}
e(\xi_+, \xi_-) &\to f'_+(\xi_+) f'_-(\xi_-) e(f_+, f_-),\\
\label{2.13}
\psi^\pm_\lambda(\xi) &\to \sqrt{f'_\pm}\psi^\pm_\lambda(f),\\
\label{2.14}
\bpsi^\pm_\lambda(\xi) &\to \sqrt{ f'_\mp} \bpsi_\lambda(f);
\end{align}\end{smalleq}
At the boundary $\dot\Sigma$, we require that the spin $\sigma_3$ takes definite values, namely:
\begin{smalleq}\begin{align} 
\label{2.17}
\sigma_3 \psi_\lambda = \lambda \psi_\lambda, \quad \text{at }\partial \Sigma,\\
\label{2.18}
\bpsi_\lambda \sigma_3 = \lambda \bpsi_\lambda, \quad \text{at } \partial \Sigma.
\end{align}\end{smalleq}
In other words, at the boundary
\begin{smalleq}\begin{align} 
\label{2.19}
\psi^-_{+} = \psi^+_{-} = \bpsi^-_{+} = \bpsi^+_{-} = 0.
\end{align}\end{smalleq}
Such "chiral bag" boundary conditions were studied in the mathematical literature \cite{Gilkey:1995} in the context of the Atiyah-Singer Index Theorem. For our purposes, it is sufficient to note that these boundary conditions are consistent with the topology of the Dirac path integral below.

The remaining half of the components is not restricted at the boundary. In the Schr\"odinger picture in Minkowski space this implies that the states with $\sigma_3 = \pm 1$ at both ends of the string are occupied. Such peculiar boundary conditions are necessary to obtain the closed loop equations (see below). Note that these boundary conditions are conformally invariant and do not depend on the form of the boundary.

At the "time" reflection
\begin{smalleq}\begin{align} \label{2.20}
\xi_2 \to -\xi_2 \quad \text{or} \quad \xi_+ \leftrightarrow \xi_-
\end{align}\end{smalleq}
the spin $\sigma_3$ changes the sign. According to our boundary conditions, this should be accompanied by a change of sign of $\lambda$:
\begin{smalleq}\begin{align} \label{2.21}
\psi_\lambda \to \psi_{-\lambda} \sigma_3.
\end{align}\end{smalleq}
One may introduce the oriented states
\begin{smalleq}\begin{align} 
\label{2.22}
\psi^{R,L}_\lambda &=  \psi^{\pm \lambda}_\lambda,\\
 \label{2.23}
\bpsi^{R,L}_\lambda &= \bpsi^{\mp \lambda}_\lambda .
\end{align}\end{smalleq}
The left-hand states are occupied at the boundary. As we shall see, the orientation $\lambda\sigma_3$ is related to the orientation of the loop $C$. The right-hand and left-hand states interchange when the loop is reoriented.

The action can be rewritten in the Left-Right (chiral) form. Based on the transformation properties in equations \eqref{2.22} and \eqref{2.23}, and using the identity $\sigma^k \partial_k = \sigma_+ \partial_+ + \sigma_- \partial_-$, the action becomes:
\begin{smalleq}\begin{align}
A_0 &= \int \dd^2\xi \left[  \bar{\psi}^R \partial_{+} \psi^L + \bar{\psi}^L \partial_{-} \psi^R  + \right. \br
&\left. m\sqrt{e} \left( \bar{\psi}^R \psi^R + \bar{\psi}^L \psi^L \right) \right].
\end{align}\end{smalleq}
In this form, the conformal properties are manifest. Under the transformations $\xi_+ \to f_+(\xi_+)$ and $\xi_- \to f_-(\xi_-)$  the kinetic term transforms as:
\begin{smalleq}\begin{align}
&\left(\sqrt{f'_-}\bar{\psi}^R\right) \left(f'_+\partial_{+}\right) \left(\sqrt{f'_-}\psi^L\right) \br
&= f'_-f'_+ \left(\bar{\psi}^R \partial_{+} \psi^L\right);\\
&\left(\sqrt{f'_+}\bar{\psi}^L\right) \left(f'_-\partial_{-}\right) \left(\sqrt{f'_+}\psi^R\right) \br
&= f'_-f'_+ \left(\bar{\psi}^L \partial_{-} \psi^R\right).
\end{align}\end{smalleq}
The mass term transforms as:
\begin{smalleq}\begin{align}
&\left(\sqrt{f'_+ f'_-}\sqrt{e}\right) \left(\sqrt{f'_-}\bar{\psi}^R\right) \left(\sqrt{f'_+}\psi^R\right) \br
&= (f'_+ f'_-) \sqrt{e} \bar{\psi}^R \psi^R;\br
&\left(\sqrt{f'_+ f'_-}\sqrt{e}\right) \left(\sqrt{f'_+}\bar{\psi}^L\right) \left(\sqrt{f'_-}\psi^L\right) \br
&= (f'_+ f'_-) \sqrt{e} \bar{\psi}^L \psi^L.
\end{align}\end{smalleq}
This factor of $(f'_+ f'_-)$ exactly cancels the Jacobian from the measure $\dd^2\xi$, making the kinetic energy as well as the mass term conformally invariant.

Note that the Lagrangian density is invariant with respect to internal $U(1) \times SU(2)$ transformations:
\begin{smalleq}\begin{align} \label{2.24}
\psi_\lambda &\to S_{\lambda \lambda'} \psi_{\lambda'} e^{\I\alpha},\\
 \label{2.25}
\bpsi_\lambda &\to \bpsi_{\lambda'} (S^{-1})_{\lambda'\lambda} e^{-\I\alpha}.
\end{align}\end{smalleq}

However, the boundary conditions lower this symmetry down to $U_+(1) \times U_-(1)$:
\begin{smalleq}\begin{align} 
\label{2.26}
\psi^\pm &\to \psi^\pm e^{\I\alpha_\pm},\\
\label{2.27}
\bpsi^\pm &\to \bpsi^\pm e^{-\I\alpha_\pm}.
\end{align}\end{smalleq}
Both $U(1)$ currents vanish at the boundary:
\begin{smalleq}\begin{align} \label{2.28}
j_k = \bpsi^\pm \sigma_k \psi^\pm = 0 \quad \text{at } \dot\Sigma; \quad k = 1, 2.
\end{align}\end{smalleq}
So, there is no flow outside. The derivative $\pd_k$ in (2.11) can be replaced by $D_k$, since the boundary terms vanish.

At vanishing mass $m$ the theory possesses chiral $\sigma_3$ invariance, which is not violated by any anomaly.

Let us now discuss the symmetry of the physical $x$-space.
Transformations of the internal coordinates $\xi$ are independent of physical rotations in $x$-space.
From the viewpoint of two-dimensional field theory in $\xi$-space, the rotations in $x$-space represent internal $O(4)$ symmetry transformations of the $C_\mu$ field,
\begin{smalleq}\begin{align} \label{2.29}
\delta C_\mu = \omega_{\mu\nu} C_\nu.
\end{align}\end{smalleq}
The Elf does not transform:
\begin{smalleq}\begin{align} \label{2.30}
\delta\psi_\lambda = \delta\bpsi_\lambda = 0.
\end{align}\end{smalleq}
This follows from the boundary conditions. The full target surface field $X^i_\mu$ transforms as one of the $SO(3)_\pm$ subgroups of $SO(4)$ 
\begin{smalleq}
    \begin{align}
        X^i_\mu \Ra \Omega^{\pm}_{i j}X^j_\nu W_{\nu\mu}; \quad \Omega^{\pm} \in SO(3), \quad W \in SO(4)
    \end{align}
\end{smalleq}

Consider now local variations of the loop $C$. 
The area derivative of $e$ can be found from the known variation of the scalar area
\begin{smalleq}\begin{align} \label{2.40}
&|S| = \int_S \dd^2\xi e  = \int_S d\sigma_{\mu\nu}(\xi_1) d\sigma_{\mu\nu}(\xi_2) \delta_{inv}(\xi_1-\xi_2);\\
&\frac{\delta |S|}{\delta\sigma_{\mu\nu}(\xi_1)} = \int \dd^2\xi_2 \frac{\delta e}{\delta\sigma_{\mu\nu}(\xi_2)} \br
&= 2\int_S t_{\mu\nu}(\xi_2) d^2\xi_2 \delta(\xi_1-\xi_2) 
\end{align}\end{smalleq}
Using the  area derivatives of the Hodge-dual minimal area established in the previous section, we find for the area derivative of
\begin{smalleq}
    \begin{align}
       \frac{\delta e(\xi)}{\delta\sigma_{\mu\nu}(\xi_0)} &= -2 T_{\mu\nu} \delta^2(\xi - \xi_0);\\
       \label{Tdef}
       T &= \frac{F + \chi\ast F}{\sqrt{2 F^2}}\\
       \label{detArder}
       \fbyf{\VEV{\exp{A_0}}}{\sigma_{\mu\nu}} &= -2 m \VEV{ \frac{T_{\mu\nu}\bpsi \psi}{\sqrt{e}}}_{\partial \Sigma}
    \end{align}
\end{smalleq}
where $F$ is given by \eqref{areader} .
According to our boundary conditions, only the right-hand components are present in \eqref{detArder}:
\begin{smalleq}\begin{align} \label{2.43}
\bpsi \psi = \bpsi^R \psi^R \quad \text{at } \dot\Sigma.
\end{align}\end{smalleq}
\section{Conformal metric as a local gauge parameter}
\label{sec:conformalmetric}
We begin by defining the effective action for the elf (Majorana fermion) field on a two-dimensional surface $S$ with metric 
\[
g_{ab} = \exp{2\rho}\delta_{ab},
\]
including the mass term and the dependence on the conformal factor $\rho$. We are reproducing here the construction in Eqs.~(2.57)--(2.63) of~\cite{M81QCDFST}.

The action of the general theory is given by
\EQ{
&A_\rho = \int d^2 \xi \br
&\left( e^{2\rho} \bar{\psi_\lambda} \sigma^k e^{-\rho} \nabla_k \psi_\lambda 
+ m 
\sqrt{e} e^\rho \bar{\psi_\lambda} \psi_\lambda - \frac{\partial_+\rho\partial_-\rho}{3\pi}\right);
}
with the spinor connection $\omega_k$:
\EQ{
\nabla_k = \partial_k + \frac{\I}{2} \omega_k \sigma^3,\quad \omega_k = \oh e_{k l} \pd_l \rho
}
vierbein
\EQ{
E_k^a = \delta^a_k e^{\rho},
}
and curvature
\EQ{
\hat R_{k l} &= [\nabla_k,\nabla_l]= \I \sigma_3 e_{k l} R \det E;\\
R &= -8 e^{-2\rho} \partial_+ \partial_- \rho.
}
At $\rho=0$ this is the old action $A_0$, whereas at $\rho = \ln \sqrt{e}$ this is the action of the Dirac particle with constant mass at the surface $S$ with the induced metric $g_{ab}$.

\begin{theorem}[Conformal Invariance of the Elf Determinant]\label{TheoremConfInv}
The functional integral
\EQ{
Z[S|e, \rho] = \int \D \bpsi \D \psi \exp{A_\rho}
}
is invariant under local variations of the conformal metric field $\rho$.
That is, the theory depends only on the conformal class (moduli) of the surface, and the induced metric $e(\xi)$, remaining independent on the local gauge parameter $\rho(\xi)$.
\end{theorem}

\begin{proof}[Proof]
The main point here is the covariant regularization of the Dirac theory at the surface. We have to define the divergent trace in the vacuum amplitude
\EQ{
&\int \mathcal{D} \bar{\psi} \mathcal{D} \psi \exp{A_\rho} =
\exp{ - \int d^2 \xi \frac{(\partial\rho)^2}{3\pi} \right.\br
&\left.+\tr_+ \log \lrb{\I \hat{D_\rho} + m \sqrt{e} e^{-\rho}} + \tr \log\lrb{\I \hat{D_\rho} + m \sqrt{e} e^{-\rho}}}.
}
Here
\EQ{
\hat{D_\rho} = \sigma^k e^{-\rho} (-\I \pd_k \omega_k),
}
is the covariant Dirac operator, and the subscript $+$ conditions $\sigma^3 = +1$.

We employ the Pauli-Villars regularization in the following form:
\EQ{
&\tr \ln (\I \hat{D_\rho} + m \sqrt{e} e^{-\rho})_{\mathrm{reg}}\br
&= \tr \left[ \ln (\I \hat{D_\rho} + m \sqrt{e} e^{-\rho}) - \frac{1}{2} \sum_i c_i \ln (\hat{D_\rho}^2 + M_i^2) \right], \br
&\sum_i c_i = 1.
}
The shift of the $\rho$ field
\EQ{
\rho \to \rho + \delta \rho
}
results in the multiplicative transformation of the Dirac operator
\EQ{
\hat{D_\rho} \Ra e^{-\tfrac{3}{2}\delta \rho} \hat{D_\rho}e^{-\tfrac{1}{2}\delta \rho}.
}
Therefore, the variation of the regularized determinant reads
\EQ{
&\tr \left[ - \delta \rho + 2 \sum_i c_i \left( \delta \rho \frac{M_i^2}{M_i^2 + \hat{D_\rho}^2} \right) \right] \br
&= - \tr \left( \delta \rho \sum_i c_i \frac{M_i^2}{M_i^2 + \hat{D_\rho}^2} \right).
}
The WKB calculations of \cite{M81QCDFST} (reproduced here in the Appendix) yield
\EQ{
\int d^2 \xi e^{2\rho} \delta \rho \frac{R}{24 \pi} = \delta\int d^2 \xi \frac{\partial_+ \rho \partial_- \rho}{6 \pi}.
}
Therefore, the variations of terms in the exponential cancel each other, as claimed.
\end{proof}
\begin{corollary}
    As a consequence of this theorem, the theory of massive Majorana fermions on a minimal surface with induced metric $g_{a b} = \delta_{a b} e(\xi)$ , and extra conformal term $-\frac{(\partial\rho)^2}{3\pi}$ in Lagrangian is equivalent to the theory of Majorana fermions on a flat surface with variable mass $m \sqrt{e(\xi)}$. The first theory corresponds to  taking $\rho = \log \sqrt{e}$ in $A_\rho$, and the second theory corresponds to taking $\rho =0$. By the proven gauge invariance these theories are equivalent.
\end{corollary}
\begin{remark}[Moduli and Topology]
The theorem establishes local independence from $\rho(\xi)$. For a disk, all metrics are conformally equivalent to the flat metric, so the determinant is a a unique functional of external loop via the conformal metric $e = |f'_\mu|^2$. For surfaces with handles or holes, the metric can be written as 
\[
g_{ab} = \exp{2\rho}\,\hat{g}_{ab}(\tau) ,
\]
where $\tau$ are the Teichmüller parameters (moduli). The determinant then depends on $\tau$, i.e., $Z_{\mathrm{elf}}[C] = Z[C|\tau]$. In the present geometric theory, the minimal surface is rigid: the boundary loop $C$ uniquely determines the solution to the Plateau problem and thus fixes the moduli $\tau = \tau[C]$. Therefore, the Dirac determinant is a well-defined functional of $C$, encoding the "gluon cloud" effects for the specific minimal surface spanned by $C$. In the degeneration limit (e.g., infinitesimal bridges), the moduli approach a boundary value, but the locality of the loop equation is preserved.
\end{remark}

\section{The fermion determinants and the planar topology}
\label{fermiomndeterminant}
Let us consider the vacuum functional $Z$ of the Elfin theory with the boundary conditions discussed in the previous sections. By virtue of conformal invariance, we can choose the local gauge $\rho(\xi) = \log \sqrt{e(\xi)}$ or $rho(\xi) =0$
The following analysis of the vacuum loops in this theory simplifies in the gauge $\rho=0$, where it is similar to the topological solution of the 2D Ising model.

As a functional of the loop $C$, the vacuum functional satisfies a certain non-linear path integral equation. This equation will be derived in this section and investigated in the following sections.
\subsection{Mathematical justification of the signed sum over loops}
The path integral representation of the Dirac determinant on a finite surface with boundary conditions, as employed in this work, is  justified by several mathematical results and earlier physical constructions. In particular, Ichinose~\cite{Ichinose:1984} and Gaveau~\cite{Gaveau:1984} provide theorems establishing the equivalence between the determinant of the Dirac operator and path integrals over spinor-valued paths with appropriate boundary conditions on finite domains. These constructions carefully incorporate the spinor connection and holonomies of spin rotations along geodesic segments, ensuring the correct treatment of spinor phases and curvature effects.

Our representation of the determinant as a sum over oriented closed loops weighted by factors of the form $\prod \exp{-m\, d l + \frac{\I \sigma_3 \delta \theta}{2}}$, where \(d l\) is an infinitesimal geodesic step and \(\delta \theta\) the corresponding trajectory rotation angle, follows directly from these  path integral formulations. This approach generalizes the combinatorial and topological methods pioneered by Vdovichenko~\cite{Vdovichenko:1965} in the exact solution of the 2D Ising model, where the fermionic nature and spinor holonomies enforce cancelations of non-planar intersections, leaving only planar loop configurations that contribute to the determinant.

Furthermore, the Gauss-Bonnet theorem relates the total tangent rotation angle around a loop to its self-intersection number, providing a topological invariant that governs the sign factors in the loop expansion. This ensures that the fermionic statistics and boundary conditions are correctly encoded in the loop amplitudes.

Thus, the path integral representation employed here is not only physically intuitive but also mathematically sound, grounded in well-established results in spin geometry, spectral theory of Dirac operators, and  path integral constructions on manifolds with boundary. This justifies its use as the foundation for deriving loop equations and analyzing the planar topology of the fermion determinants in the Elfin theory.

\subsection{Random loops at the surface}
The form of the surface $S$ is fixed as a Hodge-dual minimal surface in a conformal metric. Let us perform the functional integration over the Elfin fields $\psi$. We use the gauge $\rho =0$, where the induced metric is simply $e = |f'|^2$, the internal curvature$R_{a b}$ and spinor connection $\omega_k$ are absent, making the theory locally equivalent to a spinor on a flat surface.

This produces two determinants of the Dirac operators with the boundary conditions $\lambda\sigma_3 = \pm 1$:
\begin{smalleq}\begin{align} 
\label{3.1}
I[S] &= \int \D\psi e^{A_0} = \Delta_+ \Delta_-;\\
\label{3.2}
\Delta_\pm &= \det(D + m\sqrt{e})_\pm
\end{align}\end{smalleq}
On an infinite flat surface, such a determinant would be calculable using the trace of the logarithm of a resolvent in Fourier representation. On a finite surface $S$ bounded by an arbitrary curve $C$, there are nontrivial boundary conditions for a Dirac resolvent, complicating the analytic computation. 

Instead, we can represent the Dirac determinant as a sum over vacuum loops at the surface. The similar representation in the context of the 2D Ising model was discussed in my earlier work \cite{Migdal:1980zp}, inspired by studies of the Ising model by Vdovichenko \cite{Vdovichenko:1965}. That pioneering paper laid a framework for a topological solution of the Ising model in terms of self-avoiding paths. 
The next step of our paper \cite{Migdal:1980zp} was to generalize this method to a continuum Dirac operator on a Riemann surface using the Gauss-Bonnet theorem to relate the rotation angle to the intersection index.
There are also some theorems in the mathematical literature \cite{Ichinose:1984, Gaveau:1984} establishing the Dirac and Heat kernel path integral representations in a finite domain. In this paper, we follow the technology described in \cite{M81QCDFST}, with proper adjustments from a random surface to a Hodge-dual minimal surface.

We base our study on the superposition principle. The vacuum amplitude for the Dirac particle represents the sum over all configurations of the loop $\Gamma \in S$ of the corresponding amplitude for the given loop. The amplitude of the given loop represents the product of the amplitudes of individual events. The individual events can be regarded as a sequence of classical propagations along little geodesic steps alternating with instant rotations.

Technically, we use the logarithm of a Dirac determinant and represent it as a trace of a resolvent integrated over proper time
\begin{smalleq}
    \begin{align}
       &\Delta_\pm = \br
       &\exp{\kappa|S[C]] - \int_\epsilon^\infty\frac{ \dd T}{T}\exp{- T ( D + m\sqrt{e})}};\\
       & \kappa \propto \log \epsilon;
    \end{align}
\end{smalleq}
The logarithmically divergent term $\log \epsilon$ in the exponential is proportional to the area of the surface by virtue of the locality of this divergent contribution to the trace. The terms in the exponential, proportional to the area of the additive Hodge-dual surface, pass through the loop equation, as we already established. The role of these terms is to regularize the effective string tension (or vacuum energy density) of the Fermi model on the surface.

The path integral representation emerges when this exponential is replaced by a limit of a product (see \cite{PolyakovGFS})
\begin{smalleq}
    \begin{align}
        \exp{- T ( D + m\sqrt{e})} &\to \prod \exp{- d T (D + m \sqrt{e}};\\
       d T &= d l/\sqrt{e}
    \end{align}
\end{smalleq}
The mathematical literature replaces this product with a stochastic process that is equivalent to Feynman's sum over paths, with the path amplitude being a product of WKB amplitudes for these propagations/rotations \cite{Gaveau:1984}. 

The amplitude for the rotation is given by the rotation matrix
\begin{smalleq}\begin{align} \label{3.5}
A(\Delta\theta) = \exp{\I \sigma_3 \Delta\theta/2}.
\end{align}\end{smalleq}
The WKB amplitude for the geodesic step reads
\begin{smalleq}\begin{align} \label{3.6}
A(step) = \exp{-m \dd l }
\end{align}\end{smalleq}
\subsection{The Gauss-Bonnet theorem and self-intersections}
The total product of the factors (3.5) and (3.6) along the loop can be calculated by means of the Gauss-Bonnet formula which in our notation reads
\begin{smalleq}\begin{align} \label{3.7}
\sum \Delta\theta = 2\pi(1 - \nu)  \mod 2 \pi.
\end{align}\end{smalleq}
Here $\nu$ is the algebraic number of self-intersections of the loop $\Gamma$ at the surface. The Gauss-Bonnet theorem ensures that the relation is topologically invariant.
\begin{remark}{\textbf{Mathematical theorems}.}
   For an oriented regular simple closed curve $\gamma\subset\mathbb{R}^2$, the rotation (turning) index satisfies
$\operatorname{rot}(\gamma)=\frac{1}{2\pi}\int_\gamma \kappa\,ds\in\mathbb{Z}$ and equals $+1$ for the positively oriented boundary of a disk \cite{Whitney1937RegularClosedCurves}. 
In the planar case, this is precisely the Gauss--Bonnet theorem for the enclosed region (since $K\equiv 0$), yielding $\int_\gamma \kappa\,ds=2\pi\,\chi=2\pi$ (hence equality modulo $2\pi$ is automatic) \cite{Chern1944GaussBonnet}. 
Moreover, since $\operatorname{rot}(\gamma)$ is the degree of the unit-tangent map $T:S^1\to S^1$, it depends only on the immersed curve itself and therefore is unchanged if the ambient plane is replaced by a disk with finitely many holes (and the same mod-$2$ intersection sign $(-1)^{\nu}$ follows from Whitney's regular-homotopy classification of immersed plane curves) \cite{Whitney1937RegularClosedCurves}.

\end{remark}
The amplitude for the loop reads:
\begin{smalleq}\begin{align} \label{3.8}
A[\Gamma] = (-1)^\nu \exp{-m l_\Gamma}; \quad l_\Gamma = \int_\Gamma \sqrt{e} | d z|
\end{align}\end{smalleq}
The extra negative sign comes from Fermi statistics. This amplitude is a scalar, independent of the quantum numbers, so we should multiply it by the number of allowed states. Each two-component Dirac particle has only one allowed state which is projected by the operator
\begin{smalleq}\begin{align} \label{3.9}
P = \frac{1}{2} (1 + \sigma_a v_a)
\end{align}\end{smalleq}
where $v_a$ is the normalized 2-velocity of the rest frame. In the rest frame, there should be only one component (in 4 dimensions, there were two components).

So we are left with the factor of 2 from the sum over the orientations $\sigma_3 = \pm 1$. Yet it will be convenient to distinguish the terms with two orientations, i.e. to sum over oriented loops $\Gamma$. We mark the loop by an (anti)clockwise arrow for (left) right orientation $\sigma_3 = -1, +1$.

The point is that the loops that touch the external boundary at least once are oriented in the right direction according to our boundary conditions.

The reflection amplitude
\begin{smalleq}\begin{align} \label{3.10}
R = \frac{1}{2} (1 + \lambda\sigma_3)
\end{align}\end{smalleq}
corresponds to the repulsion of the left-hand side components from the wall (see the previous section). As discussed above, one may occupy the wall with the left-hand side states so that this repulsion would come about automatically due to the Pauli principle.

So, the sum over loops involves the sum over orientations only for the inner loops.
\subsection{From exponential of free elfin loops to sum over multiple crossing paths}
Consider a given oriented loop consisting of $N$ equal geodesic steps $\dd l$. The degrees of freedom are given by the angles $\Delta\theta$ and by the number $N$ of steps. The phase volume can be written in a covariant form
\begin{smalleq}\begin{align} \label{3.11}
\D\Gamma = N^{-1} \prod \dd(\Delta\theta) = \dd l Dx(l) \delta(\dot{x}^2 - 1) e^{-m_0 l}
\end{align}\end{smalleq}
The $\delta$-functions remove the integrations over the lengths $|\dd x_i|$ of the geodesic steps, leaving us with the normalized 2-velocities as physical variables. The mass renormalization term $m_0$ is a matter of convention.

The factor $N^{-1} = 1/l$ in \eqref{3.11} accounts for $N$ equivalent choices for the origin of the parametrization of the loop. These factors will cancel the combinatorial factors that arise at the intersections (see below).

We write the integral over loops with the above prescription at the reflections as follows:
\begin{smalleq}\begin{align} \label{3.12}
\Delta &= \Delta_+ \Delta_- \br
&\propto \exp{\sum \dashint \D\Gamma (-1)^\nu \exp{-m l_\Gamma}}.
\end{align}\end{smalleq}
We ignore the (singular) normalization factor in the fermionic determinant. This factor on a finite surface amounts to an exponential of the area $\exp{ - \const{} |S|}$, which we include in the confining factor that is already present in our solution for the Wilson loop, according to the previous section.

The $\dashint$ denotes the lack of a sum over orientation for the loop touching the boundary due to our boundary conditions. Note that both determinants contain loops that touch the boundary. These loops enter with a weight of 1/2 in each determinant $\Delta_\pm$, or, equivalently, they count as oriented loops, unlike the remaining loops, which involve a sum of two orientations.

In what follows, it will be convenient to consider the product $\Delta$ rather than the separate factors $\Delta_\pm$.

Let us expand \eqref{3.12} in a series of multiple integrals over loops:
\begin{smalleq}\begin{align} \label{3.13}
\Delta &= \sum \frac{1}{n!} \br
&\dashint D\Gamma_1 ... \dashint D\Gamma_n (-1)^{\sum\nu_i} \exp{-m \sum l_i}.
\end{align}\end{smalleq}
These loops are independent by construction, but we are going to rearrange the sum in such a way that they will start to repel. 
This rearrangement, conjectured in \cite{Migdal:1980zp} relies on the following topological and combinatorial observations.
\subsection{Cancellation of intersections: step towards planar graphs}
This section reproduces (with some simplifications and improvements) the most important results of our `81 paper: the topological and geometrical analysis of the path integrals of the Dirac particle on a disk-like surface.

The arguments of this section were inspired by a geometric solution 
of the 2D Ising model by Natasha Vdovichenko \cite{Vdovichenko:1965}. The crucial observation made in that pioneering work was that the sum of planar loops splitting  the plane into domains of up and down spin in the Ising model is equivalent to the sum over freely intersecting independent loops, with phase factors reflecting the rotation of a spin 1/2 particle. The signs arising for intersecting paths lead to cancelation between such paths, leaving only planar loops. The computation in that paper relied on the geometry of the paths on a square lattice and other specifics of the Ising model. Below, the same topological observations will be stripped of the lattice artifacts and applied to a two-fermion system with oriented loops on a Hodge-dual minimal surface.

Consider the given configuration of loops at the surface as a graph consisting of several disconnected parts. Connected subgraphs contain intersecting and touching lines. The total length in exponential factors $\exp{- m l}$ can be redistributed among the lines of the graph:
\begin{smalleq}\begin{align} \label{3.14}
l_{tot} = \sum l_{loops} = \sum l_{lines}.
\end{align}\end{smalleq}
 The length of the line is not just a geometric length of this line drawn on a planar disk; rather, it is an integral
\EQ{
l_{line} = \int |d z(t)| \sqrt{e\left(z(t), \bar z(t)\right)}
}
However, the metric $e(z,\bar z)$ in this integral being \emph{local}, this length is additive across the pieces made by intersections of these lines on a planar domain.
The factorization also applies the sign factor since the number of mutual intersections is always even in the absence of handles at the surface.
\begin{smalleq}\begin{align} \label{3.15}
\nu_{tot} = \sum \nu_{loops} = \sum \nu_{graph} \pmod 2.
\end{align}\end{smalleq}
The phase volume of the connected subgraph $G$ can also be rearranged in such a way that the old loops lose their individuality:
\begin{smalleq}\begin{align} \label{3.16}
D G &= \prod_{line} \dd W(step),\\
\dd W(step) &= \dd l Dx(l) \delta(\dot{x}^2 - 1).
\end{align}\end{smalleq}
The factors $1/l$ in front of the old loops are compensated by the same factors that arise from the sums over the numbers of geodesic steps between the intersections of lines.

It is important that there are no other factors in front of the loop. Were there, say, $k$ internal degrees of freedom, the corresponding factor
\begin{smalleq}\begin{align} 
\label{3.18}
\prod_{loops} k
\end{align}\end{smalleq}
could not possibly be distributed among the lines of the graph. In the case of $k = 1$ one could return to the non-oriented lines, then the phase volume would be distributed again. This is the case of the Ising model, which was considered in the preliminary version \cite{Migdal:1980zp} of the Elfin theory. This possibility is ruled out by the boundary conditions. As we shall see soon, the peculiar boundary conditions play a central role in the construction of the loop equations.

Let us proceed with the sum over graphs. Each internal line of our graph can be oriented in both directions, but orientation is conserved in the vertices. We neglect multiple collisions as a violation of the Pauli Principle (Grassmann algebra $(\bpsi^R\psi^R)^n=0,(\bpsi^L\psi^L)^n=0$ for $n >1$, so we deal with the two pairs of vertices depicted at Fig. \ref{fig:Intersections}:
\begin{figure}
    \centering
    \includegraphics[width=0.5\linewidth]{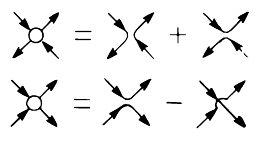}
    \caption{Types of line collision: two planar (upper drawing) and one planar, another non-planar (lower drawing) }
    \label{fig:Intersections}
\end{figure}
Naturally, one could treat these two vertices as a single vertex depending on the angles between the lines; however, in this case, there will be step functions of the angles. When the lines cross each other in coordinate space, the term in the vertex changes sign.

The non-planar vertex (Fig. \ref{fig:Intersections}) cancels the first one, removing this configuration from the sum. This is the central mechanism. The lines may touch each other with opposite orientations according to Fig. \ref{fig:Intersections}, but the terms with intersecting lines are always canceled by those with parallel touching lines in the graph with the same configuration of lines.

Note that in the general case, this is a cancelation of terms with a different number of loops, so there is no room for extra degrees of freedom. This is the Pauli principle reformulated in terms of the path integrals.

We may now return to our old loops but disregard all the non-planar configurations. We obtain a hierarchy of trapped loops that move in the free space left by the others. With one-component fermions, we would obtain the sum over all decompositions of the whole surface into two different phases, which is equivalent to an Ising model. 

The above representation of the Majorana determinant would amount to the Onsager solution of the Ising model in terms of free Majorana fermions. In our case, we have twice the number of components, providing for the oriented planar loops instead of just boundaries of the spin domains.

The time slice ($\xi_2 = \const{}$) of this picture describes a string with particles moving between neighbors. Pairs of $\psi$ and $\bpsi$ particles can be created or annihilated. Only $\psi\bpsi$ neighbors can touch each other. The touching of $\psi\psi$ or $\bpsi\bpsi$ neighbors is forbidden by the Pauli principle. One $\psi, \bpsi$ pair is placed at the two ends of the string. \footnote{This interpretation of the time slice dynamics can be used for a  proof of the equivalence of the free Dirac vacuum loops and the planar loops made of non-intersecting closed paths on a surface. Such a proof would use the Hamiltonian representation of a Fermionic system with conventional creation/annihilation operators and chiral boundary conditions.}

This resembles the 't Hooft string model for planar graphs \cite{tHooft:1974pau}. Our string has the same topology but takes into account the gauge invariance of the planar graph expansion. This will be demonstrated in the next sections.

\subsection{Loop equation for fermion determinants}

Let us now find the variation of $\Delta$ with respect to the boundary value of $e$.

The basic relation reads
\begin{smalleq}\begin{align} \label{3.20}
\frac{\delta  e^{-m l_\Gamma}  }{\delta e(x_0)} &= -m \int_\Gamma\frac{\dd l}{2e} \delta^2(\xi - \xi_0) e^{-m  l_\Gamma}, \br
\dd l &= |\dd\xi| \sqrt{e}
\end{align}\end{smalleq}
The invariant $\delta$-function
\begin{smalleq}\begin{align} \label{3.21}
e^{-1} \delta^2(\xi - \xi_0) = \delta_{inv}(x, x_0), \quad x, x_0 \in S,
\end{align}\end{smalleq}
removes one of the $\dd^2x_\parallel$ integrations in the sum over loops. The additional $\dd l$ integration that arises is required by our rules since a new line at the graph is created.

The variation selects the graph that touches the boundary at $x_0$. The touching line is oriented to the right according to our boundary conditions.

The amplitude of touching equals $-m$; the remaining factors contribute to the phase volume of the graph.

Let us follow the loop that touches the boundary. This is the correct loop; therefore, each loop that touches it from the outside is also correct. We observe that the outside domain $S_{out}$ between this loop $\Gamma$ and the external boundary $C$ is completely equivalent to the original surface $S$. The loops trapped in this domain reflect from $\Gamma$ in the same way as they do from the external boundary (this is why we need the special boundary conditions).

The sum over configurations and orientations of the loops in $S_{out}$ produces precisely the same Dirac determinant.
\begin{smalleq}\begin{align} \label{3.22}
\Delta_{out} = \Delta_+^{out} \Delta_-^{out}
\end{align}\end{smalleq}
The loops that are trapped inside $\Gamma$ reflect from inside, so they are oriented to the left at the boundary $\Gamma^{-1}$. The sum over configurations and orientations of these loops produces the Dirac determinant for the internal domain,
\begin{smalleq}\begin{align} \label{3.23}
\Delta_{in} = \Delta_+^{in} \Delta_-^{in},
\end{align}\end{smalleq}
with the opposite boundary conditions,
\begin{smalleq}\begin{align} \label{3.24}
\psi^R = \bpsi^R = 0, \quad \text{at } \Gamma^{-1}.
\end{align}\end{smalleq}
The numerical value of this determinant is the same as that for the old boundary conditions.

We see, however, that the orientation $\lambda\sigma_3$ is related to the orientation of the boundary in space. This will be important in higher order $1/N$ expansion of QCD where the multiloop propagators depend upon the relative orientation of loops.

Note that the loop $\Gamma$ may touch itself as well as the external loop $C$ (at the surface). In this case, one or both of the domains $S_{in}, S_{out}$ reduce to a set of windows touching at the corners. The corresponding functional should be understood as
\begin{smalleq}\begin{align} \label{3.25}
\Delta = \prod \Delta(windows).
\end{align}\end{smalleq}
This is clear from the original representation of $\Delta$ as an exponential of the sum over configurations of independent loops. This planar hierarchical set of loops is shown in Fig. \ref{fig:QCDWindowsSliced}, with the red line illustrating the time slice: the number of Elf pairs flying between quarks.
\begin{figure}
    \centering
    \includegraphics[width=1.\linewidth]{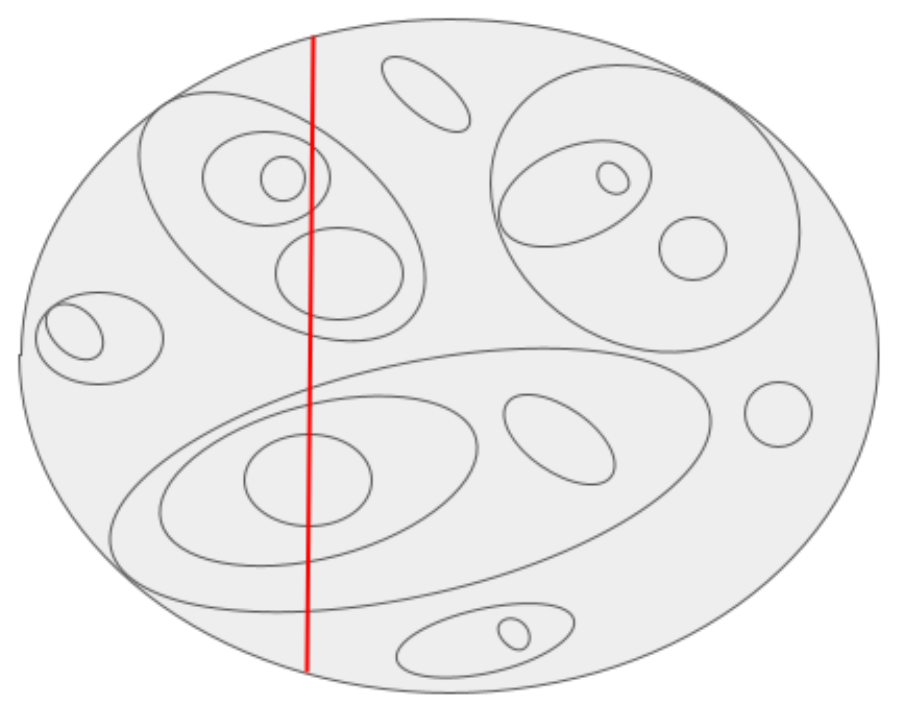}
    \caption{The hierarchical set of loops, some touching, but never intersecting each other nor self-intersecting. The time slice (red line) represents some number of Elf pairs.}
    \label{fig:QCDWindowsSliced}
\end{figure}

\section{The Minimal Surface and the Loop Equation}
\label{LoopEqOnMinSurf}
In the 1981 theory, the next step was to integrate over all embeddings $X_\mu(\xi)$. This led to the Liouville instability. In the present theory, we replace this integration with the \emph{Rigid Hodge-Dual Surface} $S_{min}[C]$ defined in \cite{migdal2025geometric} and described in the section \ref{sec:HDsurface} of this paper.

This surface is the unique additive solution to the minimal area problem with the constraint that the area derivative is self-dual. It is fixed by the boundary loop $C$. Therefore, we do not integrate over embeddings. The functional $Z[C]$ is simply the fermion determinant on this specific manifold:
\begin{smalleq}\begin{align}
Z[C] \equiv Z[S_{\chi}[C]].
\end{align}\end{smalleq}
For the time being, we fix the Hodge duality $\chi$ of this surface; later, we will discuss the restoration of parity symmetry in the loop equation.

\subsection{Loop equations}
We now demonstrate that this $Z[C]$ satisfies the loop equation of the same structure as \eqref{MMeq}.

Using the above planar loop equations for Dirac determinants and area derivatives of the HD minimal surface, we may write the following loop equation for the functional $Z[S]$:
\begin{smalleq}\begin{align} 
\label{3.28}
&\delta Z[S]/\delta\sigma_{\mu\nu}(x_1) = -m \br
&\int \D\Gamma_{11} \delta_{I,0} T_{\mu\nu}(1) Z[S_{in}] Z[S_{out}] e^{-m l_{11}}.
\end{align}\end{smalleq}
This equation is illustrated in Fig. \ref{fig:QCDAder}
\begin{figure}
    \centering
    \includegraphics[width=0.9\linewidth]{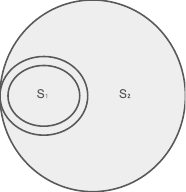}
    \caption{The first area derivative of Dirac determinant. There is one internal path $\Gamma$ which cut the surface $S$ into two pieces $S_{in} =S_1, S_{out} =S_2$}
    \label{fig:QCDAder}
\end{figure}
Here, the tensor $T$ was defined in \eqref{Tdef}. The loop $\Gamma_{11}$ touches $C_{11}$ at $x_1$. The integer index factor
\begin{smalleq}\begin{align} \label{3.29}
&I = \br
&\int_\Gamma d l_1 \int_\Gamma d l_2 \theta(l_1-l_2) e_{a b}\dot\xi_a(l_1) \dot \xi_b(l_2) \delta^2(\xi(l_1) - \xi(l_2))
\end{align}\end{smalleq}
is the absolute number of self-intersections of the loop $\Gamma_{11}$ at the surface. The Kronecker delta $\delta_{I,0} $ in the integral in \eqref{3.28} enforces the self-avoiding loop $\Gamma_{1 1}$. In the following, we denote path integrals over self-avoiding loops as $\int_{\asymp} \D \Gamma$

The measure of the open path is defined as follows:
\begin{smalleq}\begin{align} \label{3.30}
\int \D\Gamma_{11} = \int \dd l \int_{x(0)=x_1}^{x(l)=x_1} Dx(l) \delta(\dot{x}^2 - 1).
\end{align}\end{smalleq}
Eq. \eqref{3.28} was originally proposed in \cite{Migdal:1980zp}. At that time, it seemed that the equation applied to the Ising model at the surface, i.e., for the 2-component spinors. However, as we see now, one should introduce two spinors with special boundary conditions in order to obtain this equation for the propagator of the Fermi string.

Note that our theory is constructed in such a way that it is equivalent to the flat theory with a variable fermion mass. 
In the next section, we shall use the second area derivative of $Z[C]$, which reads
\begin{smalleq}\begin{align} 
\label{3.32}
&\frac{\delta^2 Z}{\delta\sigma_{\mu\nu}(l) \delta\sigma_{\alpha\beta}(r)} = A_{\mu\nu\alpha\beta} +  B_{\mu\nu\alpha\beta};\\
&A_{\mu\nu\alpha\beta} = (m/2)^2  T_{\mu\nu}(l) T_{\alpha\beta}(r)\iint_{\asymp} \D\Gamma_{l} \D\Gamma_{r} \br
&Z[S_{l}] Z[S_{mid}]Z[S_{r}]\exp{-m (l_{l}+l_{r})} ;\\
&B_{\mu\nu\alpha\beta} = (m/2)^2 T_{\mu\nu}(l) T_{\alpha\beta}(r) \iint_{\asymp} \D\Gamma_{up}\D\Gamma_{dn} \br
&\times Z[S_{up}] Z[S_{mid}] Z[S_{dn}] \exp{-m (l_{up}+l_{dn})}.
\end{align}\end{smalleq}
This equation is illustrated in Fig. \ref{fig:QCD2},\ref{fig:QCD3}.
\begin{figure}
    \centering
    \includegraphics[width=0.9\linewidth]{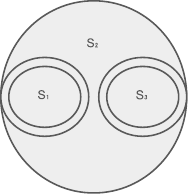}
    \caption{The first term in the second area derivative of Dirac determinant. There are two closed loops $\Gamma_{l}, \Gamma_{r}$ which cut the surface $S$ into three pieces $S_1=S_{l}, S_2=S_{mid},S_3=S_{r}$.}
    \label{fig:QCD2}
\end{figure}
\begin{figure}
    \centering
    \includegraphics[width=0.9\linewidth]{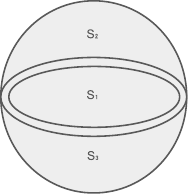}
    \caption{The second term in the second area derivative of Dirac determinant. There is one closed loop $\Gamma =\Gamma_{up}\Gamma_{dn}$, touching the loop $C$ at the left and the right points and cutting the surface $S$ into three pieces  $S_1= S_{mid},S_2=S_{up},S_3=S_{dn}$.}
    \label{fig:QCD3}
\end{figure}
The first term arose due to the variation of the last factor in \eqref{3.28}. The second term arose due to the variation of the factor $\exp{-m l}$ as has been discussed above. In this case, there are two paths $\Gamma_{up}, \Gamma_{dn}$ which cut the surface $S$ into three pieces $S_{up}, S_{mid}, S_{down}$ with the boundaries
\begin{smalleq}\begin{align} 
\label{3.33}
\partial S_{up} &= \Gamma_{up} C_{rl},\\
\label{3.34}
\partial S_{mid} &= \Gamma_{dn}^{-1} \Gamma_{up}^{-1},\\
\label{3.35}
\partial S_{dn} &= \Gamma_{dn} C_{lr}.
\end{align}\end{smalleq}
The variation of the factor $Z[S_{in}]$ in \eqref{3.28} would require one of the internal loops inside $S_{in}$ to touch the external boundary $C$. This implies that $\Gamma$ touches $C$ at the same point, so this is the triple collision that we rule out as a violation of the Pauli Principle (and Grassmann algebra) $(\bpsi^R\psi^R)^3=(\bpsi^L\psi^L)^3=0$.
As for the remaining factors in \eqref{3.28} they do not depend upon $e$.

The operator $\cL_\nu$ involves the gradient of the first area derivative. Using the translation invariance identity \cite{Mig83}
\begin{smalleq}
    \begin{align}
    \label{TransIdent}
        \pd_\mu F[C_{xx}] = -\int_{C_{xx}} \dd y_\alpha \frac{\delta F[C_{xx}]}{\delta\sigma_{\alpha\mu}(y)}
    \end{align}
\end{smalleq}
we can rewrite $\cL_\nu(Z[C])$  as a similar double integral with the second area derivative 
\begin{smalleq}\begin{align}
\cL_\nu(l) Z[C] = - \int_{C_{ll}} \dd x^\alpha_r \frac{\delta^2 Z[C]}{\delta\sigma_{\mu\nu}(l)\delta\sigma_{\alpha\mu}(r)}
\end{align}\end{smalleq}

Substituting the loop equation \eqref{3.32} for $Z[C]$ here, we find two terms corresponding to the two drawings in Fig. \ref{fig:QCD2} and \ref{fig:QCD3}. 

\subsection{Large fermion mass and induced QCD}
We assume that the fermion mass $m$ is much larger than the physical scale $\Lambda_{QCD}$, along the lines of the induced QCD scenario \cite{Kazakov:1992}.
As we shall see, in this limit, the mass of the fermion $m$ serves as a UV cutoff, and the $A$-term reduces to the same loop operator applied to the area of the Hodge-dual minimal surface. This term vanishes by the Bianchi identity, as we have established in our previous work and summarized in the beginning of this paper.
The second term reduces, in this limit, to the right side of the MM equation multiplied by a coefficient determined by the local properties of the Dirac determinant on a flat minimal surface in the vicinity of the self-intersection of the loop.

These conclusions will be based on a well-known property of geodesic paths on Riemann surfaces, which we formulate as a Lemma.
\begin{lemma}[Geodesic Inequality]
For any minimal surface embedded in Euclidean space, the geodesic distance $L_{geo}(x,y)$ is  bounded by the Euclidean distance:
\begin{smalleq}\begin{align}
L_{geo}(x,y) \ge |x-y|_{\R^4}.
\end{align}\end{smalleq}
\end{lemma}
\begin{proof}
    The geodesic distance between two points $x,y$ along the surface embedded in Euclidean space is a constrained minimum of the length of paths connecting these points in Euclidean space. The constraint is that  every point of this path belongs to the surface. The set of such paths is a subset of all Euclidean paths, and the line element $d l = |d \xi |\sqrt{e}$ in the induced metric $e = |f'|^2$ coincides with the Euclidean line element $d s = \sqrt{|d f|^2}$; therefore, the minimal length of the subset path is greater or equal to the target space distance $\int d s = |x-y|_{\R^4}$.
\end{proof}
\begin{remark}
    Note that we used the induced metric $e = |f'|^2$, the same as for the Euclidean minimal surface (Goldschmidt solution to the Plateau problem), rather than the full induced metric on a general surface in $\R^3\otimes \R^4$.  The metrics coincide on a holomorphic minimizer, but for the area derivatives, we need the full external geometry in the 12 dimensional space $\R^3\otimes \R^4$. Once the area derivatives are taken (as they are in our loop equations, resulting in Hodge-dual tensors $T_{\mu\nu}$) we could use a holomorphic minimizer for the induced metric in path integrals.
\end{remark}

We are going to use this theorem for the terms arising in the loop equation for the Dirac determinants.

\subsection{The $A$-term}
The first term (see Fig. \ref{fig:QCD2}) in the limit of large fermion mass (much larger than the external curvature $Q_{a b} = \diag\left\{-\rho,\rho\right\}$ of the minimal surface at its boundary $C$) effectively involves infinitesimal flat loops $\Gamma_{l}, \Gamma_{r}$. The factor $Z[S_{mid}]$ reduces to the constant $Z[S[C]]$, which we take out of the path integrals $\iint_\asymp \D \Gamma_l \D \Gamma_r$. These integrals can be computed on an infinite flat semiplane, and they reduce to constants independent of the points $x_l, x_r$. We find
\EQ{
&A_{\mu\nu\alpha\beta} \to  T_{\mu\nu}(l) T_{\alpha\beta}(r) Z[C] a^2;\\
& a = m/2\int_{\asymp} \D\Gamma Z[\Gamma] \exp{-m l[\Gamma]} ;
}
Now, we recall that the area derivative of the minimal surface  is also proportional to the tangent tensor $T_{\mu\nu}$, and let us factor out of $Z[C]$ the vacuum energy factor (with $S[C]$ being the Hodge-dual minimal area)
\EQ{
Z[C] &= G[C] Z_1[C];\\
 G[C] &= \exp{- \kappa S[C]}
}
We can rewrite the factor $T T$ as a second area derivative
\EQ{
&A_{\mu\nu\alpha\beta} \to \frac{-a^2 Z_1[C]}{4\kappa^2 }\frac{\delta^2 G[C]}{\delta\sigma_{\mu\nu}(l)\delta\sigma_{\alpha\beta}(r)}  ;
}
which is valid for arbitrary $\kappa$ and  $x(l) \neq x(r)$. 
Now, we use the translational identity \eqref{TransIdent} backwards
\EQ{
& \int_{C_{ll}} \dd x^\alpha_r A_{\mu\nu\alpha\mu}=\br
& \frac{-a^2Z_1[C]}{4 \kappa^2 } \int_{C_{ll}} \dd x^\alpha_r \frac{\delta^2 G[C]}{\delta\sigma_{\mu\nu}(l)\delta\sigma_{\alpha\mu}(r)}=\br
& \frac{a^2 Z_1[C]}{4 \kappa^2 } \pd_\mu \fbyf{G[C]}{\sigma_{\mu\nu}(l)} 
}
The last integral, as we already established, vanishes for the Hodge dual minimal surface by the Bianchi identity.
Thus, the (unwanted) $A-$ term in the loop equation vanishes in the local limit.

\subsection{The $B-$ term}
The $B-$ term in the Dirac determinant loop equation has a completely different structure, and it is singular in the local limit.
\EQ{
&B_{\mu\nu\alpha\beta} = (m/2)^2 T_{\mu\nu}(l) T_{\alpha\beta}(r) \iint_{\asymp} \D\Gamma_{up}\D\Gamma_{dn} \br
&\times Z[S_{up}] Z[S_{mid}] Z[S_{dn}] \exp{-m (l_{up}+l_{dn})}.
}
The $B-$term in the loop equation 
\EQ{
&\int_{C_{ll}} \dd x^\alpha_r B_{\mu\nu\alpha\mu}
}
involves the matrix product of two $T$ tensors
\EQ{
 TT_{\alpha\nu} = T_{\alpha\mu}(r)T_{\mu\nu}(l)
}
which can be reduced further using the algebra of the $\eta$ tensors (valid for both Hodge-chiralities $\chi = \pm 1$)
\EQ{
\eta^{i,\chi}_{\alpha\mu}\eta^{j,\chi}_{\mu\nu} = - \delta_{i j} \delta_{\alpha\nu} - e_{i j k} \eta^{k,\chi}_{\alpha\nu}
}
Both tensors $T$, being Hodge-dual, can be represented as
\EQ{
T_{\mu\nu}(x) = \oh \vec  \tau(x)\cdot \vec \eta^{\chi}_{\mu\nu}
}
with some unit vector $\vec \tau(x)^2=1$, depending on the point $x \in C$. This tangent tensor was defined as a limit of the normalized area element $\Sigma_{\mu\nu}$ on the surface when a point $x\in S[C]$ approaches the boundary $C$.

Afterward, the matrix product becomes
\begin{smalleq}\begin{align}
 TT_{\alpha\nu} = \frac{-1}{4}\lrb{(\vec \tau(r)\cdot \vec\tau(l)) \delta_{\alpha\nu} + (\vec \tau(r)\times \vec\tau(l))\cdot \vec \eta_{\alpha\nu}}
\end{align}\end{smalleq}
Now, we explore the consequences of the geodesic inequality. In the local limit, the points $l, r$ belong to the self-intersection of the loop $C$.  The Goldschmidt minimal surface (additive in the limit when $C = C_{lr}\cdot C_{rl}$) breaks into two surfaces covering closed loops connected at the self-intersection point, as shown in Fig \ref{fig:QCDSelfIntersect}.
\begin{figure}
    \centering
    \includegraphics[width=0.9\linewidth]{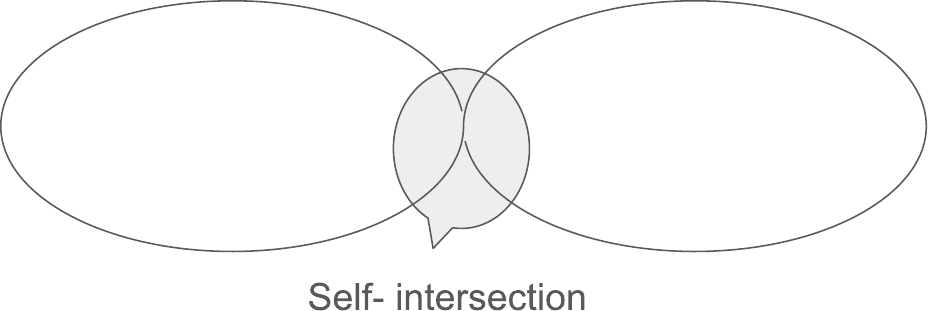}
    \caption{Self-intersecting loop.}
    \label{fig:QCDSelfIntersect}
\end{figure}
The zoom into the vicinity of the self-intersection is shown in three dimensions in Fig. \ref{fig:WhitesBridge}.
\begin{figure}
    \centering
    \includegraphics[width=0.9\linewidth]{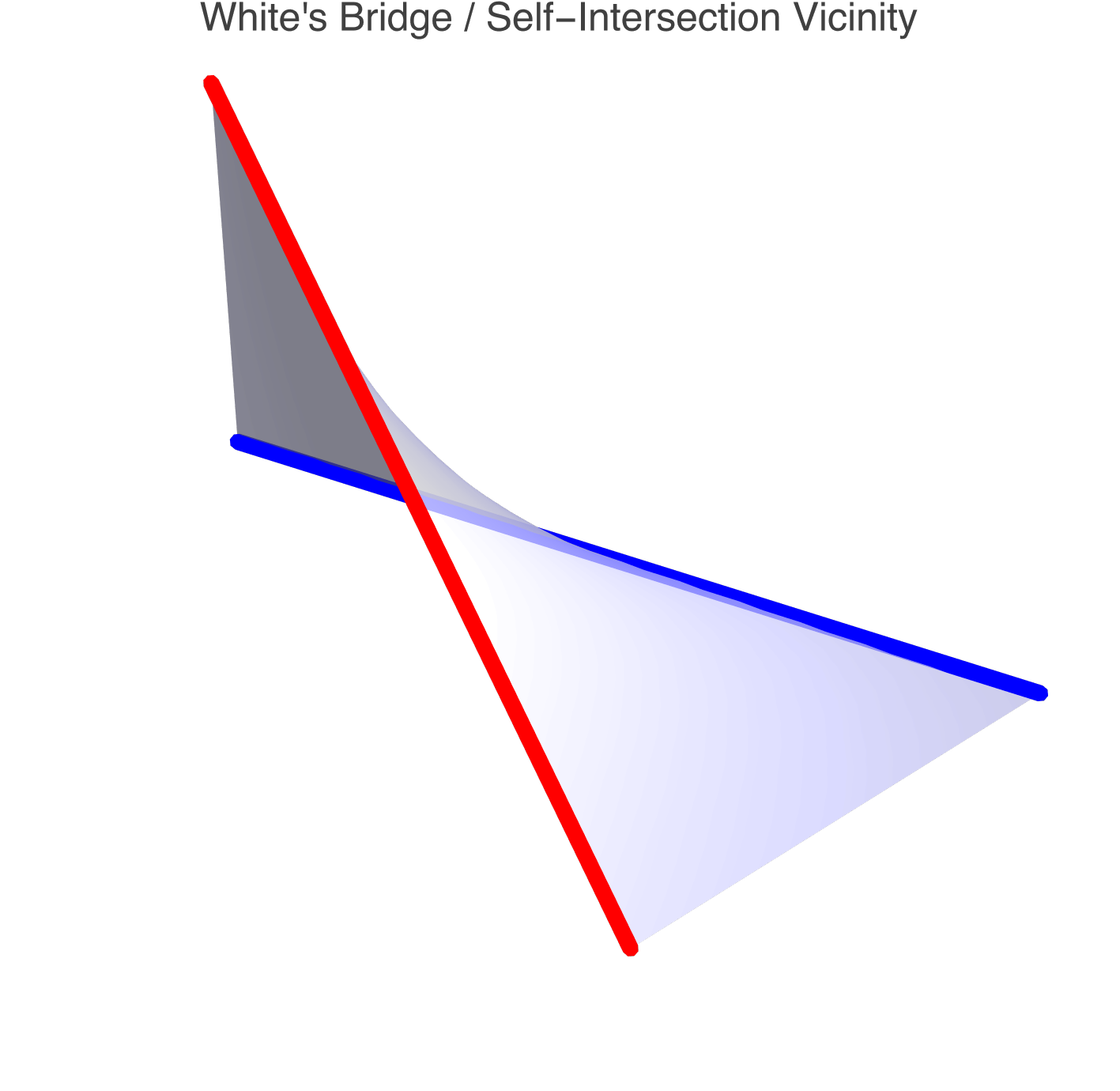}
    \caption{The additive minimal surface in the vicinity of the self-intersection. There is a narrow twisted strip (White's bridge) connecting the minimal surfaces bounded by two parts of the self-intersecting loop.}
    \label{fig:WhitesBridge}
\end{figure}
The minimal surface is, in general, curved, though the arithmetic mean curvature vanishes everywhere. In the vicinity of the self-intersection, for our Goldschmidt solution, there is an infinitesimal narrow bridge connecting two parts bounded by loops $C_{l r}, C_{r l}$. This is proven in White's bridge theorem \cite{White1996}. 
This property places the tensors $T(l), T(r)$ in the infinitesimal vicinity of each other on this bridge. In the local limit, we have $\vec \tau(l) \to \vec \tau(r)$, so that the matrix product
\begin{smalleq}\begin{align}
T_{\alpha\mu}(r)T_{\mu\nu}(l) \to \frac{-1}{4} \delta_{\alpha\nu}
\end{align}\end{smalleq}
As for the rest of the factors, the sum over the short paths $\Gamma_{up},
\Gamma_{dn}$ is dominated by the geodesic distance, which is bounded by the
Euclidean distance. This precise alignment relies on the isotropic
scaling of White's bridge in the varifold limit $\epsilon \to 0$. In the
strict limit of the topology change, the continuity of the area Hessian
demands that the tangent planes  align, yielding the required
Fierz-like identity for the MM loop equation without generating anomalous
contact terms. A rigorous geometric measure-theory proof of the Hessian's
continuity across this topology change is left for future mathematical study,
but is physically mandated by the $O(4)$ symmetry of the local loop
diffusion. We find:
\EQ{
&B_{\mu\nu\alpha\mu} \to -b\delta_{\alpha\nu} Z[C_{lr}]Z[C_{rl}] \delta^4(C(l) - C(r)); \\
& b\delta^4_{m}(C(l) - C(r))= \br
& \frac{m^2}{16} \iint_{\asymp} \D\Gamma_{l}\D\Gamma_{r} Z[S_{mid}] \exp{-m (l_{l}+l_{r})}.
}
\begin{figure}
    \centering
    \includegraphics[width=0.9\linewidth]{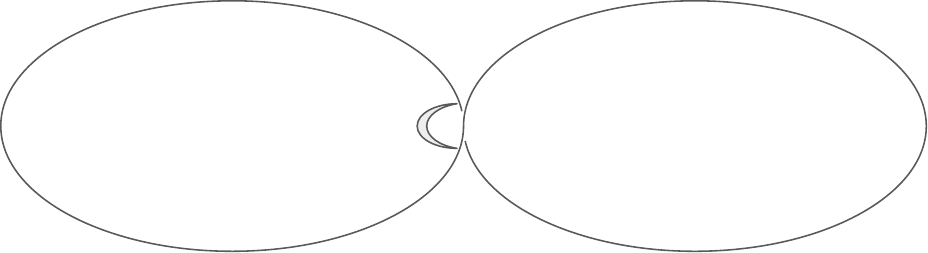}
    \caption{Two arks $\Gamma_{l},\Gamma_{r} $ connecting self-intersection points, bounding the crescent-shaped area $S_{mid}$.}
    \label{fig:QCDDeltaTerm}
\end{figure}
This path integral over small surface paths $\Gamma_{l},\Gamma_{r} $ is illustrated in fig.\ref{fig:QCDDeltaTerm}, where these paths are two arcs bounding the crescent-shaped $S_{mid}$.
We use an approximation of the delta-function
\EQ{
\delta^4_{m}(x - y) =4 m^4\exp{- 2 m |x-y|}
}
Putting the pieces together, we find the loop equation \eqref{MMeq} with
\EQ{
 &W[C] =  \frac{Z[C]}{Z[\1]};\\
 &\lambda \delta^4_{m}(C(l) - C(r))= \frac{m^2 Z[\1]^2}{16} \br
 &\iint_{\asymp} \D\Gamma_{up}\D\Gamma_{dn}  \exp{-m (l_{up}+l_{dn})} W[C_{up}\cdot C_{dn}].
}

The chirality of the Hodge dual minimal surface dropped from this equation, as it dropped from the whole internal geometry, including the minimal area and the Dirac determinant. Therefore, there is no necessity to symmetrize our solution over Hodge chirality to restore the parity of QCD: the parity is preserved by the Dirac determinant on a Hodge-dual minimal surface.

The sign is positive, as it should be for QCD with a positive 't Hooft coupling constant $\lambda$.  As for the value of this coupling constant, it is determined by normalization $Z[\1]$. The string tension normalization factor $G[C] = \exp{- \kappa S[C]}$ cancels on both sides of the loop equation, so it does not influence $\lambda$, but the normalization factor $Z[\1]$ at vanishing loop multiplies this constant twice.

The asymptotically free QCD we are looking for corresponds to the limit when this factor goes to zero. In that limit, the standard RG computations based on planar graphs would produce the QCD mass scale
\EQ{
\label{QCDMass}
&\Lambda_{QCD} =m \lambda^{-\tfrac{51}{121}} \exp{-\frac{24 \pi^2}{11\lambda}};
}
Therefore, all we need is the limit of the Elfin theory such that this factor vanishes
\EQ{
\lambda \propto Z[\1]^2 \to 0;
}
The rest of perturbative QCD will follow from the loop equation by iterations in $\lambda$ starting with $W[\1]=1$.
Let us now outline this perturbative QCD derivation.

\section{Asymptotically free planar QCD from the loop equation}
\label{sec:asymptoticfreedom}
The Elfin theory induces QCD in the limit of large fermion mass, serving as a short-distance cutoff. There are no singularities at finite mass, which resolves the ambiguities associated with cusp and perimeter singularities of the Wilson loop. The Wilson loop is defined in the Elfin theory as a regularized solution of the loop equation, with the delta function smeared by a finite fermion mass. This naturally regularizes the loop equation by smearing the delta function on the right side while preserving the gauge invariance by keeping all loops closed. The Elfin theory falls into the category of regularized solutions of the loop equation. 

As it satisfies the MM equation together with the initial condition $Z[\1] = \const{}$, it reproduces all the planar graphs of asymptotically free QCD. All we know about cusp singularities, perimeter corrections, and the dimensional transmutation of the QCD mass scale comes from that perturbation expansion. Furthermore, the Elfin theory offers a globally regularized definition of planar QCD as a system of free fermions living on a minimal surface: a well defined and solvable two-dimensional model. 

This model will be discussed in detail in the next section; in the rest of this section, we shall derive the planar graph from the loop equation in the context of Elfin theory.
\subsection{Bootstrap equation and  planar graphs with glassed windows}
The planar graphs in Feynman gauge, including the ghost loops, were reproduced from the loop equation in the original MM paper \cite{MM1981NPB} using the so-called Bootstrap equation (section 6 in that paper). This equation was based on the inversion of the point derivative operator $\pd_{\mu}$ in the equations of motion, using the Bianchi identity. The area derivative was split into linear $B_\alpha$ and quadratic $B_{\mu\nu}$  terms
\EQ{
&\fbyf{W[C]}{\sigma_{\mu\nu}(x)} = \lrb{\pd^x_\nu\delta_{\mu\alpha} -\pd^x_\mu\delta_{\nu\alpha}}B_\alpha[C_{xx}] + B_{\mu\nu}[C_{x x}];
}
The linear terms were related to the quadratic terms and the source term $J_\nu[C_{x x}]$ on the right side of the MM equation
\EQ{
  & B_\nu = \lrb{-\pd^2\delta + [\pd,\pd]}^{-1} \lrb{J_\lambda - \pd_\mu B_{\mu\lambda}};\\
& J_\nu[C_{x x}] = \lambda \INT{}{} d y_\nu \delta^4(x-y) W[C_{x y}] W[C_{y x}]
}
The next step was to expand the inversion operator in powers of the commutator
\EQ{
&\lrb{-\pd^2\1 + [\pd,\pd]}^{-1} = - \pd^{-2} \1 \br
&- \pd^{-2}[\pd,\pd]\pd^{-2} - \pd^{-2}[\pd,\pd]\pd^{-2}[\pd,\pd]\pd^{-2} + \dots
}
and use the identity
\EQ{
[\pd_\alpha,\pd_\beta] F[C_{x x}] = \left.\fbyf{F[C_{x x}]}{\sigma_{\alpha_{\beta}}(y)}\right|^{y= x-0}_{y =x +0 }
}
The last step was to represent the inversion of the point derivative operator $\pd^{-2}$ as a sum over the Brownian path $\Gamma$ in 4D space of the functional $J[C_{x x}]$ transported along this path by adding "wires" to the loop
\EQ{
\label{WWGamma}
&- \pd^{-2} J_\nu[C_{x x}] = \lambda \int d^4 z \INT{}{} \D \Gamma_{x z} \br
&\oint\displaylimits_{\Gamma_{z x}\cdot C_{x y}\cdot C_{y x}\cdot \Gamma{x z}} d y_\nu \delta^4(z-y) W[\Gamma_{z x}\cdot C_{x y}] W[ C_{y x}\cdot \Gamma_{x z}]
}
The important property of this bootstrap equation is that it involves the path integral over Brownian paths $\Gamma_{x z} $ \emph{in Euclidean 4D space}, which is the basis for the further reconstruction of the gluon propagators . This integral representation is depicted on Fig.\ref{fig:WWGamma}.
\begin{figure}
    \centering
    \includegraphics[width=1.\linewidth]{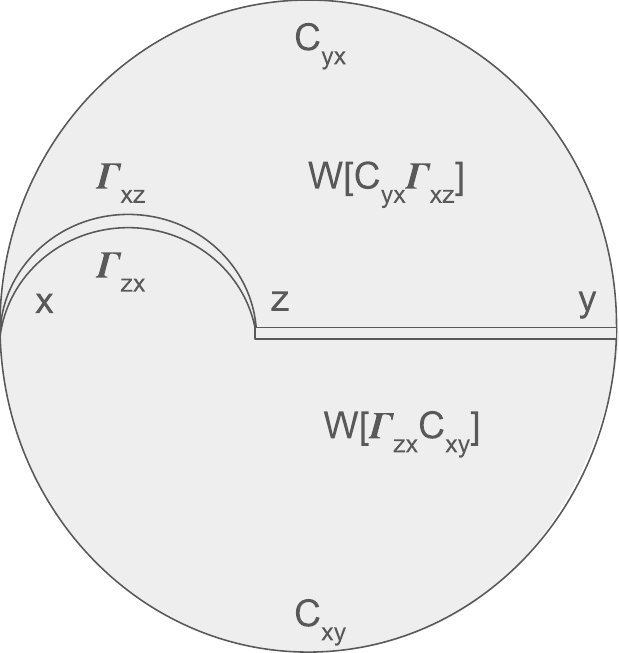}
    \caption{The path integral in \eqref{WWGamma} with delta function represented by a straight double line, and Brownian paths $\Gamma_{x z},\Gamma_{zx}$  by a crescent.}
    \label{fig:WWGamma}
\end{figure}
\subsection{Gluon graphs by iterations of Bootstrap equation}
In the old MM papers \cite{MM1981NPB, Makeenko:1981yf}, the bootstrap equation was iterated in the coupling constant $\lambda$, starting with  $W[\1] =1$.

Iterative evaluation reproduces conventional planar graphs, including Faddeev-Popov ghost loops. The first order planar graph follows immediately from the  \eqref{WWGamma}, as shown on Fig. \ref{fig:WWGamma}, by replacing $W$ by $1$ in the path integral, after which this path integral yields the gluon propagator $1/(x-y)^2$.
The second order graphs are already nontrivial (see \cite{MM1981NPB}, Appendix B). This is where the Faddeev-Popov ghost loops emerge directly from iterations of the bootstrap loop equation; this is also where the beta function appears, leading to asymptotic freedom with its running coupling constant. Perturbative QCD is reproduced by the bootstrap equation to all orders, as it was further investigated in \cite{Makeenko:1981yf}. 

Let us summarize this section as follows.
\begin{itemize}
    \item The planar MM equation \eqref{MMeq} can be transformed into an equivalent path integral equation by inverting the non-commutative operator $\pd{}$. The new equation expresses the area variation as a non-linear functional of $W[C]$, which can be represented as a series of frame diagrams, like the one in Fig. \ref{fig:WWGamma}.
\item Frame diagrams can be constructed by means of a certain operator technique. The frame diagram of an arbitrary order looks like a planar tree with windows glassed with $W$ functionals. Integration over internal paths is implied. 
\item Perturbatively, when $W[C_{ab}] \Ra 1$ in a framed graph, the sums over Brownian paths produce the gluon propagators, and the commutators $[\pd{}, \pd{}]$  generate proper tensor numerators in the gluon graphs.
\item The analytical expression for the frame diagram is manifestly gauge invariant. It contains only gauge-invariant (and parametric-invariant) quantities: $W$ functionals of closed loops and line elements, $\dd \theta\dot C_\mu(\theta)$. The loop integration goes along the external loop, $C_{x x}$, as well as along internal paths (with appropriate coefficients). 
\item The bootstrap equation is equivalent to the system including the Bianchi identity and the planar equation accompanied by euclidean boundary conditions. For this reason, iterative solution of the bootstrap equation recovers uniquely Faddeev-Popov perturbation theory (see appendices of \cite{MM1981NPB} and \cite{Makeenko:1981yf}). 
\item Various parts of the loop integral correspond to various graphs of ordinary perturbation theory including ghost loops. Being gauge variant separately, these graphs combine into a manifestly gauge-invariant diagram in the loop space. 
 \end{itemize} 
\section{Effective action, its geometry and its singularities}
\label{sec:effectiveaction}
The formal proof that the Elfin determinant $Z[C]$ in the limit of large Elf mass and vanishing normalization $Z[\1]$ satisfies the MM loop equation with an asymptotically free coupling constant $\lambda \to 0$ raises  many questions. The loop equation is local, due to the delta function on the right side. 
However, as we all know, the perturbative expansion of $W[C]$ in powers of the 't Hooft coupling constant $\lambda$ produces the planar graphs inside the loop $C$, which are nonlocal functionals of the loop because of the massless gluon propagator $\delta_{\mu\nu}/(x-y)^2$ in coordinate space. We get a sequence of logarithmically divergent integrals like
$$\oint_C d x_\mu \oint_{C_{x x}} d y_\mu/(x-y)^2,$$
leading to logarithmic singularities at the cusps of the  loop $C$.

We now trace the origin of this non-locality and the emergence of logarithmic singularities within the Elfin theory.

\subsection{Local expansion of the heavy fermion determinant}
The physical interpretation of our confining factor arises from the effective action of a heavy Dirac fermion propagating on the minimal surface $\Sigma$. In the limit of large mass $m$, the fermion determinant $W_{\text{eff}} = -\log \det(\slashed{D} + m)$ admits a local heat-kernel expansion, as we have found in previous sections:
\EQ{
\label{LiouvilleInt}
&W_{\text{eff}}[C] \approx  \INT{\Sigma}{} d^2 z  \sigma e^{2\rho} + \frac{|\pd_z\rho|^2}{3\pi}  + O(m^{-2});\\
& \rho(z,\bar z) = \frac{\log( \bar \lambda(\bar z) \lambda(z)) + \log( \bar \mu(\bar z) \mu(z))}{2};
}
The leading term is the Area Law, where the tension $\sigma \sim m^2$ provides the confining potential. The subleading term is the conformal anomaly (Liouville action). There is always another zero mode confining factor $G[C]= \exp{-\delta \sigma S[C]}$, renormalizing the string tension $\sigma$, as we found in the previous work \cite{migdal2025geometric}. Equivalently, the area term generated in the heavy-mass expansion of the determinant provides an additive renormalization of the vacuum-energy/string tension sector, already parametrized by the zero-mode dressing $\exp{-\kappa S[C]}$; only the sum defines the single physical string tension $\sigma_{phys}=2\sqrt{2}\kappa_{phys}$ in the continuum theory, while $M$ acts  as a regulator scale removed in the local limit.

In our framework, we do not integrate over fluctuating metrics (which would lead to the Polyakov-Liouville anomaly). Instead, the surface $\Sigma$ is the rigid solution to the Plateau problem bounded by the loop $C$. Consequently, the anomaly term becomes a functional of the boundary loop itself:
\EQ{
S_{\text{shape}}[C] =\frac{1}{12\pi} \INT{D}{} \left| \pd_z\lrb{\log \bar\lambda \lambda + \log \bar\mu \mu}  \right|^2 d^2z
}
This functional is non-local with respect to the loop coordinates $C_\mu(\theta)$, as it requires solving the bulk Laplace equation to determine $\lambda(z),\mu(z)$. It represents the "Loewner Energy" \cite{Wang:2019loewner} of the loop, acting as a geometric stiffness that penalizes deviations of the flux tube cross-section from circularity.

\subsection{Cusp singularities and the Hilbert transform}
The analytic structure of the minimal surface provides a natural regularization of the ultraviolet divergences associated with the Wilson loop, while correctly reproducing the universal singularities of Gauge Theory. The holomorphic tangent vector is reconstructed from the boundary loop velocity $\dot{C}_\mu$ via the Hilbert transform $H$:
\EQ{
\label{HilbertTransform}
&2\I e^{\I\theta} f'_\mu(e^{\I\theta}) = \dot{C}_\mu(\theta) + \I H[\dot{C}_\mu](\theta), \\
& H[u](\theta) = \frac{1}{2\pi} \text{P.V.} \INT{}{} d\theta' \cot\lrb{\frac{\theta-\theta'}{2}} u(\theta')
}
For smooth loops (such as the circle or helicoid), $H[\dot{C}]$ is smooth, and the area density is finite everywhere. However, for polygonal loops used in scattering amplitudes, the tangent vector $\dot{C}_\mu$ has discontinuities (cusps).

At a cusp, the Hilbert transform of the step function generates a logarithmic singularity:
\EQ{
H[\dot{C}_\mu](\theta) \sim \log|\theta-\theta_{cusp}|
}
This propagates into the area tensor $F_{\mu\nu} = \Re f'_{[\mu} \Im f'_{\nu]}$, causing the area derivatives of the minimal surface to diverge logarithmically at the corners of a polygonal loop. This divergence is the geometric origin of the Cusp Anomalous Dimension ($\Gamma_{cusp} \log \Lambda_{UV}$) and the Sudakov double logarithms found in perturbative QCD. Thus, the Hodge-dual minimal surface captures (or at least imitates) the collinear singularities of the gauge theory through the poles of the Hilbert kernel.

\subsection{New view at the cusp singularity}
The logarithmic singularity at the cusps, derived in perturbative QCD by Korchemsky and Radyushkin \cite{KorchemskyRadyushkin1987}, finds a natural interpretation in our framework. In standard perturbation theory, the Wilson loop with a cusp of angle $\gamma$ (in Euclidean space) exhibits a divergence:
\EQ{
\log W &\sim -\Gamma_{cusp}(\gamma) \log(\Lambda_{UV} L) \br
&\sim -\frac{\lambda}{8\pi^2} (\gamma \cot \gamma - 1) \log(\Lambda_{UV} L)
}
In the standard renormalization group approach, this divergence is absorbed into a multiplicative renormalization constant. However, in our geometric framework, the logarithmic divergence arises physically from the energy density of the minimal surface near the corner, generated by the pole in the Hilbert transform.

Rather than removing this term, we observe that the theory possesses an intrinsic ultraviolet cutoff $\Lambda$ determined by the mass scale of the heavy Elf particles, which is equivalent to the regularization of the fermion determinant. In the continuum limit $\Lambda \to \infty$, the bare 't Hooft coupling $\lambda(\Lambda)$ must vanish according to the law of asymptotic freedom:
\EQ{
\lambda(\Lambda) \approx \frac{24\pi^2}{11 \log(\Lambda/\Lambda_{QCD})}
}
Consequently, the "divergent" cusp energy becomes the product of a vanishing coupling and a diverging geometric factor. This product tends to a finite, universal limit:
\EQ{
E_{cusp} &\sim \lim_{\Lambda \to \infty} \lambda(\Lambda) \log(\Lambda L) \times (\text{Geometric Factor}) \br
&= \text{finite}
}
Specifically, using the one-loop beta function coefficient, the cusp contribution to the effective action becomes a scale-invariant geometric potential:
\EQ{
S_{cusp} = \frac{3}{11} (\gamma \cot \gamma - 1)
}
This mechanism is  analogous to the continuum limit in Lattice Gauge Theory, where the bare coupling $g_0 \to 0$ as the lattice spacing $a \to 0$ to keep physical masses finite. In our induced QCD, the cutoff has a physical and mathematical origin in the regularization of the minimal surface, so the limit $m \to \infty$ (where $\lambda \to 1/(\beta_0 \log m)$) has a precise mathematical meaning, replacing the formal subtraction of infinities with a finite limiting process.

\subsection{Spinor factorization and the geometry of Gauss maps}\label{GaussMapsTwistor}
The twistor factorization \cite{Penrose:1984spinors}  of the null tangent vector $f'_{a\dot{b}} = \lambda_a \mu_{\dot{b}}$ leads to a factorization of the induced metric on the worldsheet, as summarized in our \ref{GenTheory}. Using the isomorphism $SO(4) \cong SL(2)_L \times SL(2)_R$, the metric density factorizes into the product of invariant norms of the left and right spinors:
\EQ{
e(z,\bar z) = 2\sqrt{2} f'_\mu \bar{f}'_\mu = \sqrt{2} \bar\lambda \lambda \bar\mu \mu 
}
The area derivative \eqref{arderHodge} in the spinor representation simplifies as follows 
\EQ{
\label{arderGaussmap}
&\fbyf{S_\pm}{\sigma_{\alpha\beta}} = -n^{\pm}_i \eta^{i\pm}_{\alpha\beta};\\
& n^{\pm}_i = \left\{ \frac{\bar\lambda\sigma_i\lambda}{\bar\lambda\lambda}, \frac{\bar\mu\sigma_i\mu}{\bar\mu\mu}\right\}
}

The Liouville field $\rho = \oh \log(f'_\mu \bar{f}'_\mu)$ splits into a sum of potentials:
\EQ{
\label{KahllerPot}
&\rho(z,\bar z) = \Phi_L(z,\bz) + \Phi_R(z,\bz), \br
&\text{where } \Phi_{L,R} =  \left\{ \oh\log \bar\lambda \lambda , \oh\log \bar\mu \mu \right\};
}
Geometrically, $\Phi_L$ and $\Phi_R$ are the Kähler potentials for the maps from the worldsheet to the projective spinor spaces $\mathbb{CP}^1_L$ and $\mathbb{CP}^1_R$. These maps $n_i^\pm$ are identified with the Left and Right Gauss maps of the minimal surface \cite{Hoffman:1980gauss}. In our particular case of the Hodge-dual minimal surface $S_\chi$, the area element reduces to a single Gauss map for a given Hodge chirality $\chi = \pm 1$.
\EQ{
&\Sigma_{\mu\nu} \propto n^{\pm}_i \eta^{i\pm}_{\alpha\beta} H^{\mp}(z,\bar z);\\
& H^{\pm}(z,\bar z) = \left\{ \bar\lambda \lambda , \bar\mu \mu\right\};
}
The area element still depends on both left and right spinors, but it transforms under $SO(4)$ as a single Gauss map, corresponding to its chirality.
The anomalous part of the effective action, $I \propto \INT{}{} (\nabla \rho)^2$, therefore expands into:
\EQ{
I = I_L[\lambda] + I_R[\mu] + 2 \INT{D}{} \nabla \Phi_L \cdot \nabla \Phi_R \, d^2z
}
The first two terms are the actions for the Sigma models on $\mathbb{CP}^1$, which are topological invariants (proportional to the winding numbers of the Gauss map). The third term represents a non-trivial interaction between the left and right sectors. Unlike the metric, the action does not decouple; this interaction term measures the integrated correlation between the curvatures of the left and right spinor bundles. It represents the geometric energy cost of the relative orientation of the left and right Gauss maps required to satisfy the boundary conditions.

\section{Momentum loops and continuum limit}
\label{sec:momentumloops}
\subsection{Regularization, renormalizability, and the microscopic definition of Planar QCD}
\label{sec:RegularRenorm}
Since the inception of QCD half a century ago, the understanding of its regularization and renormalizability has evolved significantly. Currently, Lattice QCD serves as the standard microscopic definition of the theory. In this framework, renormalizability reduces to the mathematical statement that, in the limit of a vanishing bare coupling constant, the physical mass scale $\Lambda_{QCD}$ (measured in units of inverse lattice spacing) exponentially approaches zero. Consequently, the lattice spacing vanishes when measured in observable physical units such as the proton radius.

However, the lattice definition comes at a cost: it explicitly breaks the continuous symmetries of four-dimensional Euclidean space, from the rotation group $O(4)$ to subtle topological structures like the Hodge duality of the area tensor. For instance,  instanton solutions or Hodge-dual minimal surfaces do not exist on a discrete grid; they only emerge in the continuum limit.

These disadvantages are compensated by the universality of the Renormalization Group (RG) flow. It is well-established (though unproven)  that asymptotically free theories reach a universal fixed point at large scales, regardless of the specific microscopic regularization, provided the symmetries are restored in the infrared. This universality allows for various microscopic definitions of QCD, provided they land in the correct phase (the basin of attraction of the asymptotically free fixed point).

Previous lattice approaches illustrate the sensitivity of this limit. The large-N reduced Eguchi-Kawai model, for example, encountered unphysical fixed points. Similarly, the original Induced QCD proposal \cite{Kazakov:1992}—which attempted to induce gauge dynamics via heavy free fermions with vanishing color currents on a lattice—failed because the theory collapsed into the wrong phase rather than asymptotically free QCD. Our Elfin theory is a variation on this theme: heavy free fermions inducing QCD. However, the critical distinction is that our fermions live on a \emph{rigid minimal surface} rather than a four-dimensional lattice. This geometric constraint keeps the theory non-trivial in the local limit, making it a viable regulator for Planar QCD.

A central question arises regarding the \emph{renormalization} of the Wilson loop itself. Standard perturbation theory dictates that $W[C]$ requires multiplicative renormalization to remove divergences associated with cusps and self-intersections. However, redefining the loop equation solely to render the coordinate-space Wilson loop finite in the local limit obscures the underlying dynamics.

Physical renormalizability demands only the existence of a universal mass scale $\Lambda_{QCD}$ and the finiteness of \emph{observable scattering amplitudes}. The Wilson loop $W[C]$ is an intermediate, off-shell quantity. Physical observables involve path integrals over quark trajectories, which are governed by a Brownian measure. A Brownian trajectory is nowhere differentiable; it consists of an infinite density of cusps. Attempting to renormalize the Wilson loop for a smooth contour $C$ is physically artificial because the dominant contributions to the path integral come from fractal paths where the concept of a "cusp angle" is ill-defined.

Therefore, instead of forcing the non-renormalizable coordinate Wilson loop to be finite, we adopt a radically different strategy:
\begin{enumerate}
    \item We keep the theory regularized at the level of the Wilson loop $W[C]$ by the finite Elf mass $m$ (which acts as the UV cutoff).
    \item We perform the summation over quark paths to construct the momentum loop amplitude $W[P]$.
    \item Only \emph{after} transitioning to Momentum Loop Space do we take the local limit ($m \to \infty$).
\end{enumerate}
As demonstrated in this paper, the transition to momentum space integrates out the geometric singularities. The momentum amplitude $W[P]$ remains finite and satisfies an algebraic-differential equation. By taking the bare coupling to zero as $m \to \infty$ in accordance with the RG beta function, we recover the finite, universal observables of Planar QCD without ever needing to define a "renormalized" Wilson loop in coordinate space.

\subsection{Momentum loop path integral}
As we discussed in the previous section, the ordinary Wilson loop is not  directly observable, and its  singularities are merely intermediate results that will eventually disappear (cancel or integrate out) in observables. 

Unlike the standard Feynman-Schwinger worldline representation in coordinate space, which is inherently plagued by the UV cusp singularities of $W[C]$, we utilize the phase-space path integral representation introduced by the author in \cite{SecQuanM95}. In this formulation, the  planar quark amplitude factorizes  into a universal kinematic Dirac phase-space trace $\mathcal{K}[P]$ and the dynamical gluonic momentum-loop functional $W[P+Q]$.

The amplitude for the propagation of a free Dirac particle in the phase-space loop (the Dirac loop) is:
\EQ{
\mathcal{K}[P] =\tr \hat P \exp{-\int_{0}^\infty \dd \tau (\I \gamma_\mu P_\mu(\tau) + m_q) };
}
The full amplitude for the  quark loop in the Planar QCD vacuum, with some external momenta $q_1,\dots q_n$ injected into these amplitudes,is a product of the Dirac loop and Wilson loop times the vertex factor for the injection of external momenta:
\EQ{
\label{Amplitude}
&A[q_1\dots q_n] \propto \int\displaylimits_{\text{ordered } \tau_k} \prod \dd \tau_k \int \D P W[P + Q] \mathcal{K}[P]\br
& Q_\mu(\tau) = \sum_k q_k \Theta(\tau-\tau_k) ; \quad \sum_k q_k =0\\
\label{WPdef}
& W[P] = \lambda\int \D C W[C] \exp{ \I \int d \tau P_\mu(\tau) \dot C_\mu(\tau)}
}
For these amplitudes, the fundamental domain $D$ for the minimal surface is used as the upper semiplane $\Im z >0$, with the loop $C(x)$ corresponding to the real axis $x = \Re z$.
This formula for the quark current amplitude in Planar QCD is illustrated in Figure \ref{fig:FeynmanWheel}.

\begin{remark}{\textbf{Phase Space Normalization and the Departure from the Brownian Gauge}.}

The phase-space formulation possesses a canonical normalization. The functional measure is  normalized by the standard symplectic volume element $\mathcal{D}P \mathcal{D}C' / (2\pi)$ per continuous degree of freedom. Consequently, in the absence of the interacting Wilson loop (the free-particle limit, $W[C]=1$), integrating over the loop coordinates $C$ natively yields a functional Dirac delta function,  enforcing the conservation of local momentum:
\EQ{
&\int d^4k \int \frac{\mathcal{D}C' \mathcal{D}P}{2\pi} \exp{\I \oint (P_\mu(\tau) - k_\mu) \dot{C}_\mu d\tau} F[P] \br
&= \int d^4k \int \mathcal{D}P \, \delta[P(\cdot) - k] F[P] = \int d^4k \, F[k]
}
This demonstrates that the phase-space integral is  well-defined from standard quantum mechanics, unique  modulo the infinite 1D diffeomorphism volume $\text{Vol}(\text{Diff}(S^1))$ associated with the simultaneous reparametrization of $P(\tau)$ and $C(\tau)$.Conventionally, this $\text{Diff}(S^1)$ volume is factored out by imposing a proper-time gauge-fixing via an $\int (C')^2 d\tau$ kinetic term. This specific gauge choice inevitably generates the Feynman-Schwinger representation with its standard Brownian (Wiener) measure in loop space—a measure inherently dominated by everywhere-discontinuous fractal paths that artificially generate the  UV cusp singularities.
In our framework, we recognize that the Brownian measure is just one possible choice of gauge fixing for the 1D diffeomorphisms. We will deliberately abandon it. Instead, as detailed in Section \ref{sec:twistorParametrization}, we fix $\text{Diff}(S^1)$ by imposing the Virasoro constraint $(f')^2 = 0$ on the holomorphic twistors parametrizing the loop velocity. This specific gauge choice entirely avoids the fractal artifacts of the Brownian measure and  locks the 1D phase-space kinematics of the quark loop to the 2D internal confining geometry of the minimal surface.
\end{remark}
The momentum loop amplitude $W[P]$ was introduced in our papers \cite{MLDMig86,SecQuanM95, Mig98Hidden} and was recently used in the solution of the Navier-Stokes turbulence problem \cite{migdal2024quantum, ReviewPaperAM}. The advantage of the momentum loop space is that the loop equation for $W[P]$ is purely algebraic/differential. There are no delta functions; just polynomial functionals of the local momentum $P(\tau)$ and some non-singular functional derivatives of $W[P]$. The coordinate delta function, which leads to all perturbative singularities such as logarithmic divergences at the cusps, disappears in momentum space in the same way that it disappears from the Klein-Gordon equation $$(-\pd^2 + M^2) G = \delta(x-x_0) \Ra (p^2 + M^2) \tilde G = 1$$
\begin{figure}
        \centering
        \includegraphics[width=0.5\linewidth]{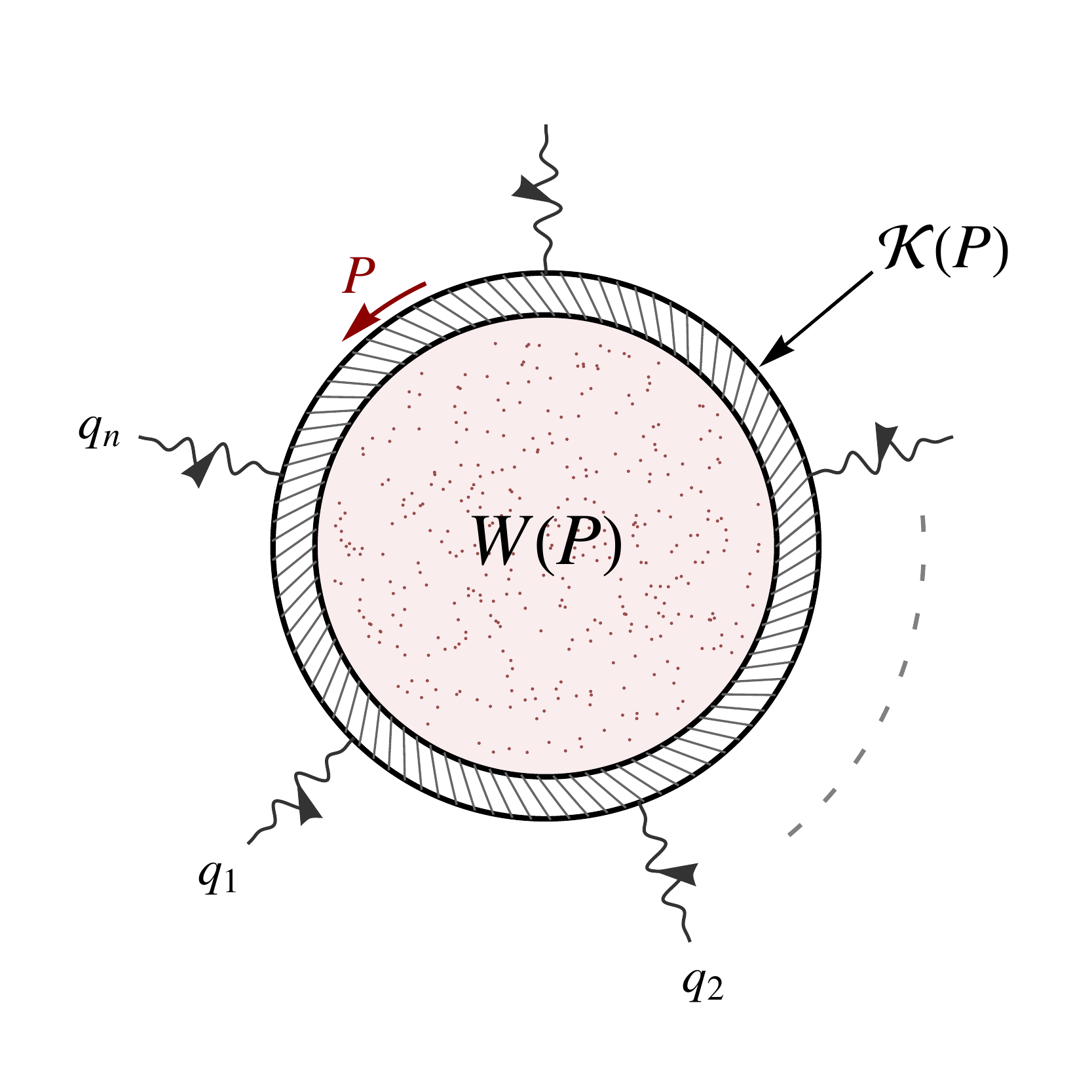}
        \caption{The momentum loop amplitude, $A(q_1,\dots q_n)$ with the inside part of the wheel corresponding to Momentum loop $W[P]$, and the outer rim corresponding to Dirac path amplitude $\mathcal K[P]$. The paths $P(t)$  are random, interacting with inner geometry of the string surface represented by amplitude $W[P]$.}
        \label{fig:FeynmanWheel}
    \end{figure}
\subsection{The Momentum loop equation}
The MM equation  \eqref{MMeq}, when transformed into the momentum loop equation (MLE) by functional Fourier transform, also has a simple algebraic structure ( expanded in  so-called Magnus invariant forms \cite{Magnus1954})
\EQ{
\label{MagnusExp}
&W[P] = \sum_n \mathcal{W}^{(n)}_{\alpha_1,\dots\alpha_n} \Omega^{n}_{\alpha_1\dots\alpha_n}[P];\\
&\Omega^{n}_{\alpha_1\dots\alpha_n}[P]=\hat T\int \prod \dd \theta_k P'_{\alpha_k}(\theta_k)
}
where $\hat T\int$ stands for the ordering of the integration angles on the circle.
The MLE becomes an algebraic differential equation in terms of Magnus forms (a.k.a iterated path integrals):
\EQ{
\label{MLENonsing}
 &Q_\nu[P] W[P] = \int_{0}^{2\pi} \dd \theta \ff{P_\nu(\theta)} \lrb{W[P_{0,\theta}]W[P_{\theta,2\pi}]};\\
 & Q_\nu[P] = T^{\alpha\beta\gamma}{}_\nu \Omega^{3}_{\alpha\beta\gamma}[P];\\
 \label{Tdef}
&T^{\alpha\beta\gamma}{}_\nu = \delta_{\alpha\beta}\delta_{\gamma\nu}+ \delta_{\gamma\beta}\delta_{\alpha\nu}-2\delta_{\alpha\gamma}\delta_{\beta\nu};
}

The algebraic expression on the left side emerged as a Fourier transformation $\ff{\dot C} \Ra \I P$ of the operator 
\EQ{
\hat L_\nu(0) &= \pd_\mu \ff{\sigma_{\mu\nu}(0)}\br
&=T^{\alpha\beta\gamma}{}_\nu \frac{\delta^3}{\delta \dot C_\alpha(\theta-0)\delta \dot C_\beta(\theta)\delta \dot C_\gamma( \theta+0)}
\label{LnuCspace}
}
The  delta function on the right (which caused all the perturbative singularities in $W[C]$)  disappeared in the Fourier integral (see original papers \cite{MLDMig86,SecQuanM95, Mig98Hidden} for details). 

Here is how this happens in our new framework, with the closed loop  described by the position field $C_\nu(\theta)$ subject to the closure constraint $C(2\pi)=C(0)$. 
The linear functional measure in loop space can be written in one of two equivalent forms ( with $v(\theta) = C'(\theta)$)
\EQ{
\label{linmeasure}
\D C&\propto \delta^4\lrb{C(2\pi)-C(0)} \prod_{\theta=0}^{2\pi} d^4 C(\theta) \br
&\propto \delta^4\lrb{\oint v(\theta) d \theta} \prod_{\theta=0}^{2\pi} d^4 v(\theta)
}
It is factorizable at self-intersections, which leads to the factorization in the momentum loop equation, with the delta function absorbed by a measure
\EQ{
&\delta^4\lrb{C(2\pi)-C(0)} \delta^4\lrb{C(\theta_1)-C(\theta_2)}\prod_{\theta=0}^{2\pi} d^4 C(\theta) \br
&= \delta^4\lrb{C(\theta_1)-C(\theta_2)} \prod_{\theta' =\theta_1}^{\theta_2} d^4 C(\theta') \br
&\delta^4\lrb{C(\theta_2)-C(\theta_1+2\pi)} \prod_{\theta'' =\theta_2}^{\theta_1+2\pi} d^4 C(\theta'') \br
&\implies \D C\delta^4\lrb{C(\theta_2) - C(\theta_1)} = \D C_{12} \D C_{21}\br
&\implies \int \D C \hat L_\nu(\theta_1) W[C] \exp{\I \oint P \cdot C'} = \br
&   \int \D C_{12} \int d\theta_2 \stackrel{\longleftrightarrow}{C'_\nu(\theta_2)} \exp{\I \int_{\theta_1}^{\theta_2} P \cdot C'} W[C_{1 2}] \br
&\int\D C_{21}  \exp{\I \int_{\theta_2}^{\theta_1+2\pi} P \cdot C'} W[C_{21}]\br
&\implies  L_\nu W[P_{1 1}] = \int d\theta_2 \stackrel{\longleftrightarrow}{\ff{P_\nu(\theta_2)}} W[P_{1 2}] W[P_{2 1}]
}
Here $$\stackrel{\longleftrightarrow}{f(x)} = \frac{f(x+0) + f(x-0)}{2}$$
After conformal transformation mapping the unit disk to the upper semiplane, we map the circle origin $\theta = \pm \pi$ by $x = \tan (\theta/2)$ to $ x = \pm \infty$. 

For our final local solution, however, we shall use a different gauge condition for diffeomorphisms, which is more adequate for string theory. The gauge ambiguity of the parametrization of the linear measure \eqref{linmeasure} leads to the overcounting of every loop by an infinite volume of the gauge orbit. This volume must be factored out of the measure to avoid overcounting. This refactoring  seems incompatible with the above factorization of the loop measure over the self-intersecting loops, but upon closer inspection, the paradox disappears. The measure must be factored by the total volume of its isometries, which, in the case of a simple loop, reduces to the volume of diffeomorphisms of a unit circle  $\text{Vol}(\text{Diff}(S^1))$ ; however, in the case of a self-intersecting loop, there are two periodic parts; therefore, such a loop \emph{maps twice} the unit circle to $\R^4$. The volume of these isometries is $\lrb{\text{Vol}(\text{Diff}(S^1))}^2$, which matches the product of linear measures on the right side of the factorization identity. So, the physical measure, with single counting of each loop, factorizes in the MM equation in momentum loop space.

\subsection{Equivalence of loop diffusion operator to the third Magnus form}
In momentum loop space, the loop diffusion operator $L_\nu$ from \eqref{LnuCspace} is recast as a trilinear product of the momentum field $P(\theta)$:
\EQ{
&T^{\alpha\beta\gamma}{}_\nu\I^3 P_\alpha(\theta_1) P_\beta(\theta_2)P_\gamma(\theta_3);\\
& \theta_1 = 2 \pi -\epsilon, \theta_2 = 0, \theta_3 = +\epsilon;\quad \epsilon\to +0;
}
By utilizing the integral representation
\EQ{
P(\theta) = \oint d \tau P'(\tau) \Theta(\theta-\tau)
}
where $\Theta(\theta-\tau)$ denotes the periodic theta function on the circle, the trilinear product is transformed into a general triple integral of the momentum derivative $P'$. The rigorous mathematical foundation for expanding loop functionals in such iterated integrals is provided by Chen's theorem \cite{Chen1977}. Chen established that the algebra of iterated path integrals separates points in loop space, meaning that the infinite hierarchy of these integrals forms a  complete functional basis. This guarantees that our formal expansion in terms of Magnus forms is exhaustive.

Through direct tensor contraction, the kinematic tensor $T^{\alpha\beta\gamma}{}_\nu$ implements the  Dynkin projector \cite{Dynkin1947} onto the nested double commutator $[\hat{D}_\alpha, [\hat{D}_\alpha, \hat{D}_\nu]]$. Here, $\hat D = \partial + A$ represents the covariant derivative operator within the gauge field representation of the Wilson loop. This mapping identifies the associative word structure of the integrated momentum field with its unique Lie algebra representation. 

Consequently, by Ree’s Theorem \cite{Ree1958} for Free Lie Algebras \cite{Reutenauer1993}, this operator identically annihilates the symmetric shuffle ideal associated with the triple integral. As the point-split coordinates $\theta_i$ collapse to the origin, the periodic step-functions define a  cyclically ordered simplex. This unique parametric-invariant Lie projection ensures that the functional converges to the third-order Magnus form $\Omega^{(3)}$. Direct evaluation of this triple integral, mediated by the $T$ projector, confirms this algebraic correspondence and yields the canonical normalization coefficient $1/6$.
\subsection{Algebraic recurrence for the Momentum Loop Equation}

The Momentum Loop Equation (MLE) exhibits a significant simplification when the momentum functional is represented as a Magnus-ordered exponential:
\EQ{
W[P] = \tr\left\{\hat{\mathcal{T}}\exp{\int d\theta\,\hat{X}_{\mu}P_{\mu}^{\prime}(\theta)}\right\}
}
where $\hat{X}_\mu$ represents the operator position of the string endpoint in spacetime. Expanding this trace in Magnus forms $\Omega^n$, we obtain universal tensor coefficients $\mathcal{W}_{\alpha_1\dots\alpha_n}^{(n)} = \mathrm{tr}\{\hat{X}_{\alpha_1}\dots\hat{X}_{\alpha_n}\}$. The functional derivative acts algebraically on this structure, bringing down  commutator insertions:
\EQ{
\label{MomentumArDer}
&\frac{\delta}{\delta P_{\mu}(\theta)}\mathrm{tr}\left\{\hat{\mathcal{T}}\exp{\int d\theta^{\prime}\hat{X}_{\alpha}P_{\alpha}^{\prime}}\right\} = \br
&P_{\nu}^{\prime}(\theta)\mathrm{tr}\left\{[\hat{X}_{\mu},\hat{X}_{\nu}]\hat{\mathcal{T}}\exp{\int_{\theta}^{\theta+2\pi}\hat{X}_{\alpha}P_{\alpha}^{\prime}}\right\}
}
By applying this exact derivative formula, the MM loop equation sheds its coordinate-space singularities and collapses into a finite, recursive combinatorial algebra. On the left-hand side, the kinematic operator $Q_\nu$ acts via a shuffle product between the 3-index $T^{\alpha\beta\gamma}_\nu$ tensor and the lower-order $\mathcal{W}^{(n-3)}$ coefficients. On the right-hand side, the integration over $\theta$ seamlessly sews the split traces back together, yielding an algebraic sum over commutator insertions without requiring cyclic sums or residual functional calculus.
These relations  are depicted in Fig.\ref{fig:MLE}
\begin{figure}
    \centering
    \includegraphics[width=1.\linewidth]{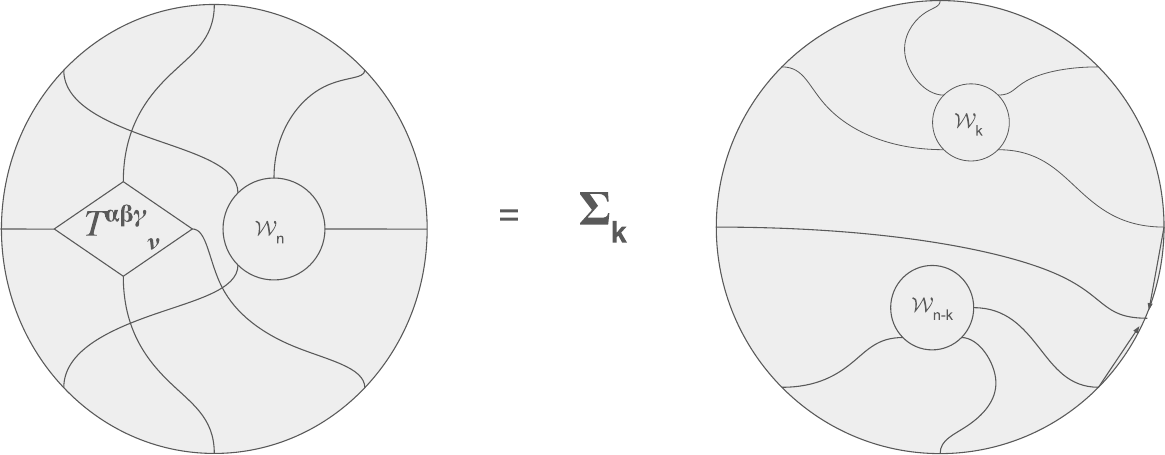}
    \caption{The recurrent equation for the coefficient tensor parameters $\mathcal{W}^{n}$ in the Magnus expansion \eqref{MagnusExp}. The external circle represents the unit circle where the angular parameter belong. The inner round blobs with wavy legs landing on a unit circle represent the $\mathcal{W}^{n}$  tensors, the rhombus with four wavy legs represents the 4-tensor $T$ in \eqref{Tdef}. Two arrows on the right side correspond to functional derivatives $\ff{P_\nu}$  bringing  commutators by  momentum area derivatives \eqref{MomentumArDer}. These equations relate the higher  $\mathcal{W}^{n}$  tensors to the lower ones, which allows the solution depending on some gauge parameters, undetermined by the recurrent relations.}
    \label{fig:MLE}
\end{figure}
Because the physical momentum loop is closed ($\oint P^\prime d\theta = 0$), the first-order form identically vanishes: $\Omega^1_\alpha=0$. This topological constraint generates a formal ``shuffle ideal'' that relates higher-order Magnus forms. We must therefore equate the left and right sides of the MLE modulo this ideal. The null space of this ideal natively gives rise to symmetric gauge ambiguities that dynamically decouple from physical observables.

\subsection{Exact solution up to $\mathcal{O}(P'^7)$ and non-analyticity of the Taylor expansion}
\label{sec:exact_MLE}

To construct the exact solution to the Momentum Loop Equation (MLE) in the perturbative regime of small loop momenta, we employ the functional Taylor--Magnus expansion of $\mathcal{W}[P]$. Because the MLE evaluates the functional derivative on closed loops, the geometric ansatz must natively reflect this exact topological closure.

When the functional area derivative acts, it splits the integration domain into two subloops on the right-hand side of the equation. To analytically enforce the closure of these subloops within the continuous path integral, we introduce a modified momentum derivative containing a boundary Dirac delta function:
\EQ{
P'_{\text{closed}}(\theta') = P'(\theta') + \Delta P \, \delta(\theta' - \theta_0)
}
where $\Delta P = \int P' d\theta$ is the boundary gap of the subloop, and $\theta_0$ is the position of the gap ($\theta_0 = \theta$ and $2\pi$ for the respective right-hand side subloops).

When substituted into the iterated integrals of the Magnus expansion, the $\delta$-function explicitly inserts the macroscopic gap vector $\Delta P_\mu \equiv \Omega^{(1)}_\mu$ precisely at the boundaries of the integration domain. Because $\Delta P$ corresponds to the first-order Magnus form, this topological closure modifies the $n$-th order form by subtracting the product $\Omega^{(1)} \Omega^{(n-1)}$. By the integration-by-parts identity of the Magnus expansion (the shuffle relations), this product translates to the sum of all interleaved operator positions. The loop-closing transformation thus yields the modified Magnus form:
\EQ{
\widetilde{\Omega}^{(n)} = \Omega^{(n)} - \sum_{k=1}^{n-1} \Omega^{(n)}(\text{interleaved})
}
This subtraction mathematically enforces the closed-loop Shuffle Ideal, identically restoring the cyclic symmetry of the invariant tensors $\mathcal{W}^{(n)}$. Operating in the pure closed-loop vacuum, all odd-parity traces identically vanish ($\mathcal{W}^{(2k-1)} = 0$). The functional $\mathcal{W}[P]$ is therefore expanded in a basis of cyclic symmetric, invariant tensor polynomials constructed entirely from the Kronecker metric $\delta_{\mu\nu}$.

Substituting this closed-loop basis into the MLE, we dynamically equate the left-hand side kinematic shuffle products against the right-hand side discrete commutator series. Starting from the perturbative vacuum $\mathcal{W}^{(0)}=1$, the algebraic system  yields finite solutions for the lowest-order tensors. The fundamental string tension enters at $\mathcal{W}^{(2)}$, and the kinematic stress generated by the loop derivative constrains the cyclic symmetric pairings of $\mathcal{W}^{(4)}$ and $\mathcal{W}^{(6)}$. (The explicit analytic formulas for these invariant tensors are provided in \ref{sec:appC} for small $n$. Crucially,  \ref{sec:appC} also contains the exact finite-dimensional linear algebra diagnostics (Table \ref{tab:rank_diagnostics}) and the explicit Fredholm inconsistency proof demonstrating the absolute lack of solutions at $\mathcal{W}^{(8)}$. The general algorithmic implementation is provided in \cite{MBMLEAlgebraic}.)

\subsection{The Reasons for the Failure of the Magnus Expansion for $W[P]$}
The general reason for the breakdown of the expansion at this order is that the linear space of $\mathcal{W}^{(8)}$ invariant tensors---which is  constrained by the Shuffle Ideal and cyclic symmetry---possesses a significantly smaller matrix rank than the highly specific multilinear stress generated by the crossing Kronecker deltas of $\mathcal{W}^{(4)} \times \mathcal{W}^{(4)}$. 

This is not a heuristic observation, but a  quantified algebraic contradiction. As detailed in \ref{sec:appC} (and summarized in Table \ref{tab:rank_diagnostics}), extracting the  linear system at the 8th order yields 379 geometric constraints but only 212 available parameters (which exhaustively includes the full $\mathcal{W}^{(8)}$ basis, all residual lower-order free parameters, and all gauge/ideal freedoms). The resulting system explicitly violates the Fredholm alternative condition, yielding a  positive rank deficiency ($\Delta\text{rank} = 1$) and a non-zero projection of the inhomogeneous stress onto the left-null space of the kinematic matrix.

Consequently, it is mathematically impossible for any choice of linear coefficients, or the remaining lower-order terms, to absorb and cancel this rigid cross-term stress. This yields an irreducible algebraic contradiction in the tensor matching.

The physical meaning of this failing Taylor--Magnus expansion is of paramount importance. It demonstrates that the MLE is finite, free of UV divergences, and expandable in a functional Taylor series up to $\mathcal{O}(P'^6)$. However, the geometric contradiction at the 8th order  proves that the  solution $W[P]$ is \textbf{not an analytic functional} of the bounding loop $P(\theta)$. It cannot be indefinitely expanded in a continuous Taylor--Magnus series.

\begin{remark}[The Vector vs.\ Scalar Dimensionality Mismatch]
The fundamental root of this non-analyticity is that the loop equation is inherently a \textbf{vector} equation,
\EQ{
    \hat{L}_\nu W = \int \frac{\delta}{\delta P_\nu} (W \times W),
}
governing a \textbf{scalar} functional $W[P]$. The number of independent tensor structures available for a scalar functional grows  slower than those required for a free-index vector equation (growing approximately four times slower in $d=4$). This is why, at some critical level of the Magnus expansion, there will inevitably be more vector equations than available scalar parameters. 

This is a foundational feature of the loop equation: it descends directly from the vector Yang--Mills equations of motion, which dictate that a local vector loop equation (evaluated at a fixed point on the contour) must be satisfied. Resolving this overdetermined system demands a mathematical ``conspiracy'' among the internal degrees of freedom within $W[P]$. The generic Taylor--Magnus expansion simply lacks the internal geometric structure to orchestrate this conspiracy. 

(Interestingly, the same vector-versus-scalar mismatch is fatal to all naive bosonic string theories of QCD: the loop equations are  vector equations, but the Nambu--Goto string Laplace operator is a scalar, meaning a scalar area-derivative constraint fundamentally cannot capture the full vector loop equation without an anomaly.)
\end{remark}

However, \textbf{our  theory natively provides this required conspiracy!} In the previous sections, we derived  \textbf{vector} equations for the Elfin determinant on the Hodge-dual minimal surface, and these vector equations are satisfied across all vector components. This property is made possible by Hodge duality and the planar factorization of the Majorana determinant on this Hodge dual surface.

The  physical reasons for this non-analyticity will become fully apparent in later sections, when we construct the continuum solution and reduce it to the  twistor string path integral. In that formulation, the formal Taylor expansion in $P(\theta)$ corresponds to the perturbative expansion of the  exponential phase factor:
\EQ{
    \exp{  2\I \int d\theta \, P_\alpha(\theta) \Im(z \lambda \sigma_\alpha \mu)\big|_{z=e^{i\theta}} }.
}
For such a functional Taylor series to converge, the boundary twistors $\lambda(z)$ and $\mu(z)$ would need to be  bounded. The algebraic contradiction at the 8th order  proves this is not the case; the path integral is not convergent in a perturbative Taylor expansion.

As we will discuss in detail when deriving the twistor string path integral, the integration over the momentum loops is not dominated by smooth, small $P$ variations at all. When integrated alongside the Dirac target-space product $\mathrm{tr} \, \hat{\mathcal{T}} \exp{ -\I \int d\theta \, \gamma_\mu P_\mu(\theta) + m_0 }$, we encounter  ultra-local integrals. At every angle $\theta$, we integrate out the local value $P(\theta)$. Because there are no worldline kinetic derivatives $\partial_\theta P$ in the action, the path integral completely factorizes point-by-point. As a result, the dominant trajectories are neither continuous nor dominated by small $P$ variations. We therefore conclude that the MLE is completely finite and anomaly-free, but its true continuum solution is inherently non-perturbative and cannot be Taylor-expanded beyond the 6th order.

\subsection{The General Master Field $\hat{X}$ and the Word Space}
\label{sec:master_field_words}
To fully grasp the physical and mathematical implications of the geometric breakdown at $\mathcal{W}^{(8)}$, it is instructive to re-examine the loop expansion through the lens of the Master Field construction.

A natural mathematical objection arises when considering the breakdown of the Taylor-Magnus expansion: Does the representation of the momentum loop functional as the trace of a path-ordered exponential,
\EQ{
W[P] = \mathrm{tr}\left\{\hat{\mathcal{T}}\exp{i\int dP_{\mu}\hat{X}_{\mu}}\right\}
}
restrict the generality of the solution? If $\hat{X}_{\mu}$ were assumed to be finite $N\times N$ matrices, the Cayley-Hamilton theorem and finite polynomial trace identities would artificially constrain the higher-order Magnus tensors $\mathcal{W}^{(n)}$, destroying the generality of the expansion.

However, in the planar limit ($N_{c}\rightarrow\infty$), the Master Field operates in a  infinite-dimensional Hilbert space. As established by the \textit{Master Field Theorem} proven by the author in 1994 \cite{SecQuanM95, Mig98Hidden}, any generic parametric-invariant loop functional expandable in iterated path integrals (guaranteed to be topologically complete by Chen's theorem \cite{Chen1977}) can be mapped  onto the non-commutative moment space of Free Probability and represented as a Magnus-ordered exponential path integral over a non-commutative word space.

The true non-perturbative Fock space of the string endpoint is the discrete space of ``words'' generated by $d=4$ Cuntz algebra operators:
\EQ{
a_{\mu}a_{\nu}^{\dagger}=\delta_{\mu\nu}, \quad a_{\mu}|0\rangle=0
}
In this space, the position operator $\hat{X}_{\mu}$ is uniquely constructed as an infinite, cascading sum of creation operators (the Voiculescu expansion \cite{Voiculescu1992}):
\EQ{
\hat{X}_{\mu}=a_{\mu}+\sum_{k=1}^{\infty}Q_{\mu,\mu_{1}\dots\mu_{k}}a_{\mu_{1}}^{\dagger}\dots a_{\mu_{k}}^{\dagger}
}
By the Free Moment-Cumulant Theorem \cite{Speicher1994}, this single, order-independent operator contains an infinite tower of algebraically unconstrained parameters (the free planar cumulants $Q$). Because the free Cuntz algebra contains no finite trace identities, our Master Field Theorem guarantees that the trace of the path-ordered exponential is mathematically unrestricted. It exhaustively spans the continuous functional space of any closed loop without truncation.

These $Q$ tensors are related to the invariant $\mathcal{W}$ tensors by a free cumulant expansion, uniquely solvable term by term for $Q(\mathcal{W})$, assuming the $\mathcal{W}$ tensors are finite. Therefore, the breakdown at the eighth order is definitively not an artifact of a restrictive algebraic ansatz. Rather, this lack of a solution for $\mathcal{W}^{(8)}$  implies that the assumption of an analytic functional $W[P]$ is violated beyond the sixth order of the Magnus expansion. The physical resolution is that the true solution must be a singular functional possessing a non-local dependence on the momentum loop $P$.

The theory of minimal surfaces provides a precise geometric precedent for this behavior: the Douglas functional for the area of a minimal surface bounded by a loop $C$ \cite{Douglas1931}. As established by Douglas, the area becomes a singular functional of the loop only after it is minimized over all possible boundary parametrizations.

Before this minimization, the area is simply a quadratic Dirichlet-Hilbert functional \ref{HilbertTrans} governed by the singular principal-value kernel $\mathcal{H}$. However, the strict minimization over parametrizations transforms it into a fundamentally singular functional because the optimal (isothermal) parametrization becomes implicitly and non-locally tied to the geometric shape of the loop. Any local change to the shape of the loop (for instance, tracking the normalized momentum $\tilde{P}' = P' / \int |P'| d\theta$) forces a global readjustment of this parametrization. This continuous global geometric shift modifies the analytic structure (though not the position) of the pole at $\theta=\theta'$ of the Hilbert kernel on the unit circle. Such shape-dependent, non-local singularities are fundamentally incompatible with a continuous, analytic Taylor-Magnus expansion in terms of iterated path integrals.

Our  continuum solution—the fermion determinant evaluated on the rigid Hodge-dual minimal surface—belongs to this class of singular geometric functionals. The required reparametrization invariance is dynamically enforced by the null-twistor (Virasoro) constraint, which acts as the  algebraic mechanism implementing the Douglas minimization over parametrizations. The algebraic inconsistency at $\mathcal{W}^{(8)}$ is thus the  mathematical signature of the string geometry asserting its non-local rigidity and justifying the transition from a 1D algebraic Master Field to the 4D continuous geometry of the rigid twistor string.

\section{The twistor parametrization of the loop space measure}
\label{sec:twistorParametrization}
The construction of the momentum loop path integral \eqref{WPdef} in planar QCD begins with a linear measure \eqref{linmeasure} over the space of closed loops. However, after gauge fixing the reparametrization symmetry ($\text{Diff}(S^1)$) via the Virasoro constraint, the correct measure on the reduced phase space becomes nonlinear. In this section, we derive the Jacobian $J = \sqrt{\det Q}$ arising from the twistor parametrization of the metric in loop space. In addition to this determinant, there is a Faddeev-Popov Jacobian arising from the gauge condition that fixes the residual local rescaling invariance of the twistor parametrization.

\begin{remark}{\textbf{Measure Transformation via Algebraic Jacobians.}}

As emphasized in Remark 1.2, we are not evaluating an abstract functional measure over a fluctuating 2D surface. Therefore, the measure we derive here does not encounter the ambiguities of Teichmüller space. Our starting point is the well defined, linear measure of the 1D quark-loop phase space: $\int \mathcal{D}C' \mathcal{D}P$. The transition to twistor space is  a local (in $\theta$) change of variables, mapping the bounding curve velocity $C^{\prime}(\theta) \propto \text{Im}(z\lambda\sigma\mu)|_{z=e^{i\theta}}$ directly to the boundary twistors $\lambda, \mu$. Because the twistors are localized at the boundary, this transformation involves no functional anomalies. Once the zero modes are fixed, we derive the corresponding twistor measure via standard, explicitly calculable algebraic matrix Jacobians.
\end{remark}

\subsection{Gauge Symmetry and the Need for Gauge Fixing}
The linear integration measure $\mathcal{D}C$ \eqref{linmeasure} is redundant due to the reparametrization invariance of the Wilson loop $W[C]$. To define the path integral without ambiguity, we must fix this gauge freedom and factor the measure by the volume of diffeomorphisms of the unit circle.

For the computation of the momentum loop $W[P]$, we choose the \textbf{conformal gauge}, where the loop parametrization is identified with the boundary value of a holomorphic null curve $f_\mu(z)$ in the bulk. The tangent vector of this complex curve factorizes into spinors (referred to as twistors in this context):
\EQ{
\label{fpimlambdamu}
f'_\mu(z) \sigma_\mu = \lambda(z) \otimes \mu(z)
}
The physical loop $C_\mu(\theta)$ is recovered as the real part of the boundary value of $2f_\mu(e^{\I\theta})$. The Virasoro constraint $(f')^2=0$ imposed by this ansatz does not restrict the \emph{shape} of the loop in $\mathbb{R}^4$; rather, it fixes the parametrization to be isothermal (where $|\dot{C}| = |H[\dot{C}]|$).

The original measure $\mathcal{D}C$ is invariant under the infinite-dimensional diffeomorphism group $\text{Diff}(S^1)$. Physically, this is a redundancy: different parametrizations describe the same geometric loop. To obtain a well-defined path integral, we must factor out this gauge volume:
\EQ{
\mathcal{D}C = \mathcal{D}C_{\text{phys}} \, \text{Vol}(\text{Diff}(S^1))
}
where $\mathcal{D}C_{\text{phys}}$ is the measure on the space of unparametrized loops. Imposing the Virasoro constraint is equivalent to choosing a specific gauge slice. In the twistor language, the physical degrees of freedom are encoded in the fields $(\lambda(z), \mu(z))$, modulo residual gauge symmetries.
\begin{remark}
    The scalar products between tilde spinors and spinors are defined as
\EQ{
&\tilde \lambda \lambda = \tilde \lambda_1 \lambda_1 + \tilde \lambda_2 \lambda_2;\\
&\tilde \mu \mu = \tilde \mu_1 \mu_1 + \tilde \mu_2 \mu_2;
}
These products are $O(4)$ -invariant, unlike products 
\EQ{
 v_\alpha = \lambda \sigma_\alpha \mu  = \sum_{a, \dot b} \lambda_{a} \sigma_\alpha^{a \dot b} \mu_{\dot b}
}
which are normal matrix-spinor products and transform as four- vectors. By design, these are null-vectors
\EQ{ 
(v_\alpha)^2 =0
}
\end{remark}
\subsection{The residual $U(1)$ gauge symmetry and normalized twistors}
A residual symmetry remains in the general twistor parametrization of the solution to the Virasoro constraint. The parametrization \eqref{fpimlambdamu} is invariant under the gauge transformation with a complex function $u(z)$:
\EQ{
\label{gaugeU1}
 \delta\lambda(z) = u(z) \lambda(z);\quad \delta\mu(z) = - u(z) \mu(z)
}
This transformation leaves the product $\lambda \otimes \mu$ (and thus the velocity $v$) invariant. To fix this redundancy, we impose the normalization constraint \textbf{ at the boundary}:
\EQ{
\label{NullTwistor}
 \lrb{\bar\lambda\lambda - \bar \mu\mu}_{|z|=1} =0
}
This constraint fixes the real part of the local rescaling (where $u(z)$ is real). Considering a variation $\delta r(z)$ corresponding to this real scaling:
\EQ{
&  \delta \lambda(z) = \delta r(z) \lambda(z);\quad \delta \mu(z) = - \delta r(z) \mu(z);\\
& \delta \lrb{\bar\lambda\lambda - \bar \mu\mu} = 2\delta r(z) \lrb{\bar\lambda\lambda + \bar \mu\mu}
}
This gauge constraint requires a Faddeev-Popov Jacobian. Since the transformation is local, the Jacobian is diagonal in the position basis, leading to an extra local factor in the measure:
\EQ{
\label{dOmegaFP}
d\Omega_{FP}(\lambda,\mu) &= \D\lambda\D\mu \, \delta\lrb{\bar\lambda\lambda - \bar \mu\mu}  \det{\lrb{\bar\lambda\lambda + \bar \mu\mu}}\br
&= \prod_{|z|=1} d \lambda d \mu \, \delta\lrb{\bar\lambda\lambda - \bar \mu\mu} \lrb{\bar\lambda\lambda + \bar \mu\mu}
}
\begin{remark}
    It is important to distinguish between the bulk variables on the disk and their boundary values. Our holomorphic twistor fields $\lambda(z),\mu(z)$ are defined on the whole disk, depending on complex variable $z,\; |z| \le 1$. The conjugate twistors $\bar\lambda(z),\bar\mu(z)$ depend on $\bar z$. The loop velocity $v = C'$ is the boundary value $v = 2 \Re \pd_{\theta}f(e^{\I \theta})$. Thus, the integration measure covers \emph{the boundary data} of the twistors, which locally parametrize the loop velocity. The bulk values are reconstructed via analyticity (Hilbert transform), so the path integral is defined  over the boundary measure.
\end{remark}

\subsection{Euclidean Self-Intersection vs. Minkowski Helicity}

In Euclidean signature, minimal surfaces of finite area bounded by a loop satisfy embedding theorems (e.g., the Meeks-Yau theorem) that preclude self-intersections. In Minkowski twistor space, however, the surface behaves like a helicoid. As the trajectory winds around the twistor pole $\lambda(z) \propto 1/(z-z_0)$ or $\mu(z) \propto 1/(z-z_0)$, it moves to a "different sheet" in the complexified coordinate space. This is why the surface can have an infinite area in the Minkowski sense—it is effectively a Riemann surface with an infinite number of sheets: the coordinate map function $ f_\mu(z) \propto \log(z-z_0)$.

This multiple-valued coordinate is what allows the poles to exist without the "self-intersection" inconsistency. The surface doesn't intersect itself for the same reason a spiral staircase doesn't intersect the floor below it: it has moved in an orthogonal dimension. The "infinite area" is simply the total area of all these sheets, and the monodromy is the "height" gained between each floor of the staircase, which we quantize as:
\EQ{\Delta S_{eff} = \int_{sheet_{n}}^{sheet_{n+1}} \mathcal{L} , d\tau = 2 \pi n}

In this complexified geometry, the 'infinite area' is not a divergence of the physical action, but the sum over the topological winding numbers that localizes the energy poles $E_n$ at the discrete values where the monodromy satisfies the phase coherence condition (see the Section \ref{sec:catastropheSpectrum}).

\subsection{Twistor-induced metric in loop space}
We treat the linear measure in loop space \eqref{linmeasure} as a measure in the space of velocities $v_\mu = C'_\mu$, subject to the periodicity constraint:
\EQ{
\D C = \delta^4\lrb{\oint d \theta v(\theta) } \prod_{\theta=0}^{2\pi} d^4 v(\theta)
}
The boundary condition provides a local relation between the loop velocity and twistor boundary values
\EQ{
v_{\alpha}(\theta) = 2 \Re \lrb{ \I z\lambda(z) \sigma_\alpha\mu(z) }_{z = e^{\I\theta}}
}
This projection induces a metric on the tangent space of the twistors. The norm of a variation in loop space is:
\EQ{
&\norm{\delta v}^2 = \oint d \theta  (\delta v_{\alpha})^2 \propto\br
&\oint d \theta \left|z\delta\lambda \sigma_\alpha\mu + z\lambda \sigma_\alpha\delta\mu - \text{c.c.}\right|^2
}
Expanding the square and using Fierz identities, this norm can be written in a quadratic form acting on the spinor variations $\delta \Lambda = (\delta\lambda, \delta \mu)^T$:
\EQ{
\norm{\delta v}^2 \propto \oint d \theta\, \delta \Lambda^\dag \cdot \hat Q \cdot\delta \Lambda
}
The kernel is the $4 \times 4$ block matrix (in spinor indices):
\EQ{
\label{Qmat}
\hat Q = \begin{pmatrix}
    (\bar\mu \mu) \1_2 &\, \lambda \bar{\mu} \\
    \mu \bar{\lambda} &\, (\bar\lambda \lambda) \1_2 \\
\end{pmatrix}
}
where $\lambda \bar{\mu}$ denotes the dyadic product (rank-1 matrix). To determine the measure, we analyze the eigenvalues of $\hat{Q}$. Let $\Lambda = \bar{\lambda}\lambda = \bar{\mu}\mu$ on the constraint surface. The eigenvalues are:
\begin{enumerate}
    \item $\omega_1 = 0$: The \textbf{zero mode}, corresponding to the gauge symmetry $\delta \lambda = \lambda, \delta \mu = -\mu$.
    \item $\omega_2 = 2\Lambda$: The dilatation mode $\delta \lambda = \lambda, \delta \mu = \mu$.
    \item $\omega_{3,4} = \Lambda$: Rotation modes orthogonal to the spinors.
\end{enumerate}
The determinant of $\hat{Q}$ vanishes due to the zero mode, which reflects the redundancy fixed by the constraint \eqref{NullTwistor}. Following the Faddeev-Popov procedure, the correct Jacobian is the square root of the determinant restricted to the physical subspace (orthogonal to the gauge orbit). This is given by the product of the non-zero eigenvalues:
\EQ{
\D v = \sqrt{\det{}' Q} \, d \Omega_{FP}(\lambda,\mu)
}
Evaluating the determinant on the physical subspace yields the factor:
\EQ{
\label{TwistorJacobian}
\sqrt{\det{}' Q}  = \sqrt{2 \Lambda^3} 
}

\subsection{Normalization of the Zero Mode and the Full Measure}
The Jacobian computed in \eqref{TwistorJacobian}, $\sqrt{\det{}' Q} = \sqrt{2}\Lambda^{3/2}$, represents the volume element of the physical configuration space (the subspace orthogonal to the gauge orbits). To obtain the correct measure for the full spinor space before dividing by the gauge volume, we must include the normalization of the zero mode itself.

The zero mode corresponds to the generator of the real rescaling symmetry. The variation along this gauge orbit is given by:
\EQ{
\delta \lambda = \delta r \lambda, \quad \delta \mu = - \delta r \mu
}
where $\delta r$ is the infinitesimal parameter of the transformation. The norm of this tangent vector in the flat spinor space is:
\EQ{
\norm{\delta_{\text{gauge}}}^2 &= |\delta \lambda|^2 + |\delta \mu|^2 \br
&= (\delta r)^2 \bar{\lambda}\lambda + (\delta r)^2 \bar{\mu}\mu
}
Imposing the constraint $\bar{\lambda}\lambda = \bar{\mu}\mu = \Lambda$, this simplifies to:
\EQ{
\norm{\delta_{\text{gauge}}}^2 = 2\Lambda (\delta r)^2
}
Thus, the metric factor (or "length") associated with the gauge direction is:
\EQ{
\label{GaugeNorm}
\sqrt{\det G_{\parallel}} = \sqrt{2\Lambda}
}
The total measure on the spinor space is the product of the measure on the physical subspace and the measure along the gauge orbit. Combining the physical Jacobian \eqref{TwistorJacobian} with the gauge normalization \eqref{GaugeNorm}, we recover the full scaling factor:
\EQ{
\label{TotalJacobian}
&J_{\text{total}} = \sqrt{\det{}' Q} \times \sqrt{\det G_{\parallel}} \br
&= \left(\sqrt{2}\Lambda^{3/2}\right) \times \left(\sqrt{2\Lambda}\right) = 2 \Lambda^2
}
This confirms that the effective measure element on the twistor space, prior to the Faddeev-Popov division of the gauge group volume, scales as $2\Lambda^2$. This factor will later cancel with another local scale factor coming from the momentum loop measure.

\subsection{Summary and Physical Interpretation}
To summarize:
\begin{itemize}
    \item The original linear measure $\mathcal{D}C$ is overparameterized due to reparameterization invariance.
    \item Gauge fixing via the Virasoro constraint and twistor parametrization reduces the integration to the physical moduli space.
    \item The physically correct path integral measure is defined by the Faddeev-Popov Jacobian for the null twistor constraint \eqref{NullTwistor}. This ensures a unique counting of twistor states.
    \item Additionally, the change of variables from velocity $C'(x)$ to twistors introduces the Jacobian \eqref{TotalJacobian}, representing the volume of the physical tangent space.
    \item While this measure covers the moduli space uniformly, there remains a local $U(1)$ gauge invariance: $(\lambda,\mu) \Ra (e^{\I\phi} \lambda, e^{-\I\phi} \mu) $. The corresponding volume $|U(1)| = 2\pi$  will be simply factored out of the measure for the normalized spinors $\xi,\eta, \bar \xi \xi = \bar \eta\eta =1$. These spinors vary on $S^3\times S^3$, therefore, the space becomes $S^3\times S^3/S^1$ after this final factorization.
\end{itemize}

\subsection{Local limit of $W[P]$ using effective action of Elfin theory}
With the non-singular functional $W[P]$, we can utilize the local limit of our continuum solution. Parametrizing the loop $C$ by twistors $\lambda, \mu$, we compute the Fourier integral for $W[\hat P]$ (with $\hat P =P_\alpha \sigma^\dag_\alpha$):
\EQ{
\label{WPlambdamu}
&W[P] = \int \D \lambda \D \mu \, 2\Lambda^3 \, \delta\lrb{\bar \lambda \lambda - \bar \mu \mu}\exp{ - \int d^2 z L(z,\bar z)}\br
&\times \int d^4 k\exp{-2\I \oint \dd \theta  \Im\lrb{ z \lambda (\hat P + \hat k)\mu}}_{z = e^{\I \theta}}
}
where the combined measure factor $2\Lambda^2 \cdot \Lambda \sim \Lambda^3$ (depending on normalization conventions) is absorbed into the effective Lagrangian:
\EQ{
\label{Leffect}
L(z,\bar z) = \frac{\sigma}{2} \Lambda^2 + \frac{1}{12\pi} \left | \pd_z\lrb{\log \Lambda} \right|^2
}
The extra integration over the "zero mode" $k_\mu$ reproduces the delta function $\delta(\oint v d \theta)$ for the periodicity of the loop ($C(\8) = C(-\8)$ in the upper plane map). This results in a nonlinear theory involving spinor fields $\lambda(z), \mu(z)$ holomorphic inside the unit circle.

\section{From null twistor to Liouville coupled to $S^1\mapsto (S^3\otimes S^3)$ sigma model}
\label{sec:nulltwistorliouville}
\subsection{Polar Coordinate Parametrization of the Null Twistor Measure}

Consider the twistor variables $(\lambda_\alpha, \mu^{\dot\alpha})$ on a unit circle $|z| =1$, subject to the null constraint:
\EQ{
\bar\lambda^\alpha \lambda_\alpha = \bar\mu_{\dot\alpha} \mu^{\dot\alpha}
}
We parametrize the spinors in polar coordinates as:
\EQ{
&\lambda_\alpha = u\, \xi_\alpha, \br
&\mu^{\dot\alpha} = v\, \eta^{\dot\alpha}
}
where $u, v \in \mathbb{R}_+$ are positive real scales, and $\xi, \eta$ are normalized spinors:
\EQ{
\bar\xi^\alpha \xi_\alpha = \bar\eta_{\dot\alpha} \eta^{\dot\alpha} = 1
}
The constraint implies $u=v$. The measure on the null twistor space in these coordinates transforms as:
\EQ{
\D\lambda\, \D\mu\, \delta(\bar\lambda \lambda - \bar\mu \mu) \propto u^5\, du\, d\Omega_\xi\, d\Omega_\eta
}
where $d\Omega_\xi$ and $d\Omega_\eta$ are the Haar measures on $S^3$.

There is one subtle issue left to settle. The local gauge $U(1)$ transformation \eqref{gaugeU1} leaves the physical variables invariant. This $U(1)$ subgroup acts diagonally on the spinors and should be factored out, defining the angular manifold as the coset space:
\EQ{
&\xi, \eta \in \frac{S^3 \times S^3}{U(1)};\\
& d \Omega_{\xi\eta} = \frac{d\Omega_\xi d\Omega_\eta}{2\pi}
}

\subsection{Reduction to $u$, $\xi$, $\eta$ Variables}

The path integral over null twistors reduces to:
\EQ{
&\int \D\lambda\, \D\mu\, \delta(\bar\lambda \lambda - \bar\mu \mu)\, F(\lambda, \mu) \br
&= \int_0^\infty u^5 du \int_{(S^3 \times S^3)/U(1)} d\Omega_{\xi\eta}\, F(u\xi, u\eta)
}
This makes explicit the separation between the radial scale $u$ and the angular geometry.

\subsection{Transformation to Liouville Field}

We define the Liouville field $\rho$ by the logarithmic map:
\EQ{
\rho = 2\log u \qquad \Rightarrow \qquad u = e^{\rho/2}
}
The radial measure transforms as:
\EQ{
u^5 du = e^{5\rho/2} (e^{\rho/2} d\rho) = e^{3\rho} d\rho
}
Consequently, the path integral measure becomes:
\EQ{
&\int \D\lambda\, \D\mu\, \dots = \int_{-\infty}^{\infty} d\rho\, e^{3\rho} \int d\Omega_{\xi\eta}\, \dots
}
\subsection{Phase space path integral for quark loop amplitudes}

The physical amplitudes in \eqref{Amplitude} require one final functional integration: the path integral over the momentum field $P_\mu(\theta)$ (defined on the boundary).

The parametric invariance $\theta \to \phi\theta)$ applies to the complete phase space integral. This invariance is broken only by the trace of the ordered path exponential of the Dirac operator. To restore it fully in the effective theory, we introduce an einbein $e(\theta) > 0$ multiplying the Dirac operator:
\EQ{
&\D e \D P \tr \mathcal{T} \exp{- \oint \dd \theta \, e(\theta) (\I \gamma_\mu P_\mu(\theta) + m_q)} \br
&= \int \prod_\theta d^4 P(\theta) de(\theta) \exp{- \oint \dd \theta \, e(\theta) (\I \gamma P(\theta) + m_q)}
}
While the conventional gauge choice is $e(\theta) = \text{const}$, we have already fixed the reparametrization gauge via the Virasoro constraint on the holomorphic vector $f_\mu(z)$. Since the reparametrization acts simultaneously on $P$ and $C$, and we have factored out the volume of $\text{Diff}(S^1)$, there are no more gauge fixing to be done. We are left with integration over the arbitrary positive local einbein field $e(\theta)$.

This integration yields the inverse matrix structure. 
\EQ{
\int_0^\8 \dd e \, \exp{ - e \, d\theta \, (\I \gamma_\mu P_\mu + m_q) } \propto \frac{1}{\I \gamma_\mu P_\mu(\theta) + m_q}
}
Taking the limit of zero quark mass (see the remark \ref{quarkMass} for the definition of our effective quark mass) and focusing on the local factor dependent on $P$, we compute the integral over the local momentum variable $q_\mu(\theta) = P_\mu(\theta) d \theta$:
\EQ{
& I(\tau) = \int d^4 q \, \frac{\exp{\I q_\alpha \tau_\alpha}}{\I \gamma \cdot q}
}
where the source term is given by the twistor bilinear:
\EQ{
\tau_\alpha = \Re \lrb{ \I z \xi \sigma_\alpha \eta} e^{\rho}
}
This Fourier integral is calculable. Up to normalization, we obtain:
\EQ{
 I(\tau) \propto \frac{\I \gamma_\mu \tau_\mu}{(\tau^2)^2}
}
Using the normalization of the spinors $\xi, \eta$, and the Fiertz identity, the square of the source vector simplifies:
\EQ{
\tau^2 &= \left( \Re \lrb{ \I z\xi \sigma_\alpha \eta} \right)^2 e^{2\rho} \br
&= \frac{-1}{4} \left(z \xi \sigma \eta - \bar z\overline{\xi \sigma \eta} \right)^2 e^{2\rho} = \frac{1}{4} (4) e^{2\rho} = e^{2\rho}
}
Substituting this back into the result:
\EQ{
 I(\tau) \propto \frac{\I \gamma_\alpha \Re \lrb{ \xi \sigma_\alpha \eta} e^{\rho}}{(e^{2\rho})^2} = \exp{- 3 \rho} \left[ \I \gamma_\alpha \Im\lrb{ z\xi \sigma_\alpha \eta} \right]
}
Remarkably, the factor $e^{-3\rho}$ from the momentum integration precisely cancels the measure factor $e^{3\rho}$ derived from the twistor transformation, rendering the effective local volume element scale-invariant.

\begin{remark}[Scale Invariance of the Phase Space Measure]
The  cancelation of the Liouville factors, $e^{3\rho} \cdot e^{-3\rho} = 1$, is a necessary consequence of the symplectic geometry of the loop phase space. The canonical path integral measure $\mathcal{D}P \mathcal{D}C$ is invariant under canonical transformations, including the scaling dilatation $C_\mu \to s C_\mu$ and $P_\mu \to s^{-1} P_\mu$.
In our construction, the twistor transformation introduces a specific scaling of the coordinate velocity $v \sim e^{\rho}$. The Jacobian of this transformation ($e^{3\rho}$) represents the compression of the coordinate measure. However, the integration over the conjugate momentum $P_\mu$ acts as an inverse scaling operator (a Fourier transform), generating the counter-factor $e^{-3\rho}$. The resulting scale invariance of the effective volume element ensures that the theory remains conformal at the classical level, consistent with the dimensionless nature of the phase space volume element $d P \wedge d C \sim \hbar$.
\end{remark}
\subsection{The Holographic Liouville Theory}

The effective action now reduces to a Liouville-type term for the field $\rho$:
\EQ{
\label{Liouville}
&S_{\text{Liouville}} = \int d^2z\, \left( \frac{1}{3\pi} |\partial_z \rho|^2 + \sigma e^{2\rho}  \right);\\
& \rho(z,\bar z) = \frac{\log( \bar \lambda(\bar z) \lambda(z)) + \log( \bar \mu(\bar z) \mu(z))}{2};
}
and the amplitude is proportional to
\EQ{
&\D \Omega_{FP}[\lambda,\mu]\exp{-S_{\text{Liouville}}} \br
&\exp{-2\I\sum_k q^\alpha_k \Im \oint\dd \theta  \hat Q_k(\theta)  z\lambda \sigma_\alpha \mu} ;\\
& Q_k(\theta) = \lrb{ \Theta(\theta-\alpha_k)-\Theta(\theta-\alpha_{k+1})}
}
The closure of the loop: $\oint d \theta v_\mu =0$ is satisfied identically in our parametrization
\EQ{
\oint\dd \theta  z\lambda \sigma_\alpha\mu =  \Im \oint \dd z \lambda \sigma_\alpha\mu =0
}
This contour integral over a unit circle of a holomorphic function vanishes by the Cauchy theorem (no singularities inside the circle).

It is important to distinguish this construction from the conventional Liouville field theory encountered in 2D quantum gravity. Here, the Liouville field $\rho(z, \bar{z})$ in the bulk is not an independent fluctuating quantum field; rather, it is rigidly determined by the boundary data via the analytic continuation of the underlying holomorphic twistors $\lambda(z), \mu(z)$. Thus, this is \textbf{not} the standard dynamical Liouville theory. Instead, it represents a hidden one-dimensional theory of twistors on the boundary, disguised as a 2D Liouville theory—a simpler, albeit less familiar, mathematical object.

In our framework, the path integral measure over loop space is constructed through a systematic procedure: we begin with a linear measure, impose the Virasoro constraint via twistor parametrization, and fix the residual local scaling symmetry. The resulting measure is manifestly real and positive. Crucially, this measure does not contain an exponential of the KKS symplectic form—a hallmark of conventional string quantization and Conformal Field Theory (CFT). We emphasize that our approach does \emph{not} involve quantizing a fundamental string in the Nambu-Goto sense, but rather constitutes a faithful implementation of the solution to the MM loop equations in momentum loop space.

\subsection{Block Matrix Structure and Diagonalization via Staggered Spinors}

The effective action involves the path-ordered product of the matrix-valued velocity operator along the loop. The fundamental link variable is the matrix $\mathcal{M}(x) = i \not{v}(x) = i \gamma_\alpha v_\alpha(x)$. In the chiral (Weyl) representation, the Dirac matrices have the off-diagonal structure:
\EQ{
\gamma_\alpha = \begin{pmatrix} 0 & \sigma_\alpha \\ \bar{\sigma}_\alpha & 0 \end{pmatrix}
}
where $\sigma_\alpha = (1, \vec{\sigma})$ and $\bar{\sigma}_\alpha = (1, -\vec{\sigma})$ in Euclidean signature (or $\sigma_\alpha$ are the standard Pauli matrices). The velocity vector is given explicitly by the projection of the loop momentum onto the spinor dyadics:
\EQ{
v_\alpha = i z (\xi \sigma_\alpha \eta) - i \bar{z} (\eta^\dagger \sigma_\alpha \xi)
}
Substituting this into the link matrix $\mathcal{M}$ and using the Fierz completeness relation $(\psi^\dagger \sigma_\alpha \chi) \sigma^\alpha = 2 \chi \psi^\dagger$, the operator takes the specific block form:
\EQ{
\mathcal{M} = i \begin{pmatrix} 
0 & \mathcal{P} \\ 
\mathcal{P} & 0 
\end{pmatrix}
}
where $\mathcal{P}$ is the Hermitian operator on the 2-component spinor space, explicitly dependent on the loop coordinate $z = e^{i\theta}$:
\EQ{
\label{Projectorlambdamu}
\mathcal{P}(\theta) = 2\I \left[ z \eta \xi^\dagger - \bar{z} \xi \eta^\dagger \right]_{z = e^{\I\theta}}
}
This off-diagonal structure implies a ``chirality hopping'' mechanism: the operator maps a left-handed spinor at point $\theta$ to a right-handed spinor at $\theta+d\theta$, and vice versa. To compute the trace of the ordered product efficiently, we employ a staggered spinor transformation (analogous to the Kawamoto-Smit transformation in lattice gauge theory).

We discretize the loop into $2N$ points and define a transformed spinor basis $\Psi'_n$ related to the original basis $\Psi_n$ by a site-dependent rotation $\Omega_n$:
\EQ{
\Psi_n = \Omega_n \Psi'_n, \quad \Omega_n = \begin{cases} \1_4 & \text{for even } n \\ \gamma_0 & \text{for odd } n \end{cases}
}
where $\gamma_0 = \begin{pmatrix} 0 & \1_2 \\ \1_2 & 0 \end{pmatrix}$ is the chirality-flipping matrix. The effective link matrix between sites $n$ and $n+1$ transforms as $\tilde{\mathcal{M}}_{n} = \Omega_{n+1}^\dagger \mathcal{M}_n \Omega_n$.

Due to the bipartite nature of the 1D loop, there are two alternating cases:
\begin{enumerate}
    \item \textbf{Even to Odd ($n \to n+1$):}
    \EQ{
    \tilde{\mathcal{M}} = \gamma_0^\dagger (i \not{v}) \1 = \begin{pmatrix} 0 & \1 \\ \1 & 0 \end{pmatrix} i \begin{pmatrix} 0 & \mathcal{P} \\ \mathcal{P} & 0 \end{pmatrix} = i \begin{pmatrix} \mathcal{P} & 0 \\ 0 & \mathcal{P} \end{pmatrix}
    }
    \item \textbf{Odd to Even ($n \to n+1$):}
    \EQ{
    \tilde{\mathcal{M}} = \1^\dagger (i \not{v}) \gamma_0 = i \begin{pmatrix} 0 & \mathcal{P} \\ \mathcal{P} & 0 \end{pmatrix} \begin{pmatrix} 0 & \1 \\ \1 & 0 \end{pmatrix} = i \begin{pmatrix} \mathcal{P} & 0 \\ 0 & \mathcal{P} \end{pmatrix}
    }
\end{enumerate}
The transformation completely diagonalizes the operator in the Dirac space, decoupling the 4-component spinor into two identical 2-component sectors.

The trace of the path-ordered product along the closed loop splits into two equal $2\times 2$ traces:
\EQ{
\label{TraceProjectors}
T =\tr_4 \prod_{\theta} (\I \not{v}(\theta)) = 2 \tr_2 \prod_{\theta} (i \mathcal{P}(\theta))
}
This reduction from $4\times 4$ to $2\times 2$ matrices significantly reduces the computational complexity for numerical evaluation. The trace of the remaining product of our projection operators is real (by the charge conjugation symmetry of Dirac traces), but it is not positive definite, and in general, it oscillates. It is this oscillatory nature of the path integral that mandates its  analytical evaluation via Picard-Lefschetz theory and Catastrophe Theory.

\section{Twistor Holography and the Rigid Twistor String}
\label{Twistorholography}
\subsection{Twistor Holography (Gauge Holography)}
\label{sec:twistorholography}
In our framework, \textbf{twistor holography} refers to the geometric duality in which the solution to the Makeenko–Migdal loop equation is realized as a rigid, Hodge-dual minimal surface embedded in flat $\mathbb{R}^3 \times \mathbb{R}^4$ (or, more generally, a flat 12-dimensional space). This surface is uniquely determined by the boundary Wilson loop and does not fluctuate quantum mechanically. The holographic principle here projects the internal geometry of the 12D target space onto the observable boundary, with the physical Wilson loop acting as the gauge-invariant observable. Unlike AdS/CFT gravity holography, this construction is entirely geometric and does not involve a gravitational bulk or fluctuating worldsheet metrics.

\subsection{Analytic Twistor String vs. Topological Twistor String}
The use of spinor variables and null twistors in this framework is closely related to the Penrose transform, which encodes solutions to spacetime field equations in terms of cohomological data on twistor space. By imposing the specific gauge condition on boundary values of twistor map: $\bar{\lambda}\lambda = \bar{\mu}\mu$, we reduce the theory to a Liouville field interacting with a one-dimensional sigma model on $S^3 \times S^3 / S^1$ on the boundary. The bulk values of the Liouville field depend on the twistor map. This parametrization makes the structure transparent: the minimal surface in the bulk is entirely encoded by the holomorphic data (the twistors), while the physical observables are determined by the boundary curve.

A crucial distinction must be made regarding the nature of the resulting bulk action. In the context of \textbf{Witten's celebrated Twistor String theory} (and related formulations such as the Berkovits model), the extension of boundary fields into the bulk typically generates a Wess-Zumino-Witten (WZW) term. The WZW functional is topological: it measures a winding number and is invariant under smooth deformations of the extension inside the disk. 

In contrast, our bulk functional is \textit{not} topological. While it is determined by the boundary data, its value depends explicitly on the unique analytic continuation of the holomorphic twistors into the disk. In this sense, our framework is closer in spirit to Penrose's \textbf{Non-Linear Graviton} construction than to topological string theory. In the Non-Linear Graviton, the curved geometry of the bulk spacetime is encoded by the global complex structure of twistor space, which arises as the unique holomorphic extension compatible with boundary (patching) data. Similarly, our action represents the geometric energy (akin to a Kähler potential) of the specific holomorphic disk filling the boundary loop. This non-topological dependence is essential: it encodes the geometric rigidity of the minimal surface, leading to the emergence of the area law and confinement in the rigid twistor string.

Because the path integral over the twistor target space is highly oscillatory, one might assume that extracting the physical spectrum requires massive stochastic lattice simulations. However, by analyzing the holomorphic structure of the effective action, we find that the discrete mass spectrum can be extracted analytically. Our solution to Planar QCD is not a theory of quantum fluctuations in any space-time dimensions; it is a problem of classical  geometry.

\section{Analytic continuation of Twistor String Path integral}\label{sec:ComplexIntegral}

The Euclidean path integral of Twistor String theory, as formulated so far, is similar to a partition function of a statistical system with positive definite energy , given by the Liouville bulk integral. There are, however, some pre-exponential factors that violate this statistical interpretation. First, there is a "vertex operator" factor $$ \exp{-2\I\sum_k q^\alpha_k \Im \oint\dd \theta  \hat Q_k(\theta)  z\lambda \sigma_\alpha \mu}.$$ Second, there is a trace factor \eqref{TraceProjectors}. Both of these factors are not positive, and they oscillate as functionals of our dynamical variables: i.e., the complex curve $\lambda(z),\mu(z)$ in the Twistor space $\mathcal CP^3$. These oscillations become even more obvious when the path integral is analytically continued in  Minkowski space, where the effective action is imaginary, and the exponential decay of the statistical weight becomes a quantum oscillation of the Feynman-Dirac weight $\exp{\I S/\bar h}$.

In ordinary quantum mechanics, this phenomenon is well known: the path integral, when formulated in terms of physical, real trajectories in phase space, is exponentially small as a result of the interference of the classical trajectories in the Feynman path integral. However, when reformulated as a complex path integral by analytic continuation, it is dominated by complex saddle points, where there is no interference along the steepest descent path. The spectrum of certain quantum mechanical systems can be exactly computed by finding the dominant complex trajectory winding around the complex branch point of the analytically continued effective Action. This is a mechanism for the conversion of the Bohr-Sommerfeld quantization condition from the WKB approximation to an exact equation for the spectrum \cite{Dashen1974, Gutzwiller1971, Duistermaat1982}.

We are looking for a similar phenomenon in our Twistor String theory, and the first step towards such a solution is the analytic continuation of the path integral into the complex plane in parameter space.

\subsection{Complexification of the parameter space}
To ensure the strict analyticity of the path integral, we must remove complex conjugation from the effective action and its parameter space. We achieve this by complexifying the target manifold and analytically continuing the action integral itself. This ensures that the drift forces are partial derivatives of a globally holomorphic function. Then the path integral can be regarded as a complex integral in multi-dimensional space of holomorphic functions  in Twistor space.

We replace the original real and Hermitian variables with a set of independent holomorphic variables. The angular coordinates $\alpha_k$ are mapped to the complex plane via $z_1 = e^{i\alpha_1}, \; z_2 = e^{i\alpha_2}$. The spinor fields and their conjugates are treated as independent complex vectors. The full set of dynamical variables for the discretized loop is:
\EQ{
    \mathcal{V} = \{\{z_1, z_2\},\{ \rho_k, \xi_k, \eta_k, \tilde{\xi}_k, \tilde{\eta}_k, \gamma_{1,k}, \gamma_{2,k} \}_{k=1}^N \}.
}
Here,  $\tilde{\xi}$ and $\tilde{\eta}$ replace $\bar{\xi}$ and $\bar{\eta}$ as independent variables. The geometric constraints are enforced by the complex Lagrange multipliers $\gamma_{1,k}$ and $\gamma_{2,k}$, leading to the holomorphic constraint terms in the effective action:
\EQ{
    S_{\text{constr}} = i \sum_k \left[ \gamma_{1,k}(\tilde{\xi}_k \xi_k - 1) + \gamma_{2,k}(\tilde{\eta}_k \eta_k - 1) \right].
}

\subsection{Spectral Interpolation of Fields}

To evaluate the action integrals analytically, we construct global holomorphic polynomials that interpolate the values at the vertices $z_k$.

\subsubsection{Spinor Fields}
Let $\Lambda^{(k)} = e^{\rho_k/2}(\xi_k, \eta_k)$ denote the spinor data  $\Lambda(z) = \{\lambda(z), \mu(z)\}$ at vertex $k$. The Lagrange interpolation polynomial on the spectral unit circle is:
\EQ{ \label{Interpol}
    &\Lambda(z) = \frac{1}{N} \sum_{k=0}^{N-1} \Lambda^{(k)} K\lrb{z e^{-\I \theta_k}};\\
    & K(z) =\frac{1 - z ^N}{1 - z };
}
Using the Fast Fourier Transform (FFT), we extract the spectral coefficients $L_n, M_n$ (and $\tilde{L}_n, \tilde{M}_n$ from the independent $\tilde{\xi}, \tilde{\eta}$ variables) such that:
\EQ{
   & \lambda(z) = \sum_{n=0}^{N-1} L_n z^n, \quad \tilde{\lambda}(z) = \sum_{n=0}^{N-1} \tilde{L}_n z^n.\\
   &  \{L_n,M_n\} = \frac{1}{N}\sum_{k=0}^{N-1} \Lambda^{(k)}e^{-\I \theta_k}
}

\begin{remark}{\textbf{Reduction to Pole Moduli and the Algebraic Jacobian}}

While the path integral is formulated using the discrete boundary values at the roots of unity $\theta_k = 2\pi k/N$, the physical saddle points and their monodromies are governed by the meromorphic structure of the twistors—specifically, the positions $z_i$ and residues $R_i$ of their poles inside the unit disk.The transformation of the integration measure from the boundary values to these pole moduli is exact and purely algebraic. For a twistor field parameterized by a finite number of poles, its evaluation on the boundary circle $z = e^{i\theta}$ yields a geometric series:
\EQ{\lambda(e^{i\theta}) = \sum_i \frac{R_i}{e^{\I\theta} - z_i} = - \sum_i \frac{R_i}{z_i} \sum_{n=0}^\infty z_i^{-n} e^{\I n\theta}}
Comparing this to our spectral expansion, the Fourier coefficients $L_n$ are determined as algebraic superpositions of the pole parameters:
\EQ{L_n = - \sum_i \frac{R_i}{z_i^{n+1}}}
Consequently, transforming the functional measure from the $N$ discrete boundary values to the physical moduli $(R_i, z_i)$ involves two finite, calculable matrix Jacobians:1. The constant Jacobian of the Discrete Fourier Transform matrix (the Fourier matrix $e^{2 \pi \I k n/N}$), which maps the boundary values to the Fourier coefficients $L_n$.2. The purely algebraic Jacobian mapping the Fourier coefficients to the pole parameters $(R_i, z_i)$.We do not explicitly compute these algebraic Jacobians here, because they contribute  to the pre-exponential fluctuation determinant of the path integral. For the extraction of the geometric mass spectrum, the relevant physics is governed entirely by the stationary phase (the saddle-point equations) and the topological monodromies of the effective action. As we will demonstrate, the S-matrix poles emerge from the constructive interference of the geometric series over the topological winding number $k$. In this resonant limit, the multi-valued action scales as $k \Delta S_{eff}$. Consequently, the transverse quantum fluctuations around the saddle point are dynamically damped by a factor of $1/k$. This  exponential part dominates over the pre-exponential factor, rendering the explicit form of the Jacobian irrelevant to the exact localization of the S-matrix poles. However, the explicit existence of this algebraic mapping mathematically guaranties that our twistor measure is finite and constructible, firmly rooting the evaluation of the path integral in standard finite-dimensional complex analysis rather than abstract string-theoretic functional geometry.
\end{remark}

\subsubsection{Liouville Field}
The Liouville field $\rho(z,\bar z)$ is extended inside the unit circle via an  formula
\EQ{
\rho(z,\bar z) = \frac{\log \tilde \lambda(\bar z) \lambda(z) + \log \tilde \mu(\bar z) \mu(z) }{2}
}
where ${\lambda(z),\mu(z)}$ is extended by interpolation \eqref{Interpol}.
The derivatives of the above kernel $K'(z)$ involved in the Liouville Lagrangian with this $\rho$ were computed in the previous sections.

These expressions provide the values of the fields and derivatives at any complex $z$, not just at the vertices. These expressions also depend upon complex parameters in $\mathcal V$ through the Fourier coefficients $\Lambda(k), \tilde \Lambda(k)$. This dependence is exponential for $\rho_k$ and polynomial for the rest of variables in $\mathcal V$, which allows for required analytic continuation of the CLE.

\subsection{The Boundary Terms in the action}

The boundary part of the action involves integrals of the form $\Im \int (\lambda \times \mu) e^{i\theta} d\theta$. Substituting the spectral expansions, the integrand becomes a sum of monomials $z^n z^m \cdot z = z^{n+m+1}$. We define the elementary holomorphic integral function:
\EQ{
    G_{p}(z_1, z_2) \equiv \int_{z_1}^{z_2} z^{p} \frac{dz}{iz} = \frac{z_2^{p} - z_1^{p}}{i p}.
}
With the restored measure factor, the power is $p = n+m+1$. Since $n, m \ge 0$, the denominator $n+m+1 \ge 1$ is  positive, ensuring there are no singularities in the complex plane.

The contribution to the action of external momentum $q$ injected into the quark loop from the arc between $z_1$ and $z_{2}$ is:
\EQ{
\label{Sq}
    &S_{q} = -\I q_\alpha\br
    & \sum_{n=0}^{N-1} \sum_{m=0}^{N-1} \left[ (L_n \sigma_\alpha M_m) - (\tilde{L}_n \sigma^\dag_\alpha \tilde{M}_m) \right] G_{n+m+1}(z_1, z_{2}).
}
Then, there are terms related to the constraints on the $\xi, \eta $ variables. These terms are just discrete sums
\EQ{
S_{\gamma}= \I\sum_{k=0}^{N-1} \gamma_{1,k} \lrb{\tilde \xi_k \xi_k -1} + \gamma_{2,k} \lrb{\tilde \eta_k \eta_k -1}
}

In addition, there are terms related to the trace \eqref{TraceProjectors}. Our complexification reduces this trace to
\EQ{
& S_{T} = -\log T(\xi,\eta,\tilde \xi, \tilde \eta);\\
&T(\xi,\eta,\tilde \xi, \tilde \eta) = \tr \prod_{k=0}^{N-1}\I \mathcal P(\theta_k);\\
& \mathcal{P}(\theta) = 2\I \left[ z \eta \tilde \xi^T - \bar{z} \xi \eta^T \right]_{z = e^{\I\theta}}
}

Finally, our reweighting adds one more complex term to the action, 
\EQ{
&S_{O} = \log \mathcal{O};\\
&\mathcal{O} = \pbyp{S_{q}}{q_4}= \br
& \sum_{n=0}^{N-1} \sum_{m=0}^{N-1} \left[ (L_n  M_m) - (\tilde{L}_n \tilde{M}_m) \right] G_{n+m+1}(z_1, z_{2}).
}

All the above expressions are  purely holomorphic functions of all the variables in $\mathcal{V}$.

\subsection{Stochastic Quantization and Complex Langevin Dynamics}

The  mathematical evaluation of this highly oscillatory complexified path integral is achieved via Stochastic Quantization. The quantum statistical evolution of the twistor string is governed by the Complex Langevin Equation (CLE) with holomorphic drift forces derived from the total spectral action $S_{total} = S_q + S_\gamma + S_T + S_\rho + S_{\mathcal{O}}$. We emphasize that the CLE is utilized here not as a numerical simulation technique, but as an  continuous mathematical reformulation of the quantum path integral.

\begin{enumerate}
    \item \textbf{Coordinate Evolution ($z_k$):} The drift force on the boundary points $z_k$ arises from the boundary terms of the $G$ functions in $S_q$ and the integration limits in $S_{bulk}$:
    \EQ{
        \frac{d z_k}{d\tau} = - z_k^2 \frac{\partial S_{total}}{\partial z_k} + \I z_k \nu_{z_k}
    }
    (The factor $z_k^2$ arises from the metric of the map $\alpha = -\I \log z$, and $\nu$ represents complex Gaussian white noise).

    \item \textbf{Spinor Evolution ($\xi_k, \tilde{\xi}_k$, etc.):} The drift forces are computed via the chain rule through the spectral coefficients $L_n$. For example:
    \EQ{
        \frac{d \xi_k}{d\tau} = - \frac{\partial S_{total}}{\partial \xi_k} - \I \gamma_{1,k} \tilde{\xi}_k + \nu_{\xi_k}
    }
    Crucially, the noise terms $\nu_{\xi}$ and $\nu_{\tilde{\xi}}$ are independent complex Gaussian variables representing the  quantum fluctuations of the twistor fields.

    \item \textbf{Liouville Evolution ($\rho_k$):} The force on $\rho_k$ is nonlocal, mediated by the interpolation kernel $K(z)$ appearing in both the bulk quadrature points and the spectral coefficients of the boundary term:
    \EQ{
        \frac{d \rho_k}{d\tau} = - \frac{\partial S_{total}}{\partial \rho_k} + \nu_{\rho_k}
    }

    \item \textbf{Constraint Evolution ($\gamma_k$):} The multipliers evolve to enforce the  algebraic constraints:
    \EQ{
        \frac{d \gamma_{1,k}}{d\tau} = -\I (\tilde{\xi}_k \xi_k - 1) + \nu_{\gamma_{1,k}}
    }
\end{enumerate}

In this mathematical framework, the deterministic drift terms ($-\nabla S$) represent the  holomorphic gradient flow in complexified phase space. In Picard-Lefschetz theory, the union of \textit{all} these deterministic flow lines originating from a complex saddle point constructs a middle-dimensional real manifold known as a Lefschetz thimble. If the noise $\nu$ were zero, the system would trivially collapse onto a 1D classical trajectory (the WKB approximation). However, the multidimensional stochastic noise $\nu$ dynamically forces the probability distribution to completely explore the transverse width of these $N$-dimensional thimbles,  incorporating all  quantum fluctuations (the ``vicinity'' of the saddle points). 

The phenomenon of exact S-matrix quantization we are looking for corresponds to a topological bifurcation of this stochastic system. At certain critical resonance energies, the complex Hessian develops a null space and the classical restoring drift  vanishes along a continuous topological valley. Unimpeded by any restoring gradient force, the stochastic noise $\nu$ drives the complex trajectory to infinity along this flat zero mode. This infinite flat-valley integration geometrically generates the exact macroscopic S-matrix pole of the path integral, governed by Catastrophe Theory. Before studying this phenomenon in general form, we consider the simple solvable limit: the WKB quantization corresponding to the large rotational momentum of our rigid string.

\section{Catastrophe Theory and the Geometric Mass Spectrum}\label{sec:catastropheSpectrum}
\subsection{Twisted Boundary Conditions, Symmetrization, and Spin Projections}

Before solving the  spectral equations, we justify the choice of the
twisted geometry from the perspective of the semiclassical quark path integral.
In Planar QCD, the effective Minkowski action involves a sum
over quark trajectories weighted by the area of the confining minimal surface.

To extract the spectrum of states with macroscopic angular momentum $J$, we
impose twisted boundary conditions. We require the whole string to undergo 
a physical spatial rotation by an angle $\alpha$ over
the period of the motion, with the angular momentum $J$ acting as the
thermodynamic conjugate to this twist. The minimization of the boundary length
for a rotating particle yields a helical trajectory in Minkowski spacetime.
The minimal surface spanning two twisted helical boundaries is the helicoid,
as illustrated in Figure \ref{fig:helicoid}.

\begin{figure}[htbp]
\centering
\includegraphics[width=0.8\columnwidth]{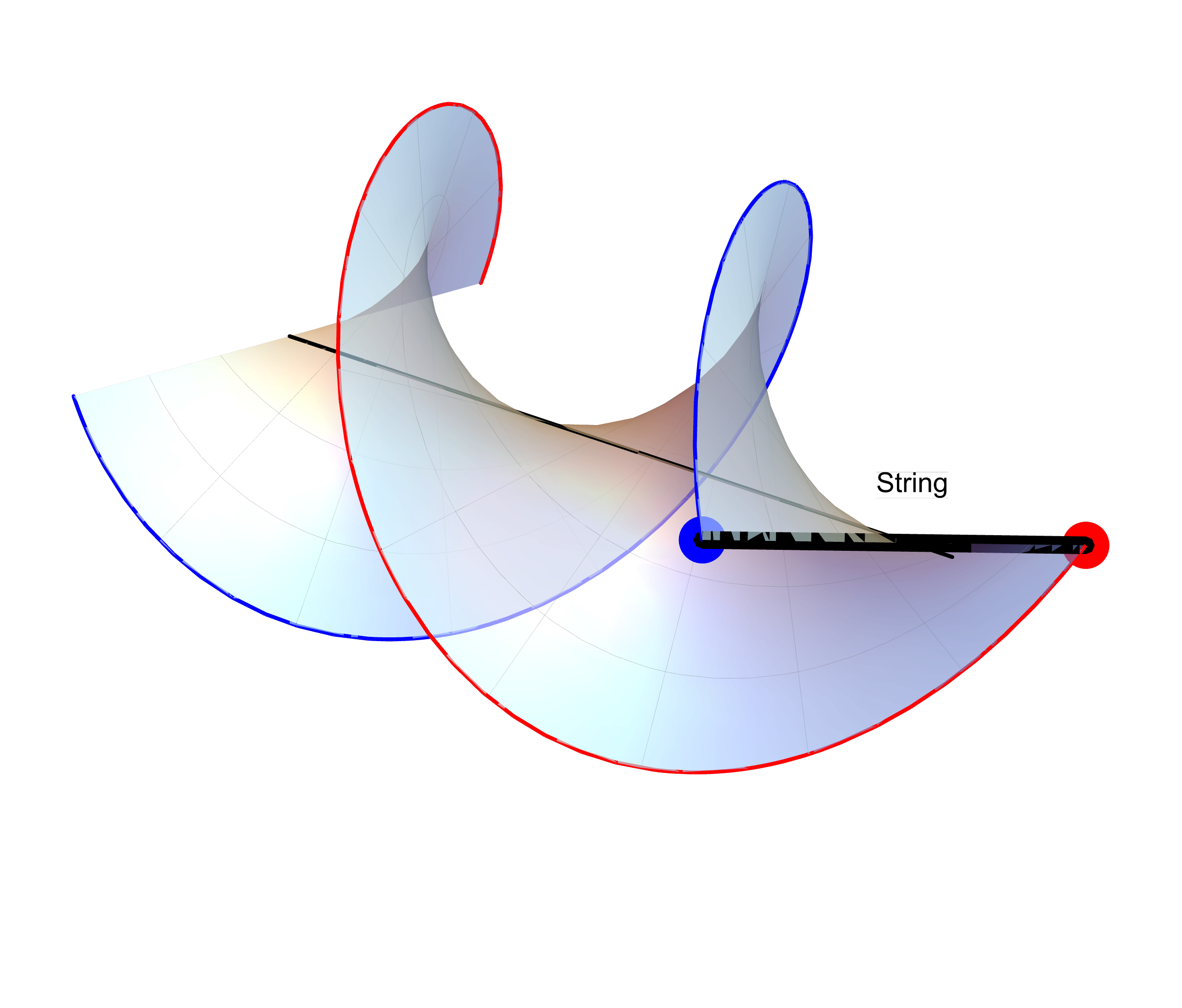} % Adjust filename if needed
\caption{The helicoid spanned by a rotating $\overline{q}q$ pair connected by a
rigid stick (string). This minimal surface bounded by a double helix, discovered
by Meusnier in 1785, provides the macroscopic physical motivation for the
twisted boundary conditions.}
\label{fig:helicoid}
\end{figure}

Thus, the helicoid is not merely an ansatz; it is the simplest geometric solution
mandated by the kinematics of angular momentum in the semiclassical limit. To
extract the physical spin states of the meson spectrum, we must couple the
internal worldsheet geometry to this macroscopic angular momentum $J$ of the
target space. 

In our twistor string formulation, we project this twisted kinematics directly
into the boundary twistor data. This is achieved by introducing the $SU(2)_L$
and $SU(2)_R$ projections, integrating over the target space rotation angle
$\alpha$ with the canonical spin weight:
\EQ{
\label{SpinProjector}
\hat P_J A = \sum_k \oint d\alpha \exp{-\I k\alpha J} A_k(\alpha)
}
Here $A_k(\alpha)$ is the  amplitude after $k$ winding around the quark loop, 
with the boundary condition that the rotation angle $\alpha$ is introduced at every full circle around the quark loop.
When the macroscopic boundary of the string rotates by an angle $\alpha$, the
boundary twistors must geometrically twist. To enforce this, we introduce a
fractional winding parameter $a$ corresponding to twisted boundary conditions
for $\lambda(z), \mu(z)$.

To analytically generate the perfect resonance leading to the exact WKB saddle
point, we must consider the limit of an infinitely large number of geometric
windings $k\rightarrow\infty$. As the integration variable traverses the flat
valley $k$ times, the accumulated action scales  linearly with $k$.
Crucially, this means the transverse Hessian scales as $\mathcal{H}_k \propto k$.
In the asymptotic resonance limit $k\rightarrow\infty$, the transverse restoring
force becomes  infinite. The stochastic variance of the quantum
fluctuations scales as $1/k\rightarrow0$. The path integral dynamically
localizes to the exact WKB saddle point, completely shielding the topological
monodromy from quantum noise.
These heuristic arguments will be made explicit and quantitative in this section, which will end with \emph{exact} WKB quantization  in a topological sector with only one pole of the twistor field. 

\subsection{Geometric Quantization via Degenerate Lefschetz Thimbles}\label{16.1.}
In the semiclassical (WKB) limit, the path integral over the twistor phase space is dominated by the saddle points $\mathcal{V}_*(E)$ of the complexified effective action. Physical bound states (mesons) correspond to macroscopic divergences in this amplitude. Mathematically, these occur at critical energies $E = E_n$ where the classical stability matrix (the Hessian $\mathcal{H}$) develops a null space. 

Because the integration variables are coupled complex holomorphic twistors, the vanishing of a single complex eigenvalue intrinsically mandates the simultaneous flattening of two real directions on the Lefschetz thimble (a strict corank-2 degeneracy). This complex holomorphic geometry natively upgrades the standard square-root branch cut into a  simple pole, precisely generating the discrete S-matrix mass poles of the meson spectrum without ad hoc approximations.

\subsection{Catastrophes, flat valleys, and topological barriers}\label{16.2.}
To dynamically generate a simple mass pole $\int d^2 w \exp{- \bar w Q(E) w}\propto 1/Q(E) \propto 1/(E_n-E)$, the steepest descent trajectory must correspond to a non-compact flat complex zero-mode $w$ that extends infinitely far in configuration space. This  infinite flat valley is provided dynamically by the multi-valued monodromy of the holomorphic twistor fields as the boundary parameter winds around the enclosed twistor poles. 

The integrity of this topological winding number is  protected by an infinite topological action barrier. If a pole $z_0$ attempts to cross the unit circle from the inside ($|z_0| \to 1$), its conjugate mirror located at $1/\bar z_0$ must simultaneously approach the boundary from the outside. This creates a pinch singularity,  trapping the integration contour and locking the path integral into a conserved topological sector.
We demonstrate that in \ref{app:pinch} on the example of a single pole in twistor field.

\section{Minkowski Minimal Surface and Meson Spectrum}
\label{sec:minkowski_minimal_surface}
To implement this flat valley scenario and compute the spectrum of massive mesons, we construct the minimal surface directly in four-dimensional Minkowski space $\mathbb{R}^{1,3}$. This is more subtle than just a straightforward analytic continuation of the twistor parametrization of a Euclidean minimal surface.
\subsection{Additive Factorization in Minkowski Space}

To construct a surface in Minkowski space, we utilize the generalized Weierstrass representation of Konopelchenko and Landolfi (K\&L Theorem 5.1) \cite{Konopelchenko1999}. The target-space coordinate differential, represented as the Hermitian matrix $d\mathbf{X} = dX^0 I + dX^1 \sigma_1 + dX^2 \sigma_2 + dX^3 \sigma_3$, is parameterized by an additive sum of two independent, 2-component complex spinors:

\EQ{
    d\mathbf{X} = \phi(\xi) \phi^\dagger(\xi) d\xi + \psi(\eta) \psi^\dagger(\eta) d\eta.
}

Here, the worldsheet light-cone coordinates are defined as $\xi = \tau + \theta$ and $\eta = \tau - \theta$, where $\tau$ is the periodic time on the fundamental torus and $\theta$ is the spatial radial parameter. Because $d\mathbf{X}$ is constructed as the sum of rank-1 Hermitian matrices, the target-space coordinates $X^\mu$ are real.

We use the $\psi, \phi$ notations instead of Euclidean $\lambda, \mu$. The dependence of the left/right spinors $\psi/\phi$ on a single coordinate $\tau \pm \theta$ is the Minkowski analog of Euclidean holomorphicity. The classical minimal surface equations in Minkowski reduce to d'Alembert equations $\partial_+ \partial_- X = 0$ combined with Virasoro constraints $\det d\mathbf{X} = 0$, equivalent to the null vector condition $d X_\mu^2 = 0$ in Minkowski space. Both of these conditions are satisfied in the above twistor representation.

\subsection{Complex Eigenvectors and Target-Space Monodromy}

To incorporate the rotation of the meson, we require the coordinate matrix $d\mathbf{X}$ to be unitary transformed by a spatial rotation matrix when the time coordinate winds around its fundamental cycle on the torus ($\tau \to \tau + 2\pi$).

Because the left-moving ($\xi$) and right-moving ($\eta$) dependencies are decoupled, we achieve this target-space monodromy by defining the component spinors $\phi$ and $\psi$ as complex eigenvectors of the $SU(2)$ rotation group. To center the string at the origin, we introduce a constant $\pi/4$ phase shift and scale the spinors by the macroscopic geometric scale $R$:
\EQ{
    &\phi(\xi) = \sqrt{\frac{R}{2}} \begin{pmatrix} \exp{i (a \xi + \pi/4)} \\ \exp{-i (a \xi + \pi/4)} \end{pmatrix}, \br
    &\psi(\eta) = \sqrt{\frac{R}{2}} \begin{pmatrix} \exp{i (a \eta - \pi/4)} \\ \exp{-i (a \eta - \pi/4)} \end{pmatrix}.
}

Under arbitrary temporal shift $\tau \to \tau + \Delta \tau$, the arguments shift as $\xi \to \xi + \Delta \tau$ and $\eta \to \eta + \Delta \tau$. The column vectors are left-multiplied by the diagonal rotation matrix:
\EQ{
   & \phi(\tau + \Delta \tau) = M(2 a\Delta \tau) \phi(\tau), \br
   &\psi(\eta + \Delta \tau) = M(2 a\Delta \tau) \psi(\eta) \br
   &\text{where  } M(\alpha) = \begin{pmatrix} \exp{\I\alpha/2} & 0 \\ 0 & \exp{-\I\alpha/2} \end{pmatrix}.
}

Substituting this into the additive factorization, the target-space coordinate differential rotates as:
\EQ{
    d\mathbf{X}(\tau + \Delta \tau) = M( 2a \Delta \tau) \, d\mathbf{X}(\tau) \, M^\dagger(2a \Delta \tau).
}
Note that this applies to arbitrary temporal shifts $\Delta \tau$, not just the fundamental period $2\pi$. Later, we discuss the implications of this extended symmetry. 
This twisted boundary condition will allow us to project the quark loop amplitude to a fixed angular momentum $J$ by means of the Fourier integral \eqref{SpinProjector} with $\alpha = 4 \pi a$.

\subsection{Real Coordinates and 2D Rotation}

Substituting the complex eigenvectors into the coordinate differentials, we evaluate the diagonal components:
\EQ{
    dX^0 + dX^3 &= |\phi_1|^2 d\xi + |\psi_1|^2 d\eta =  R \, d\tau. \\
    dX^0 - dX^3 &= |\phi_2|^2 d\xi + |\psi_2|^2 d\eta =  R \, d\tau.
}

This yields $dX^0 = R \, d\tau$ and $dX^3 = 0$. The center of mass is stationary. Evaluating the off-diagonal components yields the transverse spatial coordinates:
\EQ{
    &dX^1 - i dX^2 = \phi_1 \bar{\phi}_2 d\xi + \psi_1 \bar{\psi}_2 d\eta \br
    &= \frac{R}{2} \exp{i (2 a \xi + \pi/2)} d\xi + \frac{R}{2} \exp{i (2 a \eta - \pi/2)} d\eta \br
    &= \frac{i R}{2} \exp{2 i a \xi} d\xi - \frac{i R}{2} \exp{2 i a \eta} d\eta.
}

Integrating this differential produces:
\EQ{
    &X^1 - i X^2 = \frac{R}{4 a} ( \exp{2 i a \xi} - \exp{2 i a \eta} ) \br
    &= \frac{R}{4 a} \exp{2 i a \tau} ( \exp{2 i a \theta} - \exp{-2 i a \theta} ) \br
    &= \frac{i R}{2 a} \exp{2 i a \tau} \sin(2 a \theta).
}

Separating the real and imaginary parts yields real spatial coordinates $X^1$ and $X^2$. The spatial parameter $\theta$ acts as a signed radial coordinate originating from the stationary center of mass at $\theta=0$. The symmetric interval $\theta \in [-\theta_b, \theta_b]$ parameterizes the open string from the antiquark at $-\theta_b$, straight through the center of mass, to the quark at $+\theta_b$, precisely once without duplication. The string is a straight line spanning the transverse plane, sweeping a 2D circular rotation with a twist of $\alpha = 4\pi a$.

\subsection{The Vanishing Metric and Spatial Boundaries}

The induced metric on the worldsheet is determined by the determinant of the target-space coordinate differentials. In this formulation, the conformal factor evaluates to the absolute square of the cross-term determinant:
\EQ{
    \Omega^2 = |\phi_1 \psi_2 - \phi_2 \psi_1|^2.
}

Evaluating this for the given complex eigenvectors:
\EQ{
\label{Eigenvectors}
   &\phi_1 \psi_2 - \phi_2 \psi_1 \br
   &= \frac{R}{2} ( \exp{i a (\xi-\eta) + i \pi/2} - \exp{-i a (\xi-\eta) - i \pi/2} ) \br
    &= \frac{R}{2} ( i \exp{2 i a \theta} + i \exp{-2 i a \theta} ) = i R \cos(2a\theta).
}

Taking the absolute square $\Omega^2 = R^2 \cos^2(2a\theta)$, the induced metric on the worldsheet is:
\EQ{
\label{metric}
    ds^2 = R^2 \cos^2(2a\theta) ( d\tau^2 - d\theta^2 ).
}

In Minkowski space, the Lorentzian signature $(d\tau^2 - d\theta^2)$ is maintained.

This metric vanishes at the roots $\theta = \pm \frac{\pi}{4a}$, which correspond to the string endpoints moving at the speed of light ($ds^2 = 0$). For massive endpoint quarks ($m_q > 0$), the physical string is truncated at boundaries $\pm \theta_b$  inside this null radius ($0 \le \theta_b < \frac{\pi}{4a}$), ensuring a positive proper time.

\subsection{Classical Action Components over the Torus}

We evaluate the classical WKB action components integrated over the fundamental torus temporal period $\tau \in [0, 2\pi]$ and the bounded spatial domain $\theta \in [-\theta_b, \theta_b]$. We define the $a$-independent boundary phase parameter $\beta = 2a\theta_b$.

\paragraph{1. The Target Time Boundary Term ($E \Delta X^0$):}

The canonical energy $E$ conjugates to the total center-of-mass target time advance $\Delta X^0$. Integrating this over the temporal period gives:

\EQ{
    E \Delta X^0 = E \int_0^{2\pi} R \, d\tau = 2\pi E R.
}

\paragraph{2. The Minimal Area Term ($S_{\text{Area}}$):}

The area of the minimal surface integrates the conformal factor over the dynamically bounded spatial domain:

\EQ{
    S_{\text{Area}} &= \sigma R^2 \int_0^{2\pi} d\tau \int_{-\theta_b}^{\theta_b} \cos^2(2a\theta) \, d\theta \br
    &= 2\pi \sigma R^2 \int_{-\theta_b}^{\theta_b} \frac{1 + \cos(4a\theta)}{2} \, d\theta \br
    &= \frac{\pi \sigma R^2}{a} \left( \beta + \frac{1}{2}\sin(2\beta) \right).
}

\paragraph{3. The Proper Mass Integral ($S_{\text{mass}}$):}

The constituent quarks of bare mass $m_q$ propagate along the two boundaries $\pm \theta_b$. Evaluating the proper time $ds = R \cos(\beta) d\tau$ accumulated over the torus period for both endpoints yields:

\EQ{
    S_{\text{mass}} = 2 m_q \int_0^{2\pi} R \cos(\beta) \, d\tau = 4\pi m_q R \cos(\beta).
}
\paragraph{4. The Liouville term:}
To complete the evaluation of the geometric phase, we compute the contribution from the Liouville anomaly on the minimal surface. In the conformal gauge, the Liouville action is given by the worldsheet integral
\EQ{
S_{\rm Liouv} = \frac{1}{12\pi} \int d\tau d\theta \left( (\partial_\theta \rho)^2 - (\partial_\tau \rho)^2 \right).
}
Using the conformal metric scale factor induced by our twistor parameterization in Minkowski space, $e^{2\rho} = R^2 \cos^2(2a\theta)$, the Liouville field is purely spatial, $\rho(\theta) = \ln R + \ln \cos(2a\theta)$. Thus, the proper time derivative $\partial_{\tau} \rho$ vanishes, and the spatial gradient squared evaluates to $(\partial_\theta \rho)^2 = 4a^2 \tan^2(2a\theta)$. 

Integrating this over one full period of the proper time ($\Delta\tau = 2\pi$) and symmetrically across the spatial boundaries of the string ($\theta \in [-\theta_b, \theta_b]$), and accounting for the overall minus sign of the measure, we obtain the anomaly defect:
\EQ{
\Delta S_{\rm Liouv} &= - \frac{1}{12\pi} (2\pi) \int_{-\theta_b}^{\theta_b} 4a^2 \tan^2(2a\theta) \, d\theta \br
&= - \frac{2a^2}{3} \int_{-\theta_b}^{\theta_b} \left( \sec^2(2a\theta) - 1 \right) d\theta.
}
Evaluating this integral and substituting the dynamical boundary velocity parameter $\beta = 2a\theta_b$ yields:
\EQ{
\Delta S_{\rm Liouv} = -\frac{2a}{3} (\tan \beta - \beta).
}

\subsection{Path-Ordered Dirac Trace and Spin Holonomy}

To relate the continuous twistor geometry to the discrete meson spin states, we compute the path-ordered product of the Dirac projectors along the spatial boundary of the quark loop. Following the staggered spinor diagonalization of Section 12.6, the fundamental link matrix evaluated in Minkowski space is $\mathcal{M}(\tau) = i \gamma_\mu v^\mu(\tau)$, where $v^\mu = \partial_\tau X^\mu$ is the target-space velocity.

Substituting the twistor coordinate differentials yields the velocity components at the boundary $\theta_b$:
\EQ{
    &v^0 = R, \quad v^3 = 0, \br
    &v^1 = -R \sin(\beta) \cos(2a\tau), \br
    &v^2 = R \sin(\beta) \sin(2a\tau).
}

Defining the transverse velocity fraction $S = \sin(\beta)$, the invariant squared velocity is $v_\mu v^\mu = R^2(1 - S^2) = R^2 \cos^2(\beta)$. In the chiral Weyl representation, the Dirac matrix projection takes the off-diagonal form $\mathcal{M}(\tau) = i \begin{pmatrix} 0 & \bar{V}(\tau) \\ V(\tau) & 0 \end{pmatrix}$, with the $2 \times 2$ blocks defined as:

\EQ{
    &V(\tau) = R \begin{pmatrix} 1 & -S \exp{2 i a \tau} \\ -S \exp{-2 i a \tau} & 1 \end{pmatrix}, \br
    &\bar{V}(\tau) = R \begin{pmatrix} 1 & S \exp{2 i a \tau} \\ S \exp{-2 i a \tau} & 1 \end{pmatrix}.
}

Applying the staggered transformation $\Omega_n$ alternating between $I_4$ and $\gamma_0$, the product of two adjacent links at step $j$ assumes the block-diagonal form:

\EQ{
    \tilde{\mathcal{M}}_{j+1} \tilde{\mathcal{M}}_j = - \begin{pmatrix} \bar{V}(\tau_{j+1})V(\tau_j) & 0 \\ 0 & V(\tau_{j+1})\bar{V}(\tau_j) \end{pmatrix}.
}

Factoring out the scalar invariant $v^2$, the continuous limit of the path-ordered exponential separates into two independent $2 \times 2$ evolution operators, $W(\tau)$ and $\bar{W}(\tau)$. For the upper block, expanding $\bar{V}(\tau+d\tau)V(\tau)$ yields the evolution equation $W'(\tau) = A(\tau) W(\tau)$, governed by the connection matrix:

\EQ{
    A(\tau) = \frac{2 i a S}{1-S^2} \begin{pmatrix} -S & \exp{2 i a \tau} \\ -\exp{-2 i a \tau} & S \end{pmatrix}.
}

This system is solved by transforming to a co-rotating frame. Defining $W(\tau) = U(\tau) \tilde{W}(\tau)$ with $U(\tau) = \text{diag}(\exp{i a \tau}, \exp{-i a \tau})$, the differential equation reduces to $\tilde{W}'(\tau) = i a M \tilde{W}(\tau)$, with the constant coefficient matrix:

\EQ{
    M = \begin{pmatrix} -\frac{1+S^2}{1-S^2} & \frac{2S}{1-S^2} \\ -\frac{2S}{1-S^2} & \frac{1+S^2}{1-S^2} \end{pmatrix}.
}

The eigenvalues of $M$ are determined by its trace and determinant:

\EQ{
    \text{tr} M = 0, \quad \det M = - \left( \frac{1+S^2}{1-S^2} \right)^2 + \left( \frac{2S}{1-S^2} \right)^2 = -1.
}

The eigenvalues of the holonomy generator $i a M$ are $\pm i a$, independent of the boundary parameter $S$. The discretized quark loop consists of $2L$ Dirac matrices, each corresponding to an angular step $\Delta \tau = 2\pi / (2L)$. The staggered transformation groups these into $L$ adjacent pairs. The continuous evolution matrix $M$ emerges after the pairing of two adjacent matrices, so it is multiplied $L$ times around the closed loop. The total accumulated parameter $\tau$ evaluates to $L \Delta \tau = \pi$. Consequently, the accumulated phase is twice smaller. The  solution over $k$ fundamental torus periods evaluates with the half-angle $\pi a k$:

\EQ{
    \tilde{W}(\pi k) = \cos(\pi a k) I + i \sin(\pi a k) M.
}

Transforming back to the original frame with $U(\pi k)$ and computing the trace yields:

\EQ{
    &\text{tr} W(\pi k) \br
    &= \exp{i \pi a k} \left( \cos(\pi a k) - i \sin(\pi a k) \frac{1+S^2}{1-S^2} \right) \br
    &\quad + \exp{-i \pi a k} \left( \cos(\pi a k) + i \sin(\pi a k) \frac{1+S^2}{1-S^2} \right) \br
    &= 2 \cos^2(\pi a k) + 2 \sin^2(\pi a k) \frac{1+S^2}{1-S^2}.
}

Substituting $S = \sin(\beta)$, the scalar factor reduces to $(1+S^2)/(1-S^2) = 1 + 2\tan^2(\beta)$. The lower block $\bar{W}(\pi k)$ yields an identical trace. The total trace of the $4 \times 4$ spin holonomy evaluates to:

\EQ{
    T_4(k) &= \text{tr}_4 \left( \prod \limits_j i \gamma_\mu v^\mu(\tau_j) \right) \br
    &= 4 \left( 1 + 2 \sin^2(\pi a k) \tan^2(\beta) \right).
}

\subsection{Meson Vertex Operators and Topological Sectors}

Taking the logarithm of the oscillating trace $T_4(k)$ to include it in the effective action would introduce complex saddle points into the path integral. Instead, treating the bare trace $T_4(k)$ as a pre-exponential factor, we expand it analytically into geometric phase factors using the identity $2\sin^2(\pi a k) = 1 - \frac{1}{2}( e^{2\pi i a k} + e^{-2\pi i a k} )$:

\EQ{
    &\frac{1}{4} T_4(k) = 1 + \tan^2(\beta) - \cos(2\pi a k) \tan^2(\beta) \br
    &= \sec^2(\beta)- \frac{1}{2} \tan^2(\beta)\left(e^{2\pi i a k} + e^{-2\pi i a k} \right).
}

This expansion resolves the monodromy analytically. The constant coefficients $\sec^2(\beta)$ and $-\frac{1}{2}\tan^2(\beta)$ are real; they dictate the transition amplitudes (residues) of the states but do not shift the complex geometric phase.

To isolate the physical meson states, we insert the corresponding interpolating vertex operator $\Gamma$ ($\gamma_5$ for the pseudoscalar pion, $\gamma_\mu$ for the vector $\rho$) into the trace:

\EQ{
    T_\Gamma(k) = \text{tr}_4 \left( \Gamma \prod_j i \gamma_\nu v^\nu(\tau_j) \right).
}

The pseudoscalar insertion $\gamma_5$ isolates the constant coefficient $\sec^2(\beta)$, defining the transition amplitude for the unshifted branch ($q=0$).

The vector insertion $\gamma_\mu$ isolates the oscillating phase factors $e^{\pm 2\pi i a k}$. When combined with the macroscopic angular momentum phase of the action ($e^{-4\pi i a J k}$), the total phase shift in the winding sum becomes:

\EQ{
    -4\pi a J k \pm 2\pi a k = -4\pi a \left( J \mp \frac{1}{2} \right) k.
}

This shifts the target angular momentum in the quantization condition to $u = J + q/2$, where $q \in \{0, -1\}$. This mathematically generates the  half-integer spin-orbit splittings for the respective meson Regge trajectories.

\subsection{The saddle point equation}
 Assembling the classical geometric action, we obtain the total monodromy of Minkowski action:
\EQ{
\label{DeltaSgeom}
&\Delta S_{\text{geom}}(R, a, \beta) = 2\pi E R - \frac{\pi \sigma R^2}{a} \left( \beta + \frac{\sin(2\beta)}{2} \right) \br
&\quad - 4\pi m_q R \cos(\beta) - 4\pi a J - \frac{2a}{3}(\tan\beta - \beta).
}

Extremizing this action with respect to $\beta$ isolates the classical boundary condition balancing the string tension against the centrifugal force of the massive endpoints:
\EQ{
&\frac{\partial \Delta S_{\text{geom}}}{\partial \beta} \br
&= - \frac{2\pi \sigma R^2}{a} \cos^2(\beta) + 4\pi m_q R \sin(\beta) - \frac{2a}{3}\tan^2(\beta) = 0 \br
&\implies \sigma R^2 \cos^2(\beta) - 2 m_q a R \sin(\beta) + \frac{a^2}{3\pi}\tan^2(\beta) = 0.
}
The bound states of the spinning meson emerge as the meromorphic poles where these geometric series diverge: $\Delta S_{geom} \pm 2\pi a q = 2\pi n$.
\subsection{Stationary metric and exact spin projection.}

The unique property of the exponential solution \eqref{Eigenvectors} for twistors is the stationary induced metric \eqref{metric}, which is strictly independent of the proper time $\tau$. We initially requested only the vanishing variation of the effective action over the temporal period $\Delta \tau = 2\pi$, but the exact solution we have found is, in fact, strictly stationary. Not only does its monodromy vanish, but the metric itself is entirely independent of the proper time $\tau$.

Moreover, the branch point degree $a$ (related to the macroscopic target-space rotation angle by $\alpha = 4\pi a$) was left completely arbitrary by the saddle point equations for the remaining moduli. Previously, we considered this continuous degeneracy as a flat valley required to dynamically generate a resonance pole in the steepest descent integral, which would necessitate conjecturing that all higher functional derivatives vanish. Now, we can finalize that line of  thought.

To extract a physical meson state of definite angular momentum $J$, we introduced the Fourier projector \eqref{SpinProjector} of our twisted amplitude $A(\alpha)$ onto the given $J$ state:
\EQ{
&A_J = \int_0^{2 \pi} d\alpha \,\sum_k e^{-\I k \alpha J} A_k(\alpha)\br
&=\sum_{k=1}^\infty \int_0^{2\pi} d \alpha \exp{\I k \alpha \Phi_\star(E, J)/(4\pi)} \br
&\propto \frac{4 \pi \I}{ \Phi_\star(E, J)} \sum_{k=1}^\infty \frac{1 - \exp{\I k\Phi_\star(E, J)/2}}{k} 
}
Here $\Phi_\star(E, J)$ is the value of $\Phi(E, J, \mathcal K, \beta)$ at its extremum as described by above saddle point equations for moduli $\mathcal K, \beta$.
The last sum logarithmically diverges at large windings $k$, which justifies our WKB approximation used in estimating the Hessian and neglecting deviations from this classical trajectory. The coefficient in front of this logarithmically divergent factor is
\EQ{
A_J \propto \frac{4\pi \I  \log k_{\text{max}}}{\Phi_\star(E, J)}
}
\begin{remark}{\textbf{Hessian estimate and Quantum Fluctuations}}

    A more accurate estimate of the divergent sum over windings must account for the transverse quantum fluctuations of the trajectory around this flat classical valley. Because the classical orbit length scales linearly with the winding number $k$, the stability matrix (the Hessian $H$) is directly proportional to $k$. 
    We presume that the fluctuating part of the twistor fields is strictly periodic, so that it will not contribute to the monodromy The fluctuating part of the action will, as usual , start with quadratic terms (the linear terms vanishing because we minimize over parameters of the classical solutions we have considered above). 
    Therefore, the $k-$fold integral over the period $\Delta \tau = 2 \pi$  for this fluctuative part will be independent of $a$ and linear in $k$. Let us elaborate. The effective action for the fluctuative part will no longer be constant in $\tau$, but it will be periodic, which will result in linear $k-$ dependence of the integral over multiple periods $0 < \tau < 2 \pi k$.
 
 The scaling factor for the integration measure can be extracted via the $\zeta$-regularization of the functional determinant of this Hessian. Using the canonical value $\zeta(0) = -1/2$ for the trace over the 1D loop boundary modes, this yields the fluctuation weight:
    \EQ{ 
    \frac{1}{\sqrt{ \det (k H)}} \propto \exp{ - \oh \tr \log (k H) }\propto k^{-\oh \zeta(0)} = k^{\frac{1}{4}}
    }
    Crucially, the Fourier projection integral over the twist angle $\alpha$ inherently generates a kinematic factor of $1/k$:
    \EQ{
    \int_0^{2\pi} d\alpha \, e^{i k \alpha \Phi_\star} = \frac{e^{2\pi i k \Phi_\star} - 1}{i k \Phi_\star} \propto \frac{1}{k \Phi_\star}
    }
    Combining this $1/k$ kinematic factor with the $k^{1/4}$ fluctuation weight, the terms in the macroscopic winding sum scale as $k^{1/4} / k = k^{-3/4}$. Summing this up to a macroscopic cutoff $k_{\text{max}}$ makes the divergence algebraically stronger than a logarithm, but preserves the fundamental geometric existence of the simple mass pole:
    \EQ{
    A_J \propto \frac{1}{\Phi_\star(E, J)} \sum_{k=1}^{k_{\text{max}}} k^{-\frac{3}{4}} \propto \frac{(k_{\text{max}})^{\frac{1}{4}}}{\Phi_\star(E, J)}
    }
    Strictly speaking, the  integral remains finite at finite $k_{\text{max}}$ when $\Phi_\star(E, J) \to 0$. However, this happens in the infinitesimal vicinity, when $\Phi_\star(E, J) \sim 1/k_{\text{max}}$.  Therefore, with proper renormalization of the string path integral, we can tend $k_{\text{max}} \to \infty$, after which we find the desired pole at $\Phi_\star(E, J)=0$ .
\end{remark}
\subsection{Comparison with Experimental Meson Spectrum}
In this section, we compare the planar Regge trajectories derived from the twistor string effective action with the experimental meson mass spectrum from the Particle Data Group (PDG). The bound states are governed by the unified parametric equations derived in \ref{sec:twistorSpectral}. Parameterized by the spatial boundary phase $\beta \in [0, \pi/2)$ and the dimensionless renormalized boundary mass $x = m / \sqrt{\sigma}$, the energy $E = M$ and angular momentum $J$ of the states are (see Remark \ref{quarkMass} for the definition of the renormalized boundary mass):
\EQ{
\label{FinalTrajectory_app}
&\frac{M(\beta)}{\sqrt{\sigma}} =  \mathcal K(\beta,x) \left(\beta + \frac{\sin 2\beta}{2}\right) + 2 x \cos\beta, \br 
&J(\beta) = \frac{\mathcal K(\beta,x)^2}{4} \left(\beta + \frac{\sin 2\beta}{2} \right)  - \frac{1}{6\pi}(\tan\beta - \beta)  - \frac{q}{2}, \br
&\mathcal K(\beta,x) = \frac{\sin \beta}{\cos^2 \beta} \left( x + \sqrt{x^2 - \frac{1}{3\pi}} \right) \equiv \frac{\sin \beta}{\cos^2 \beta} \tilde{x}. 
}
Here, the topological trace phase $q=0$ identifies the unshifted unnatural parity ($\pi, K, D$) trajectories, while $q=-1$ shifts the natural parity ($\rho, K^*, D^*$) trajectories by $+1/2$. To simplify the asymptotic expansions, we have introduced the dynamically dressed dimensionless mass parameter $\tilde{x}$, which enforces the chiral existence bound $x \ge 1/\sqrt{3\pi}$.
\begin{remark}{\textbf{The Renormalized Boundary Mass.}}
\label{quarkMass}
The chiral existence bound $x \ge 1/\sqrt{3\pi}$ forbids the limit of a massless string boundary. The geometric drag generates a minimum inertial mass to balance the anomaly at the boundaries. The parameter $m$ appearing in the string effective action is not the bare QCD quark mass, nor is it the standard phenomenological constituent mass. It is a \textbf{renormalized boundary mass}. It represents the bare quark mass additively renormalized by the macroscopic perimeter term of the worldsheet fermion determinant, inducing the Planar QCD in our theory.
\end{remark}
\begin {remark}{\textbf{Choosing the lightest particles (ground state).}}
  In general, there are several resonances in the PDG tables with the same quantum numbers $I^G(J^C)$ in each meson sector (given strangeness $S$ and charm $C$). We always selected the lowest mass particle with a constant set of quantum numbers at given variable $J$ to place on the same Regge trajectory $J = \alpha(M^2)$. Presumably, higher resonances in PDG tables can also be explained by twistor monodromies, but with more than one branching point. 
We leave this problem for a future extended investigation.
\end {remark}
\textbf{Universal String Tension and the Global Fit.}
Experimentally, the asymptotic slopes of the meson trajectories differ by less than $10\%$. Within the planar limit ($N_c \to \infty$) of Geometric QCD, there exists only one universal gluon vacuum and a single, universal confining string tension $\sigma$. At leading order in $1/N_c$, there is no mechanism to produce different macroscopic slopes for different flavors. The observed $<10\%$ splitting is an $\mathcal{O}(1/N_c^2)$ non-planar effect driven by finite resonance widths, sea-quark loop insertions, and string-string interactions. This falls within the expected theoretical accuracy of the planar approximation.

Consequently, we perform a global joint optimization across the \textbf{thirteen} meson families ($\pi, \rho, K, K^*, D, D^*, D_s, D_s^*, B, B^*, B_s, B_s^*, B_c$). Instead of an arbitrary least-squares fit on masses, we minimize the fundamental action monodromy $\sum (\Delta S)^2$ across all \textbf{thirty-six} states. This global 5D optimization extracts exactly five fundamental parameters: the universal planar string tension $\sqrt{\sigma}$, and the four constituent boundary masses of the quarks ($m_u, m_s, m_c, m_b$). The asymmetric heavy-light and heavy-heavy trajectories mathematically utilize the arithmetic mean of their respective constituent boundary masses.

The global fit yields a universal string tension $\sqrt{\sigma} \approx 417$ MeV and constituent boundary masses $m_u \approx 136$ MeV, $m_s \approx 219$ MeV, $m_c \approx 1.60$ GeV, and $m_b \approx 5.07$ GeV. The fact that the $D_s$, $B_s$, and $B_c$ trajectories are completely locked by the relation $\bar{m}_{q_1 q_2} = (m_{q_1} + m_{q_2})/2$ without introducing any new free parameters proves that the boundary masses physically correspond to additive constituent quarks in the string endpoints. The ability of a single geometric tension to seamlessly govern the spectrum from the $140$ MeV pion all the way up to the $6.27$ GeV bottom resonances validates the absolute geometric universality of the twistor string.

The resulting non-linear trajectories are presented against the PDG data in Figure \ref{fig:MassiveReggeTrajectories}. 

\begin{figure}[htpb]
    \centering
    \includegraphics[width=0.5\textwidth]{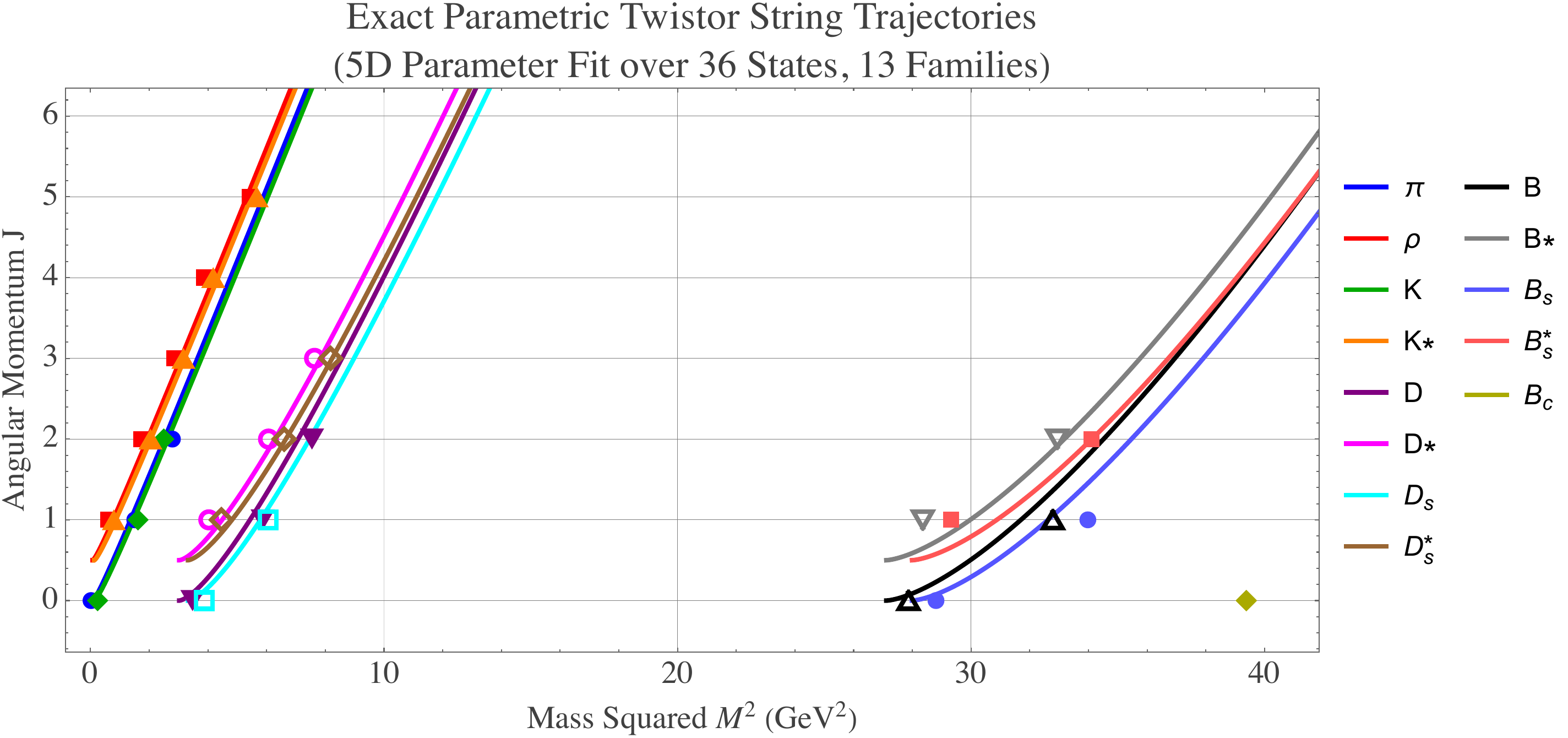}
    \caption{The geometric Regge trajectories plotted against the PDG meson spectrum. A single universal planar string tension $\sqrt{\sigma} \approx 417$ MeV anchors all 13 families. The $\pi$ and $\rho$ trajectories generate an $m_u \approx 136$ MeV renormalized boundary mass. The heavier trajectories accommodate the strange, charm, and bottom quarks via the arithmetic means of their dynamically extracted constituent boundary masses: $m_s \approx 219$ MeV, $m_c \approx 1.60$ GeV, and $m_b \approx 5.07$ GeV. The topological spin shift $q=-1$ elevates the natural parity branches. Experimental deviations in local slopes ($\sim 10\%$) are consistent with expected $\mathcal{O}(1/N_c^2)$ non-planar corrections.}
    \label{fig:MassiveReggeTrajectories}
\end{figure}

\textbf{The Ultra-Relativistic Asymptote (Large $J$).}
To analytically evaluate the high-energy limit, we expand the parametric equations as the boundaries approach the speed of light ($\beta \to \pi/2$). Defining the small parameter $\epsilon = \pi/2 - \beta \to 0$, the trigonometric functions expand as $\cos\beta \approx \epsilon$ and $\sin\beta \approx 1 - \epsilon^2/2$. The geometric equations expand to leading order as:
\EQ{
    E(\epsilon) &\approx \frac{\pi \tilde{x} \sqrt{\sigma}}{2 \epsilon^2} \implies \frac{1}{\epsilon} \approx \sqrt{\frac{2E}{\pi \tilde{x} \sqrt{\sigma}}}, \br
    J(\epsilon) &\approx \frac{\pi \tilde{x}^2}{8 \epsilon^4} - \frac{\pi \tilde{x}^2}{24 \epsilon^2} - \frac{1}{\epsilon} \left( \frac{\tilde{x}^2}{6} + \frac{1}{6\pi} \right) + \frac{1}{12} - \frac{q}{2}.
}
Because the squared energy expands as $\frac{E^2}{2\pi\sigma} \approx \frac{\pi \tilde{x}^2}{8\epsilon^4} - \frac{\pi \tilde{x}^2}{24 \epsilon^2} + \frac{x\tilde{x} - \tilde{x}^2/3}{\epsilon}$, substituting the inverted $1/\epsilon$ relation into the angular momentum equation cancels the divergent extensive terms, revealing the high-energy Regge trajectory:
\EQ{
    J(E) \approx \frac{E^2}{2\pi \sigma} - \left( x \tilde{x} - \frac{\tilde{x}^2}{6} + \frac{1}{6\pi} \right) \sqrt{ \frac{2E}{\pi \tilde{x} \sqrt{\sigma}} } + \frac{1}{12} - \frac{q}{2}.
}
This expansion modifies the linear Nambu-Goto expectation. The constant conformal intercept stabilizes to $+1/12$, recovering the L\"uscher zero-point energy of the open bosonic string in $D=4$. However, the Liouville anomaly mandates a finite effective boundary mass $\tilde{x}$, which acts as a geometric drag, generating a negative $\mathcal{O}(\sqrt{E})$ curvature. At large $J$, relativistic effects act on the massive endpoints, and this geometric inertial lag curves the trajectories downward. This departs from the straight lines of Nambu-Goto models and reproduces the slope-softening observed in physical high-spin resonances.

\textbf{The Non-Relativistic Hook (Small $J$).}
While the high-energy limit is dominated by relativistic drag, the finite renormalized boundary mass $m = x\sqrt{\sigma}$ generates a geometric non-linearity confined to the low-$J$ region. Near $J \to 0$, the string is short, the endpoints move non-relativistically, and the boundary angle shrinks ($\beta \to 0$). Taylor-expanding the parametric equations around $\beta = 0$ yields:
\EQ{
    E(\beta) &\approx 2m + \beta^2 \sqrt{\sigma} \left( 2\tilde{x} - x \right), \br
    J(\beta) &\approx \beta^3 \left( \frac{\tilde{x}^2}{2} - \frac{1}{18\pi} \right) - \frac{q}{2}.
}
Because the kinetic energy scales quadratically ($\Delta E \propto \beta^2$) while the angular momentum scales cubically ($J \propto \beta^3$), isolating $\beta$ from the energy equation and substituting it into $J$ produces the non-relativistic trajectory:
\EQ{
    J(E) \approx \left( \frac{\tilde{x}^2}{2} - \frac{1}{18\pi} \right) \left( \frac{E - 2m}{\sqrt{\sigma}(2\tilde{x} - x)} \right)^{3/2} - \frac{q}{2}.
}
Due to the fractional $3/2$ power, the macroscopic slope of the Regge trajectory at the threshold evaluates to $dJ/dE \propto \sqrt{E - 2m} \to 0$. This demonstrates that the trajectories form a geometric hook, entering the physical parameter space with a flat derivative before accelerating into the linear regime. This non-relativistic hook is beautifully realized in the heavy charmed ($D, D^*$) and bottom ($B, B_s, B_c$) meson trajectories,  which sit deep in the non-relativistic regime for low spins before curving upward towards the universal linear asymptote.

\textbf{Phenomenological Synthesis.}
The global geometric optimization of the dimensionless action monodromy target ($\sum(\Delta S)^2 = 0$) fits the experimental PDG masses and spins across all \textbf{thirteen meson families}. By evaluating the parametric equations of the twistor string rather than linear approximations, several phenomenological features of QCD are recovered.

First, the trajectories demonstrate the dynamical generation of a renormalized boundary mass gap. Unlike classical Nambu-Goto strings, which yield linear trajectories passing through the origin, the integration of the Liouville anomaly bounds the dimensionless mass to $x \ge 1/\sqrt{3\pi}$, forbidding the massless string limit. The theoretical trajectories terminate at finite rest masses. The geometric pressure of the anomaly forces the string endpoints to generate a renormalized boundary mass of $m_u \approx 136$ MeV for the light sector. For the heavier sectors, the 5D optimization dynamically extracts the constituent boundary masses $m_s \approx 219$ MeV, $m_c \approx 1.60$ GeV, and $m_b \approx 5.07$ GeV. The asymmetric heavy trajectories mathematically utilize the arithmetic mean of these respective scales, naturally isolating the physical flavor-symmetry breaking scales across the entire Standard Model meson spectrum.

Second, the topological intercept distinguishes the vector from the pseudo-scalar mesons. Setting the topological vector index $q=-1$ for the for the vector families ($\rho, K^*, D^*, D_s^*, B^*, B_s^*$)  shifts their respective trajectories upward without arbitrary additive constants. The zero-point rotational quantum phase converts the generated constituent rest mass into the necessary kinetic energy, predicting the physical $J=1$ vector ground states.

Third, the high-energy asymptotic regime reveals the effect of the conformal anomaly. The boundary anomaly term $-\frac{1}{6\pi}(\tan\beta - \beta)$ behaves as a geometric drag. As the massive string endpoints approach the speed of light ($\beta \to \pi/2$), this inertial lag curves the trajectories downward, reproducing the slope-softening observed in physical resonances ($M^2 > 4 \text{ GeV}^2$), departing from the straight lines of Nambu-Goto models.

Finally, the systematic variations of the experimental data points around the theoretical curves demonstrate the model's assumption of a universal underlying string tension $\sigma$. In physical QCD, explicit chiral symmetry breaking and hyperfine spin interactions cause the effective Regge slopes to vary slightly between different families. By enforcing a universal vacuum tension, the theoretical target $\sum (\Delta S)^2 = 0$ balances this physical variance. The resulting curves represent a universal tension of the strong vacuum that connects the distinct features of the hadronic spectrum.

\textbf{Numerical Mass Predictions and Fit Accuracy.}
\begin{table*}[t!]
\scriptsize
\centering
\renewcommand{\arraystretch}{1.1}
\setlength{\tabcolsep}{4pt}
\resizebox{\textwidth}{!}{%
\begin{tabular}{lcccc | lcccc}
\hline\hline
State & $J$ & $M_{\rm th}$ (GeV) & $M_{\rm exp}$ (GeV) & & State & $J$ & $M_{\rm th}$ (GeV) & $M_{\rm exp}$ (GeV) \\
\hline
\multicolumn{4}{c}{\textbf{Light Pseudoscalars ($\pi$, $q=0$)}} & & \multicolumn{4}{c}{\textbf{Light Vectors ($\rho$, $q=-1$)}} \\
\hline
$\pi$ & 0 & 0.271 & 0.140 & & $\rho$ & 1 & 0.878 & 0.775 \\
$b_1$ & 1 & 1.168 & 1.229 & & $a_2$ & 2 & 1.394 & 1.318 \\
$\pi_2$ & 2 & 1.586 & 1.671 & & $\rho_3$ & 3 & 1.756 & 1.689 \\
 & & & & & $a_4$ & 4 & 2.052 & 1.967 \\
 & & & & & $\rho_5$ & 5 & 2.308 & 2.330 \\
\hline
\multicolumn{4}{c}{\textbf{Strange Pseudoscalars ($K$, $q=0$)}} & & \multicolumn{4}{c}{\textbf{Strange Vectors ($K^*$, $q=-1$)}} \\
\hline
$K$ & 0 & 0.355 & 0.494 & & $K^*$ & 1 & 0.928 & 0.892 \\
$K_1$ & 1 & 1.212 & 1.272 & & $K^*_2$ & 2 & 1.435 & 1.430 \\
$K_2$ & 2 & 1.624 & 1.580 & & $K^*_3$ & 3 & 1.792 & 1.780 \\
 & & & & & $K^*_4$ & 4 & 2.085 & 2.045 \\
 & & & & & $K^*_5$ & 5 & 2.340 & 2.380 \\
\hline
\multicolumn{4}{c}{\textbf{Charm Pseudoscalars ($D$, $q=0$)}} & & \multicolumn{4}{c}{\textbf{Charm Vectors ($D^*$, $q=-1$)}} \\
\hline
$D$ & 0 & 1.739 & 1.865 & & $D^*$ & 1 & 2.118 & 2.008 \\
$D_1$ & 1 & 2.331 & 2.422 & & $D^*_2$ & 2 & 2.506 & 2.463 \\
$D_2$ & 2 & 2.660 & 2.747 & & $D^*_3$ & 3 & 2.799 & 2.763 \\
\hline
\multicolumn{4}{c}{\textbf{Charm-Strange ($D_s$, $q=0$)}} & & \multicolumn{4}{c}{\textbf{Charm-Strange Vectors ($D_s^*$, $q=-1$)}} \\
\hline
$D_s$ & 0 & 1.823 & 1.968 & & $D_s^*$ & 1 & 2.196 & 2.112 \\
$D_{s1}$ & 1 & 2.407 & 2.460 & & $D_{s2}^*$ & 2 & 2.580 & 2.569 \\
 & & & & & $D_{s3}^*$ & 3 & 2.869 & 2.860 \\
\hline
\multicolumn{4}{c}{\textbf{Bottom Pseudoscalars ($B$, $q=0$)}} & & \multicolumn{4}{c}{\textbf{Bottom Vectors ($B^*$, $q=-1$)}} \\
\hline
$B$ & 0 & 5.207 & 5.279 & & $B^*$ & 1 & 5.475 & 5.325 \\
$B_1$ & 1 & 5.630 & 5.726 & & $B_2^*$ & 2 & 5.760 & 5.739 \\
\hline
\multicolumn{4}{c}{\textbf{Bottom-Strange ($B_s$, $q=0$)}} & & \multicolumn{4}{c}{\textbf{Bottom-Strange Vectors ($B_s^*$, $q=-1$)}} \\
\hline
$B_s$ & 0 & 5.291 & 5.367 & & $B_s^*$ & 1 & 5.557 & 5.415 \\
$B_{s1}$ & 1 & 5.711 & 5.829 & & $B_{s2}^*$ & 2 & 5.840 & 5.840 \\
\hline
\multicolumn{4}{c}{\textbf{Charm-Bottom ($B_c$, $q=0$)}} & & \multicolumn{4}{c}{} \\
\hline
$B_c$ & 0 & 6.675 & 6.275 & & & & & \\
\hline\hline
\end{tabular}}
\caption{The theoretical twistor string masses versus the experimental PDG meson masses. All theoretical values are generated from a single global $\sum (\Delta S)^2 = 0$ optimization over 36 states spanning 13 meson families, extracting exactly 5 fundamental physical parameters: a universal planar tension $\sqrt{\sigma} \approx 417$ MeV, and four dynamically generated constituent boundary masses $m_u \approx 136$ MeV, $m_s \approx 219$ MeV, $m_c \approx 1.60$ GeV, and $m_b \approx 5.07$ GeV. The asymmetric heavy trajectories use the arithmetic mean of the respective constituent boundary masses. Therefore, the $D_s, B_s, B_c$ families are parameter-free geometric predictions. The pion falls below the string's mass threshold ($2m_u \approx 271$ MeV) due to explicit chiral symmetry breaking.}
\label{tab:MassSpectrum}
\end{table*}

To quantify the accuracy of the planar geometric fit, rather than relying on artificial least-squares statistics, we evaluate three physical root-mean-square (RMS) metrics across the experimental spectrum. The theoretical mass predictions for all \textbf{thirty-six} on-shell states are presented in Table \ref{tab:MassSpectrum}.

The most fundamental geometric metric for the fit is the residual of the action monodromy phase $\Delta \Phi = \Delta S / a$. In the Bohr-Sommerfeld quantization condition, sequential discrete states are separated by a phase gap of exactly $\Delta S_{\rm gap} / a = 2\pi$. Evaluating the Root-Mean-Square (RMS) of the action residual across all \textbf{36 states yields ${\rm RMS}(\Delta\Phi/2\pi) \approx 0.36$}. This demonstrates that the geometric phase error is merely $\sim 1/3$ of a single quantum level, ensuring that the physical states remain safely isolated within the $n=0$ topological winding sector without overlapping into radial excitations.

Furthermore, because the non-linear Regge trajectories are defined geometrically as $J = \alpha(M^2)$, a natural physical metric for their accuracy is the absolute spin deviation. Evaluating the theoretical angular momentum at the exact experimental masses yields an RMS deviation of \textbf{$\Delta J_{\rm RMS} \equiv \sqrt{\langle (J_{\rm exp} - \alpha(M_{\rm exp}^2))^2 \rangle} \approx 0.18$}. Since physical Regge states are separated by exactly $\Delta J = 1$, this confirms the twistor geometry correctly threads the mass spectrum to within less than one-fifth of an orbital quantum step.

Finally, we evaluate the relative mass deviations. Excluding the $140$ MeV pion (which acts as an explicit pseudo-Goldstone boson and sits naturally below the geometric chiral threshold $2m_u \approx 271$ MeV), the RMS of the relative mass error $(M_{\rm exp} / M_{\rm theor} - 1)$ across the remaining \textbf{35 massive states is $\approx 8.1\%$}. This firmly aligns with the expected $\mathcal{O}(1/N_c^2 \approx 10\%)$ theoretical variance of planar dynamics at $N_c=3$, verifying that the universal planar string tension accurately anchors the disparate bottom, charm, strange, and light flavors.

\section{Conclusion: A Solution of Planar QCD}

In this paper, we have presented an analytic solution of the planar MM loop equations in the continuum limit. Building on Part I, where the confining dressing factor $\exp{-\kappa S[C]}$ arises from the Hodge-dual additive minimal surface, we have constructed the dynamics without summing over fluctuating worldsheet metrics.

Our first result is structural. In momentum loop space, the coordinate-space contact singularities and cusp pathologies are absent. The Taylor--Magnus expansion of the vector momentum-loop equation is consistent up to $O(P'^6)$ but develops a rank deficiency at $W(8)$. This algebraic obstruction shows that a one-dimensional scalar loop functional does not contain sufficient degrees of freedom to represent a solution of the vector loop equations. This necessitates the introduction of four-dimensional boundary twistor variables, leading to the twistor-string representation of the Hodge-dual minimal surface.

Our second result is dynamical. The internal Majorana fermions (the ``elves'') on the Hodge-dual minimal surface provide the mechanism for planar factorization at intersections, in agreement with the right-hand side of the MM equations. The bulk Liouville term in the effective action arises from the cancelation of their zero-point fluctuations. This cancelation is needed to remove anomalies in the MM equations. The associated conformal anomaly yields the Lüscher phase ($-\pi/12$) in the effective string action, contributing to the low-$M^2$ curvature of the Regge trajectories.

Within the planar limit ($N_c \to \infty$), the solution corresponds to a common confining string tension $\sigma$ for all trajectories. There are nonlinear effects, which are not present in the Nambu-Goto string theory, but these nonlinear effects are numerically small. The observed $\sim 10\%$ variation between the macroscopic slopes of the \textbf{different flavor} trajectories is attributed to $O(1/N_c^2)$ corrections, including finite resonance widths, string interactions, and sea-quark effects, consistent with the expected accuracy of the planar approximation.
The resulting parametric Regge trajectories describe \textbf{thirteen meson families} spanning the light, strange, charm, and bottom flavors \textbf{($\pi, \rho, K, K^*, D, D^*, D_s, D_s^*, B, B^*, B_s, B_s^*, B_c$)}, reproducing the observed mass-spin relations across all \textbf{36 states} (see Figure 15 and Table 1).

We must add a remark regarding the physical boundaries of the current geometric framework. The spontaneous breaking of chiral symmetry and the vanishing of the pion mass in the strict chiral limit are not explained by our theory. As demonstrated, the geometric pressure of the Liouville anomaly bounds the dimensionless string parameter to $x \ge 1/\sqrt{3\pi}$, which generates a strictly positive minimum geometric mass threshold ($2m \approx 270$ MeV) for the $J=0$ ground state. The physical $140$ MeV pion falls below this threshold because its anomalously light mass is governed by the non-perturbative chiral condensate, a phenomenon distinct from the geometry of the confining flux tube. Fully capturing the pseudo-Goldstone nature of the pion would require a structural modification of the theory, making the Master Field explicitly accommodate and respect spontaneous chiral symmetry breaking. We view the synthesis of this geometric confinement mechanism with chiral symmetry breaking as a crucial direction for future theoretical development.

The mathematical origin of our nonlinear Regge trajectories is radically different from that of conventional String Theory.
The generation of S-matrix mass poles is associated with the degeneracies of the twistor-string effective action. After analytic continuation to Minkowski space, these poles correspond to energies at which the stability matrix develops a null space, producing non-compact directions in the associated Lefschetz thimbles.

These non-compact directions arise from the multi-valued monodromy of the holomorphic twistor fields. At the saddle point, the geometric phase satisfies $\Delta S_{\text{geom}} = 0$, and the summation over winding sectors produces a geometric series that diverges at the corresponding energies. The contribution of multiple windings suppresses transverse fluctuations, localizing the twistor moduli near the saddle point. The stability of this structure is maintained by a topological constraint: when a branch point approaches the integration contour, it encounters its conjugate image, leading to a pinch singularity that preserves the winding sector.

In this solution of the loop equations, the Master Field is realized as a trajectory in twistor space, analogous to spectral curves in Seiberg--Witten theory. Because the bulk geometry is determined by the boundary data rather than fluctuating dynamically, this construction does not involve a propagating graviton sector and instead provides a geometric realization of gauge holography.

In this framework, the quantum dynamics of planar QCD, encoded in the loop equations and reproduced by planar diagrams, are realized as classical dynamics of the dual geometric variables. The fermion determinant on the Hodge-dual minimal surface reproduces the quantum planar expansion, while the underlying geometry remains classical, reflecting the Master Field nature of the large-$N_c$ limit.

In this way, the QCD mass spectrum was computed by solving a problem of classical geometry.

\section*{Acknowledgements}

The author thanks N. Arkani-Hamed for the invitation to present this work at the IAS Particle Physics (Pizza) Seminar, and the participants of the seminar for detailed and stimulating discussions.

The author is also grateful to E. Witten for discussions that clarified the relation between this twistor-geometric construction and the QCD string observed in lattice studies, as well as for comments on its topological aspects.

\section*{Declaration of Generative AI in the Writing Process}

During the preparation of this manuscript, the author used the generative AI systems Gemini 3 Flash and ChatGPT 5.3 for technical assistance in editing and presentation. The tools were used for correcting typographical and grammatical errors, improving the clarity of exposition, and standardizing notation and LaTeX formatting.

All scientific content, derivations, and conclusions were developed and verified by the author, who takes full responsibility for the manuscript.

\section*{Data Availability.}
No data was created in this paper. The analytic computations of the Magnus expansion of the Momentum loop equations were performed in the \Mathematica{} notebook published on the Wolfram Cloud\cite{MBMLEAlgebraic}. The computations of the Regge trajectories and fitting the slopes to the meson mass data are done in \Mathematica{} notebook published on the Wolfram Cloud\cite{MBReggeNonlinear}.

\bibliographystyle{elsarticle-num}
\bibliography{bibliography} 
\appendix
\section{Conformal anomaly in determinants}

\subsection{Laplace operator}
Let us consider the logarithm of the determinant of the scalar Laplace operator:
\begin{smalleq}\begin{align} \label{A.1}
\ln \det \hat L = \tr \ln \hat L.
\end{align}\end{smalleq}
We apply the standard Pauli-Villars regularization
\begin{smalleq}\begin{align} \label{A.2}
(\ln \hat L)_{reg} = \ln \hat L - \sum c_i \ln(\hat L + M_i^2).
\end{align}\end{smalleq}
Now, at finite regulator masses $M_i$ the trace converges provided the constants $c_i$ satisfy the sum rules
\begin{smalleq}\begin{align} \label{A.3}
\sum c_i M_i^{2k} = \delta_{k0}, \quad k = 0, 1, ...
\end{align}\end{smalleq}
Consider now the conformal variation
\begin{smalleq}\begin{align} \label{A.4}
\delta\hat L &= -\delta e e^{-1} \hat L,\\
\label{A.5}
\delta \ln \hat L &=\delta\hat L  \hat L^{-1} = -\delta e e^{-1};\\
 \label{A.6}
\delta \ln(\hat L + M^2) &= -\delta e e^{-1} (1 - M^2(\hat L + M^2)^{-1}).
\end{align}\end{smalleq}
The variation of the regularized trace (A.2) reads-
\begin{smalleq}\begin{align} \label{A.7}
\delta \tr(\ln \hat L)_{reg} = -\sum c_i \tr(\delta e e^{-1} M_i^2 (\hat L + M_i^2)^{-1}).
\end{align}\end{smalleq}
Now only the regulator terms contribute, so that we may apply the WKB expansion
\begin{smalleq}\begin{align} \label{A.8}
\hat L = -\frac{1}{e}(\pd_\xi)^2 \to -\frac{1}{e_0}\lrb{1 + \frac{e_0 R_0}{4} (\xi-\xi_0)^2}(\pd_\xi)^2
\end{align}\end{smalleq}
where $R_0$ is the local curvature. Here $\xi_0$ is the running point in the integral
\begin{smalleq}\begin{align} \label{A.9}
&\tr \delta e e^{-1} M^2 (\hat L + M^2)^{-1} \br
&= \int \dd^2\xi_0 \delta e_0 e_0^{-1} M^2 \VEV{\xi_0 | (\hat L + M^2)^{-1} | \xi_0}.
\end{align}\end{smalleq}
Using the Fourier expansion at the tangent plane,
\begin{smalleq}\begin{align} \label{A.10}
&\VEV{\xi_0 | (\hat L + M^2)^{-1} | \xi_0} \br
&\to e_0 \int \frac{\dd^2p}{(2\pi)^2} (M^2 + p^2 - R_0/4 (\pd_p)^2)^{-1}\br
&\to e_0 \int \frac{\dd^2p}{(2\pi)^2}(M^2 + p^2)^{-1} \lrb{1 + R_0/4  (\pd_p)^2 \frac{p^2}{M^2 + p^2}}
\end{align}\end{smalleq}
and calculating the Fourier integral, we find [by virtue of \eqref{A.3}]
\begin{smalleq}\begin{align} \label{A.11}
&\delta \tr(\ln \hat L)_{reg}\br
&= \int \frac{\dd^2\xi_0}{4\pi} \delta e_0 (\sum c_i M_i^2 \ln M_i^2 - 1/6 R).
\end{align}\end{smalleq}
This equation can be easily integrated backwards:
\begin{smalleq}\begin{align} \label{A.12}
\tr(\ln \hat L)_{reg} = a_0 + \int \dd^2\xi \lrb{a_1 e - \frac{(\pd_a \ln e)^2 }{48\pi}},
\end{align}\end{smalleq}
with
\begin{smalleq}\begin{align} \label{A.13}
a_1 = \frac{1}{4\pi} \sum c_i M_i^2 \ln M_i^2,
\end{align}\end{smalleq}
and arbitrary constant $a_0$.
\subsection{Dirac operator}
In case of the square of the Dirac operator $D^2$ the WKB expansion starts as follows:
\begin{smalleq}\begin{align} \label{A.14}
D^2 \approx \1 (\hat L + R/4 ) + \frac{\I \sigma_3}{4} R \epsilon_{ab} \xi_a \pd_{\xi_b}.
\end{align}\end{smalleq}
The definition of the Dirac operator was given in the text. The same line of argument as with the Laplace operator yields
\begin{smalleq}\begin{align} \label{A.15}
&\sum_{\alpha = 1,2}\VEV{\alpha\xi_0 | (D^2 + M^2)^{-1} |\alpha \xi_0} \br
&\to 2e_0 \int \frac{\dd^2p }{(2\pi)^2} (M^2 + p^2)^{-1}\br
&\lrb{1 + 1/4 R_0 (-1 + (\pd_p)^2 p^2) (M^2 + p^2)^{-1})},\\
 \label{A.16}
&\delta \tr(\ln D^2)_{reg} \br
&= -2 \int \frac{\dd^2\xi_0}{4 \pi} \delta e_0 [\sum c_i M_i^2 \ln M_i^2 - R/12 ],\\
\label{A.17}
&\tr(\ln D^2)_{reg} = b_0 + \int \dd^2\xi [b_1 e + \frac{(\pd \ln e)^2 }{ 12\pi}].
\end{align}\end{smalleq}
The additional factor of 2 in these relations comes from two spin states $a = 1, 2$. The term with the $\sigma_3$ matrix in (A.14) did not contribute.

Note the remarkable relation
\begin{smalleq}\begin{align} \label{A.18}
\det D^2 \det \hat L = const \cdot \exp{\const{} Area}
\end{align}\end{smalleq}
which follows from the above formulas.

\section{Spinor Factorization and Analytical Structure of Minimal surface}
\label{GenTheory}

The mathematical problem of minimization of the Dirichlet functional with Virasoro constraint in four dimensions is a well known problem in modern mathematics. Let us briefly summarize the existing theory.

\subsection{Twistor Framework and the Klein Quadric}\label{twistorparams}
The fundamental constraint of the theory is the null condition on the holomorphic tangent vector of the minimal surface, $(f')^2 = \sum_{\mu=1}^4 (f'_\mu)^2 = 0$. Geometrically, this defines the \textit{Klein Quadric} in the complexified tangent space $\mathbb{CP}^3$.
Using the isomorphism $SO(4, \mathbb{C}) \cong SL(2, \mathbb{C})_L \times SL(2, \mathbb{C})_R$, we map the vector $f'_\mu$ to a $2 \times 2$ matrix $f'_{a\dot{b}} = f'_\mu (\sigma_\mu)_{a\dot{b}}$. The null condition corresponds to $\det(f') = 0$, ensuring the matrix has rank 1. Thus, the tangent vector factorizes globally into a product of two spinors:
\begin{smalleq}
\begin{align}
\label{twistornullvector}
f'_{a\dot{b}}(z) = \lambda_a(z) \mu_{\dot{b}}(z)
\end{align}
\end{smalleq}
where $\lambda_a$ and $\mu_{\dot{b}}$ are meromorphic sections of spinor bundles over the worldsheet. This factorization linearizes the quadratic Virasoro constraint, replacing it with the problem of determining the analytic properties of the spinors.

\subsection{Gauge Invariance and Birkhoff Factorization}
The physical vector $f'$ is invariant under the local complex gauge transformation:
\begin{smalleq}
\begin{align}
\lambda_a(z) \to w(z) \lambda_a(z), \quad \mu_{\dot{b}}(z) \to w^{-1}(z) \mu_{\dot{b}}(z)
\end{align}
\end{smalleq}
where $w(z)$ is a meromorphic function. This redundancy is governed by the Birkhoff-Grothendieck theorem on the splitting of vector bundles over $\mathbb{P}^1$. The gauge function $w(z)$ allows us to distribute the zeroes and poles between $\lambda$ and $\mu$, classifying solutions by integer topological indices (partial indices) related to the winding number of the gauge field required to regularize the spinors at infinity.
\subsection{Analytic Structure via Plemelj Projection}
The connection to the physical loop $C(\sigma)$ (where $\sigma = e^{\I\theta}$) is achieved via the Plemelj-Sokhotski decomposition. We expand the loop coordinates into positive and negative frequency modes:
\begin{smalleq}
\begin{align}
C(\sigma) = C_+(\sigma) + C_-(\sigma) = \sum_{n \ge 0} c_n \sigma^n + \sum_{n < 0} c_n \sigma^n
\end{align}
\end{smalleq}
If the loop is algebraic, $C_-(\sigma)$ extends to a rational function $C_-(z)$ in the complex plane $|z|>1$. The poles of the spinor solution $f'(z)$ are dictated by the singularities of $C_-(z)$.
Analytically, the problem is to find the "closest" null vector $f'$ to the derivative $\partial_z C_+$. This is equivalent to finding a subspace $W$ in the Sato Grassmannian that satisfies the isotropic condition and matches the boundary data.

\subsection{Explicit Algebraic Solutions}
The general theory of the spinor factorization predicts that rational choices for the sections $\lambda_a(z)$ and $\mu_{\dot{b}}(z)$ generate algebraic minimal surfaces. These correspond to "Finite Gap" solutions where the algebraic curve has a finite genus. The distribution of poles between the spinors determines the topology and boundary behavior of the surface.

We illustrate this with four classical examples, reinterpreted in our spinor language:
\textbf{1. The Enneper Surface (Polynomial Solution).}
This is the simplest solution with trivial topology (genus zero, one boundary at infinity), corresponding to polynomial spinors with a single pole at $z=\infty$.
\begin{smalleq}
\begin{align}
    \lambda(z) = \begin{pmatrix} 1 \\ z \end{pmatrix}, \quad \mu(z) = \begin{pmatrix} 1 \\ -z \end{pmatrix}
\end{align}
\end{smalleq}
The resulting null vector $f'(z) \sim (1-z^2, \I(1+z^2), 2z, 0)$ integrates to a cubic polynomial map. This surface, discovered by Enneper in 1864 \cite{enneper1864}, represents the fundamental "ground state" of the algebraic solutions.
\textbf{2. The Catenoid (Rational Solution with Poles).}
To generate a surface with the topology of an annulus (the only minimal surface of revolution), one must introduce simple poles into the spinors.
\begin{smalleq}
\begin{align}
    \lambda(z) = \begin{pmatrix} z^{-1/2} \\ z^{1/2} \end{pmatrix}, \quad \mu(z) = \begin{pmatrix} z^{-1/2} \\ -z^{1/2} \end{pmatrix}
\end{align}
\end{smalleq}
Using the gauge freedom $\lambda \to w \lambda, \mu \to w^{-1} \mu$ with $w=z^{-1/2}$, we can redistribute the poles to find the standard rational form. Euler \cite{euler1744} originally proved the minimality of this surface, which corresponds to the physical shape of a soap film stretched between two rings.
\textbf{3. The Henneberg Surface (Non-Orientable).}
A more complex rational ansatz generates the Henneberg surface \cite{henneberg1875}, which is a non-orientable minimal surface (a realization of the projective plane in 3D space). It arises from a rational map of degree 4:
\begin{smalleq}
\begin{align}
    \lambda(z) = \begin{pmatrix} 1-z^{-2} \\ z - z^{-1} \end{pmatrix}, \quad \mu(z) = \begin{pmatrix} 1 \\ -1 \end{pmatrix}
\end{align}
\end{smalleq}
\textbf{4. The Helicoid (Transcendental Solution).}
The Helicoid is the locally isometric conjugate surface to the Catenoid. In the spinor formalism, it arises from the same meromorphic data but with a phase shift that exposes the period of the complex logarithm.
\begin{smalleq}
\begin{align}
    \lambda(z) = e^{\I\pi/4} \begin{pmatrix} z^{-1/2} \\ z^{1/2} \end{pmatrix}, \quad \mu(z) = e^{\I\pi/4} \begin{pmatrix} z^{-1/2} \\ -z^{1/2} \end{pmatrix}
\end{align}
\end{smalleq}
The resulting null vector possesses a pole with purely imaginary residue, $f'(z) \sim \I/z$. Integration yields a logarithmic branch cut $f(z) \sim \I \ln z$, corresponding to the multi-valued height function of the double helix ($X_3 \propto \Im(\ln z) = \theta$). 

\section{Explicit invariant tensor solutions for the MLE}
\label{sec:appC}
\begin{table*}[t]
\centering
\renewcommand{\arraystretch}{1.2} % Adds a tiny bit of vertical padding so the math breathes
\begin{tabular}{|c|c|c|c|c|c|c|}
\hline
Order $n$ & $N_{eq}(n)$ & $N_{unk}(n)$ & $\text{rank}(A_n)$ & $\dim \text{LeftNull}(A_n)$ & $\Delta\text{rank}$ & Solvable? \\
\hline
4 & 4 & 6 & 4 & 0 & 0 & yes \\
6 & 31 & 30 & 26 & 5 & 0 & yes \\
8 & 379 & 212 & 82 & 297 & 1 & no \\
\hline
\end{tabular}
\caption{Linear algebra diagnostics for the exact Vector Momentum Loop Equation. Here $\Delta\text{rank} := \text{rank}([A_n|b_n]) - \text{rank}(A_n)$. A strict positivity, $\Delta\text{rank} > 0$,  proves the system is mathematically inconsistent.}
\label{tab:rank_diagnostics}
\end{table*}
In this appendix, we provide the explicit polynomial solutions for the finite functional Taylor--Magnus expansion of the Momentum Loop Equation, valid up to the non-analytic boundary at $\mathcal{W}^{(8)}$.
\subsection{Solving for the lowest terms $W^{(n <8)}$}
The vacuum is normalized to $\mathcal{W}^{(0)} = 1$. With the topological closure of the boundary gaps enforced via the subtraction of interleaved permutations, all odd traces identically vanish, $\mathcal{W}^{(2k-1)}=0$. 

The fully symmetric even-parity invariant tensors are constructed from all possible pair-wise contractions of the Kronecker delta $\delta_{\mu\nu}$. The lowest-order non-trivial tensors, parameterized by the unconstrained fundamental mass scales, evaluate to:
\begin{align}
\mathcal{W}^{(2)}_{\mu_1 \mu_2} &= \delta_{\mu_1 \mu_2} \br
\mathcal{W}^{(4)}_{\mu_1 \mu_2 \mu_3 \mu_4} &= \frac{1}{6} \delta_{\mu_1\mu_3}\delta_{\mu_2\mu_4} \br
&+ c_{4,1} \lrb{\delta_{\mu_1\mu_2}\delta_{\mu_3\mu_4} + \delta_{\mu_1\mu_4}\delta_{\mu_2\mu_3}}
\end{align}
The tensor $\mathcal{W}^{(6)}_{\mu_1 \dots \mu_6}$ is completely determined by the geometric stress cascading from the left-hand side loop derivative acting on the lower-order terms,  absorbing the kinematic constraints without contradiction. Its explicit expansion in terms of the 15 independent 6-point Kronecker pairings is uniquely fixed by the lower-order parameters $c_{2,1}\equiv 1$, $c_{4,1}$, and $c_{4,2}$. 

At $\mathcal{W}^{(8)}$, this linear absorption fails. The expansion truncates due to the irreducible non-linear $\mathcal{W}^{(4)} \times \mathcal{W}^{(4)}$ stress, which exceeds the tensor rank of the pure cyclic Shuffle Ideal basis,  proving the non-analyticity of the functional $\mathcal{W}[P]$.

We compute these terms in \cite{MBMLEAlgebraic}, where we prove lack of solution for $\mathcal{W}^{(8)}$.
\subsection{Counting of unknowns/equations and Fredholm inconsistency at $\mathcal{W}^{(8)}$}

To  demonstrate the breakdown of the Taylor-Magnus expansion at the 8th order ($\mathcal{W}^{(8)}$), we cast the Momentum Loop Equation at each order $n$ into a finite linear system. After imposing the closed-loop shuffle-ideal reduction (loop-closing prescription), cyclic symmetry, and parity constraints (all odd orders vanish), the  functional matching reduces to:
\EQ{
A_n x_n = b_n
}
where the vector of unknowns $x_n$ is explicitly decomposed as:
\EQ{
x_n = (c^{(n)}; \; \alpha^{(<n)}; \; g^{(n)})
}
Here, $c^{(n)}$ are the coefficients of the invariant-tensor basis of $\mathcal{W}^{(n)}$, $\alpha^{(<n)}$ are the residual free parameters surviving from lower-order solutions (e.g., those left unconstrained in $\mathcal{W}^{(4)}$ and $\mathcal{W}^{(6)}$), and $g^{(n)}$ are the gauge parameters associated with solving modulo the closed-loop shuffle ideal. 

The total number of unknown variables solved for at order $n$ is therefore $N_{unk}(n) = N_{\mathcal{W}}(n) + N_{free}(<n) + N_{gauge}(n)$. We analyzed the exact linear algebra of these systems using a specialized symbolic script. The  matrix ranks and subsystem counts are summarized in Table \ref{tab:rank_diagnostics}.

The finite-dimensional system $A_n x_n = b_n$ is solvable if and only if it satisfies the Fredholm alternative consistency condition:
\EQ{
u^T b_n = 0 \quad \forall u \in \text{LeftNull}(A_n)
}
At $n=8$, the matrix $A_8$ possesses 297 left-null vectors. Projecting the inhomogeneous right-hand side $b_8$ (the geometric stress generated by the discrete non-linear cross-terms of $\mathcal{W}^{(4)} \times \mathcal{W}^{(4)}$) onto these null directions yields explicit, non-zero rational obstructions. For example, selecting the first canonical left-null basis vector $u_1$ yields the  rational projection:
\EQ{
u_1^T b_8 = \frac{1}{36} \neq 0
}
This establishes a  mathematical contradiction. Two critical robustness statements follow directly from this matrix representation:
\begin{itemize}
    \item Because the vector of unknowns $x_8$ explicitly includes all residual lower-order parameters $\alpha^{(<8)}$, this inconsistency proves that \textit{no possible choice} of lower-order free parameters can restore solvability.
    \item Setting the physical scale parameter to zero does not alleviate the structural rank deficiency.
\end{itemize}
The full basis construction, shuffle-ideal reduction, and  linear-algebra verification---including the generation of the matrices $A_n$ and the explicit null-vector projections---are provided in the executable Mathematica notebook \cite{MBMLEAlgebraic} (see Data Availability).

\section{The Pinch Singularity and the Topological Barrier}
\label{app:pinch}
To mathematically illustrate the appearance of the conjugate pole at $1/\bar{z}_0$ and the resulting impenetrable topological barrier at the unit circle, consider the evaluation of the Liouville kinetic anomaly term in the effective action. In the case of a meromorphic twistor field with one  pole $z_0$, the conformal factor generates an area integral of the form:
\EQ{
I(z_0) = \int_D d^2z \frac{1}{(z-z_0)(\bar{z}-\bar{z}_0)} 
}
In order to continue this integral as an analytic function of $z_0$ we rewrite the integrand as a total derivative with respect to $\bar{z}$:
\EQ{
I(z_0)= \frac{-\I}{2}\int_{z \bar z<1} dz d\bar{z} \, \partial_{\bar{z}}\lrb{\frac{\log(\bar{z}-\bar{z}_0)}{z-z_0}}
}
By applying Stokes' theorem, the two-dimensional bulk integral over the disk $D$ reduces to a one-dimensional contour integral over the boundary unit circle $z \bar z=1$:
\EQ{
I(z_0) =\frac{-\I}{2} \int_{z \bar z=1} dz \frac{\log(\bar{z}-\bar{z}_0)}{z-z_0}
}
Crucially, on the unit circle, the complex conjugate is  given by the geometric inversion $\bar{z} = 1/z$. Substituting this into the logarithm yields the  boundary representation of the action:
\EQ{
I(z_0) =\frac{-\I}{2} \oint_{|z|=1} dz \frac{\log(1/z - \bar{z}_0)}{z-z_0}
}
This  formula explicitly reveals the analytic structure governing the boundary dynamics. The integrand possesses a simple pole at $z = z_0$ (originating from the twistor itself) and a logarithmic branch point at $z = 1/\bar{z}_0$ (originating from the boundary inversion of the conjugate twistor). 
Because the physical domain restricts the pole to the interior of the disk ($|z_0| < 1$), its conjugate mirror  resides outside the disk ($|1/\bar{z}_0| > 1$). However, if any continuous Langevin fluctuation attempts to push the twistor pole across the boundary ($|z_0| \to 1$), its conjugate mirror must simultaneously approach the boundary from the outside. 

At the  moment of crossing, these two singularities  collide directly on the physical integration contour $|z|=1$, trapping the integration path between them. This creates a severe, non-integrable pinch singularity, causing the action integral to macroscopically diverge. This geometric divergence acts as an infinite energetic barrier, protecting the twistor pole from escaping and locking the string into its quantized topological sector.
\section{Analytical and numerical study of the Twistor String equation}
\label{sec:twistorSpectral}

The physical states of the twistor string are governed by the total action monodromy $\Delta S$ accumulated over one period of rotation. The Bohr-Sommerfeld quantization condition requires $\Delta S = 2\pi n$. As derived in Section 16, factoring out the overall twistor scale parameter $a$, we can write $\Delta S = a \Phi$, where the phase functional is
\EQ{
\Phi(K, \beta) &= 2\pi E K - 4\pi \left(J + \frac{q}{2}\right) - \pi \sigma K^2 \left(\beta + \frac{\sin 2\beta}{2}\right) \br
&- 4\pi m K \cos\beta - \frac{2}{3}(\tan\beta - \beta).
}
Here, $K = R/a$ is the kinematic radius, and $q$ represents the topological vector index ($q=0$ for the pseudo-scalar families, and $q=-1$ for the vector families).

\subsection{Significant reduction of the spectrum: only the $n=0$ level is left}

The classical string trajectory is determined by the saddle-point equations with respect to the dynamical variables $a, K$, and $\beta$. 
Setting $\partial_a \Delta S = 0$, we get the first equation:
\EQ{
\label{Phi0}
\Phi(K, \beta) = 0,
}
which implies that there are \textbf{no levels with $n \neq 0$} in the general Bohr-Sommerfeld quantization condition:
\EQ{
\Delta S = 2 \pi n, \quad n = 0, \pm 1, \pm 2, \dots
}
This important conclusion follows from the fact that our $\Delta S$ depends linearly on $a$ at fixed $\beta$ and $K = R/a$. Therefore, by changing our variables to $a, K, \beta$, we completely eliminate $a$ from the saddle-point equations and also obtain the extra equation \eqref{Phi0} for $\beta$ and $K$. 

Next, setting $\partial_K \Phi = 0$ yields the energy equation:
\EQ{
E = \sigma K \left(\beta + \frac{\sin 2\beta}{2}\right) + 2m \cos\beta. \label{eq:E_param}
}
Setting the boundary variation $\partial_\beta \Phi = 0$ balances the boundary tensions and yields a quadratic equation for $K$:
\EQ{
\sigma \cos^2\beta \, K^2 - 2 m \sin\beta \, K + \frac{1}{3\pi} \tan^2\beta = 0.
}
Note that we have obtained three equations for three unknowns $E, K, \beta$, with $a$ decoupled. Thus, we found the two equations for the two parameters of the saddle point and the third equation for the energy spectrum, just as intended, with the important restriction that all the extra Bohr-Sommerfeld levels $n \neq 0$ are absent.

\subsection{Algebraic simplifications of the spectral equation}

To systematically analyze and numerically fit these equations, it is advantageous to cast them into a dimensionless form. We introduce the dimensionless renormalized boundary mass parameters $x=m/\sqrt{\sigma}$ for the light quarks and $x^{\prime}=m^{\prime}/\sqrt{\sigma}$ for the strange trajectories, where $m^{\prime} = (m+m_s)/2$. Solving the quadratic equation for the radius, we find the  analytical root:
\EQ{
\mathcal K(\beta, x) \equiv \sqrt{\sigma} K(\beta) = \frac{\sin \beta}{\cos^2 \beta} \left( x + \sqrt{x^2 - \frac{1}{3\pi}} \right).
}
For a physical string to exist, the classical radius $K$ must be real, which requires the discriminant to be non-negative. Thus, the Minkowski version of the Liouville anomaly imposes a rigid chiral existence bound on the dimensionless mass:
\EQ{
x \ge \frac{1}{\sqrt{3\pi}}.
}

Substituting $E$ from Eq.~\eqref{eq:E_param} back into $\Phi = 0$, the mass boundary terms cancel, yielding the parametric angular momentum:
\EQ{
J(\beta) = \frac{\sigma K^2}{4} \left(\beta + \frac{\sin 2\beta}{2}\right) - \frac{1}{6\pi} (\tan\beta - \beta) - \frac{q}{2}. \label{eq:J_param}
}

To fit the unknown fundamental parameters---\textbf{the dimensionless constituent quark masses ($x_u, x_s, x_c, x_b$)} and the universal string tension---to the experimental Particle Data Group (PDG) data, we utilize the action monodromy target $\Phi = 0$. Rather than performing an arbitrary least-squares fit on $M^2$ or $J$ independently, we mathematically minimize the fundamental physical constraint of the theory itself (see the remark \ref{quarkMass} for the definition of our effective quark mass).

For any given observed resonance with experimental mass $M$ and  spin $J$, we first invert Eq.~\eqref{eq:J_param} to find the unique theoretical angle $\beta(J)$ such that $J_{\rm param}(\beta) = J$. Substituting this angle back into the master phase equation, the terms corresponding to angular momentum  cancel out. Since the theoretical parameters natively satisfy $\Phi = 0$, the complicated area, boundary mass, and Liouville anomaly terms  collapse, leaving the  geometric phase mismatch:
\EQ{
\Phi(M, J) = 2\pi K(\beta) \left( M - E_{\rm param}(\beta) \right).
}
This target function has  physical and mathematical meaning. Because $K \propto 1/\sqrt{\sigma}$, the target $\Phi$ is naturally dimensionless. At high energies, where $K \propto E/\sigma$, the action monodromy asymptotically scales as $\Phi \propto M \Delta M / \sigma \propto \Delta(M^2) / (2\sigma)$, intrinsically generating the correct mass-squared weighting expected for Regge trajectories without imposing it ad hoc. Crucially, for the non-rotating ($J=0$) pseudo-Goldstone bosons ($\pi$ and $K$), the string shrinks to a point ($\beta \to 0$ and $K \to 0$), causing $\Phi$ to identically vanish. This naturally decouples the anomalously light chiral masses from the macroscopic string tension.

The global \textbf{5D} parameter optimization was performed by minimizing the sum of the squared geometric action residuals,
\EQ{
    \chi^2(m_u, m_s, m_c, m_b, \sigma) = \sum_{q_1, q_2}\sum_{J,M} \left( \Phi(M, J, x_{q_1 q_2}, \beta(J)) \right)^2,
\label{E.10}
}
across all \textbf{thirty-six} on-shell PDG states spanning the \textbf{thirteen} meson families, where $x_{q_1 q_2} = (x_{q_1} + x_{q_2})/2$ is the dynamically constructed average boundary mass for the given state.
\end{document}